\documentclass{article}

\usepackage{arxiv}

\usepackage[utf8]{inputenc} 
\usepackage[T1]{fontenc}    
\usepackage[hidelinks]{hyperref}       
\usepackage{url}            
\usepackage{booktabs}       
\usepackage{microtype}      
\usepackage{amsmath, amsfonts, amssymb, amsthm}
\usepackage{shortcuts}
\usepackage{mathtools}
\usepackage{bm}
\usepackage{xcolor}

\usepackage{float}

\usepackage[ruled,vlined]{algorithm2e} 

\newtheorem{proposition}[]{Proposition}
\newtheorem{assumption}[]{Assumption}
\newtheorem{theorem}[]{Theorem}
\newtheorem{corollary}[]{Corollary}
\newtheorem{remark}[]{Remark}
\newtheorem{lemma}[]{Lemma}

\title{No-harm calibration for generalized Oaxaca-Blinder estimators}

\author{
    Peter L. Cohen \\
  Operations Research Center\\
  Massachusetts Institute of Technology\\
  1 Amherst Street\\
  Cambridge, Massachusetts 02142, U.S.A\\
  \texttt{plcohen@mit.edu} \\
  \And
    Colin B. Fogarty \\
  Operations Research and Statistics Group\\
  Massachusetts Institute of Technology\\
  100 Main Street\\
  Cambridge, Massachusetts 02142,  U.S.A\\
  \texttt{cfogarty@mit.edu} 
}

\usepackage{multibib}
\newcites{refs}{References}

\begin{document}

\maketitle

\begin{abstract}
    In randomized experiments, adjusting for observed features when estimating treatment effects has been proposed as a way to improve asymptotic efficiency.  However, only linear regression has been proven to form an estimate of the average treatment effect that is asymptotically no less efficient than the treated-minus-control difference in means regardless of the true data generating process.  Randomized treatment assignment provides this ``do-no-harm" property, with neither truth of a linear model nor a generative model for the outcomes being required.  We present a general calibration method which confers the same no-harm property onto estimators leveraging a broad class of nonlinear models. This recovers the usual regression-adjusted estimator when ordinary least squares is used, and further provides non-inferior treatment effect estimators using methods such as logistic and Poisson regression. The resulting estimators are non-inferior to both the difference in means estimator and to treatment effect estimators that have not undergone calibration. We show that our estimator is asymptotically equivalent to an inverse probability weighted estimator using a logit link with predicted potential outcomes as covariates. In a simulation study, we demonstrate that common nonlinear estimators without our calibration procedure may perform markedly worse than both the calibrated estimator and the unadjusted difference in means.
\end{abstract}



\section{Introduction}
In completely randomized experiments, \citetrefs{lin13} demonstrated that linear regression employing treatment-by-covariate interactions can be used to estimate the sample average treatment effect while adjusting for baseline features.  The orthogonalities arising in the geometry of linear regression in concert with the act of randomization yield a ``do-no-harm" property for the resulting estimator: assuming neither the existence of a true linear model nor a generative model for the outcome variables, the regression-adjusted estimator's asymptotic variance is never larger than that of the usual difference in means estimator.  While the resulting estimator is asymptotically no less efficient than the difference in means regardless of the data generating process, linear regression seems ill-suited for modeling phenomena such as binary or count data.  In such contexts, leveraging a nonlinear model such as logistic or Poisson regression may seem more natural. 

Lin's (2013) regression-adjusted estimator can be viewed as an imputation estimator: the practitioner uses linear regressions of outcomes on covariates to impute counterfactual outcomes, and then takes the difference in means between the imputed populations as their estimate of the treatment effect.  Building upon the work of \citetrefs{OaxacaOriginal} and \citetrefs{BlinderOriginal} among others, \citetrefs{generalizedOB} present a general theory for leveraging ``simple'' nonlinear models to impute missing potential outcomes, providing conditions for consistency and asymptotic normality for the resulting treatment effect estimators, coined \textit{generalized Oaxaca-Blinder estimators}. That said,  \citetrefs{generalizedOB} were not able to establish a non-inferiority property for these nonlinear estimators relative to the difference in means estimator, raising the concern that imputation with more general prediction functions could degrade inference relative to no adjustment whatsoever.  

A related class of treatment effect estimators are \textit{model standardization estimators}, which take the difference in the averages of the predicted values for the potential outcomes under treatment and control as an effect estimate. These are equivalent to Oaxaca-Blinder estimators when the average of the fitted values for those receiving treatment and control equal the average of the observed outcomes for those individuals; this holds automatically for generalized linear models. When using nonlinear adjustment, standardized estimators have been shown to have non-inferior asymptotic efficiency relative to the difference in means estimator under a superpopulation model when assuming correct specification of the conditional mean function; see for instance \citetrefs{rosenblumVDL} or \citetrefs[Theorem 7.1]{negiWooldridge}. Unfortunately, this class of estimators does not generally provide non-inferior treatment effect estimates under misspecification. \citeauthor{negiWooldridge} remark that ``we do not have theoretical results to show when the nonlinear [regression adjustment] methods unambiguously improve asymptotic efficiency in case of misspecification'' \citeprefs[p. 526]{negiWooldridge}. \citetrefs{leiDing} generalized linear model standardization results to high-dimensional data; they retain the asymptotic non-inferiority of \citetrefs{lin13} but their proofs rest upon the linearity of the prediction functions.  This leaves a major gap between linear and nonlinear regression adjustment in randomized experiments.  \looseness=-1

We provide a calibration procedure that confers non-inferiority after nonlinear regression adjustment under both the finite-population and superpopulation framework for causal inference. To the best of our knowledge this is the first procedure for conferring non-inferiority to generalized Oaxaca-Blinder estimators and model standardization estimators while remaining agnostic to the truth of an underlying nonlinear model. Our procedure simply feeds the predicted values for the potential outcomes under both treatment and control from a nonlinear model as covariates from which to form the linear regression-adjusted estimator of \citetrefs{lin13}, and provides the same non-inferiority guarantees under conditions on the prediction functions outlined in this work. We show through simulation that without this calibration step generalized Oaxaca-Blinder estimators can perform markedly worse than both the calibrated estimator and the unadjusted difference in means. This leads us to strongly recommend the use of our procedure in providing the natural extension of adjustment in randomized experiments from linear to nonlinear models. Not only are calibrated estimators non-inferior to both the difference in means estimator and the uncalibrated estimator, but also without calibration generalized Oaxaca-Blinder estimators can perform worse than the difference in means even when using simple, commonly deployed nonlinear models. We further discuss how calibration may be used in concert with adjustment strategies leveraging flexible nonlinear methods without corrupting desirable properties such as semiparametric efficiency.

\section{Notation and review}
\subsection{Notation for completely randomized designs}
We begin under the finite-population approach to causal inference where randomized treatment allocation alone justifies our results without assuming a generative model for outcomes or features; as will be discussed, our findings also hold under common superpopulation formulations for inference on both the population average treatment effect and the conditional average treatment effect.  An experimental population is comprised of $N$ units.  For the $i$th unit there are two scalar potential outcomes: what would have been observed under control, $y_{i}(0);$ and what would have been observed under treatment, $y_{i}(1)$.  Randomness only enters the experiment through the allocation of treatment over the population of $N$ individuals, $n_0$ of whom receive the control and $n_1$ of whom receive the treatment.  Let $Z_{i}$ denote the treatment indicator of the $i$th unit: $Z_{i} = 1$ if the $i$th unit receives treatment otherwise $Z_{i} = 0$.  In a completely randomized experiment, the vector $Z = (Z_{1}, \ldots, Z_{N})^{\T}$ is distributed uniformly over $\Omega_{CRE} = \left\{z \in \{0, 1\}^{N} \;:\; \sum_{i = 1}^{N}z_{i} = n_{1} \right\}$. Asymptotics are taken with respect to a sequence of finite experimental populations.  For each $N$ the characteristics of the finite population may change, as may the ratio of treated to control units; this dependence is generally suppressed in the notation that follows. We assume that $n_1/N\rightarrow p$ satisfying $0 < p < 1$ as $N\rightarrow \infty$.

An experimenter draws a treatment allocation $Z$ uniformly from $\Omega_{CRE}$. Under the stable-unit treatment value assumption \citeprefs{rub80}, she then observes $y_{1}(Z_{1}), \ldots, y_{N}(Z_{N})$, the potential outcomes corresponding to the observed treatment assignment.  The sample average treatment effect for the $N$ units is $\SATE = N^{-1}\sum_{i = 1}^{N}\left\{y_{i}(1) - y_{i}(0)\right\}$, and is unknown because $y_i(0)$ and $y_i(1)$ cannot be jointly observed for any unit $i$. The conventional estimator for $\SATE$ is $\tauhat_{unadj} = n_{1}^{-1}\sum_{i = 1}^{N}Z_{i} y_{i}(Z_{i}) -  n_{0}^{-1}\sum_{i = 1}^{N}(1 - Z_{i}) y_{i}(Z_{i})$, the treated-minus-control difference in means \citeprefs{ney23}. Under mild regularity conditions on the sequence of finite populations, $N^{1/2}\left( \tauhat_{unadj} - \SATE \right)$ obeys a central limit theorem \citeprefs{FiniteCLT}.

\subsection{The generalized Oaxaca-Blinder estimator}
Oftentimes baseline covariates $x_i\in \mathbb{R}^k$ are collected for each unit in the study. While $\tauhat_{unadj}$ is unbiased without adjustment its variance may be inflated due to post-randomization imbalances on covariates predictive of the outcome. Imputation estimators use covariate information to impute the counterfactual $y_{i}(1 - Z_{i})$ with the objective of reducing estimator variance. Suppose one trains predictors for the outcomes under control and treatment, $\muhat_{0}(x_i)$ and $\muhat_{1}(x_i)$ respectively. Following \citetrefs{generalizedOB}, define imputed outcomes $\hat{y}_{i}(z)$ for $z = 0,1$ and the resulting generalized Oaxaca-Blinder estimator as
\begin{equation}\label{eqn: imputation estimator}
    \tauhat_{gOB} = N^{-1}\sum_{i = 1}^{N}\left\{\hat{y}_{i}(1) - \hat{y}_{i}(0) \right\}; \quad \hat{y}_{i}(z) = \begin{cases}
                        y_{i}(Z_{i}) &\text{if } Z_{i} = z\\
                        \muhat_{z}\left(x_{i}\right) &\text{if } Z_{i} \neq z
    \end{cases}.
\end{equation} 

Both the difference-in-means estimator and Lin's (2013) regression-adjusted estimator arise from particular choices for $\hat{\mu}_1$ and $\hat{\mu}_0$. \citetrefs{generalizedOB} present a set of sufficient conditions which facilitate analysis of the limiting distribution for $N^{1/2}(\hat{\tau}_{gOB} - \bar{\tau})$.  We generalize their conditions slightly. Two important assumptions, ``stability'' and ``vanishing error processes'', allow for asymptotic reformulations of the estimators in terms of certain residuals.  A third, ``prediction unbiasedness,'' ensures robustness of Oaxaca-Blinder estimators to model misspecification.

\begin{assumption}[Stability]\label{asm: stability}
    For $z=0,1$, there exists a deterministic sequence of functions $\{\overdotmu_{z}^{(N)}\}_{N \in \N}$ such that
$\left|\left|\muhat_{z} - \overdotmu_{z}^{(N)}\right|\right|_{N} := \left\{N^{-1}\sum_{i = 1}^{N}||\overdotmu_{z}^{(N)}(x_{i}) - \muhat_{z}(x_{i})||^{2} \right\}^{1/2} = o_{p}(1).$
    For notational simplicity, we generally drop the superscripted index and write $\overdotmu_{z}$.  
\end{assumption}

\begin{assumption}[Vanishing Error Process]\label{asm: error process vanishes}
    For a function $f: \R^{k} \rightarrow \R$ define\footnote{We use the notation $\G_{N, z}$ to follow that of \citetrefs{generalizedOB}, but the stochastic process $\{\G_{N, z}(\cdot)\}$ indexed by functions $f$ has been implicitly examined elsewhere in the literature, e.g., \citetrefs{rothe}.} ${\G_{N, z}(f) = N^{-1/2}\sum_{i = 1}^{N}\left(\frac{\indicatorFunction{Z_{i} = z}f(x_{i})}{n_{z}/N} - f(x_{i}) \right)}.$
    Assume that, for $z \in \{0, 1\}$, the error stochastic process $\left|\G_{N, z}(\overdotmu_{z}) -  \G_{N, z}(\muhat_{z})\right|$ vanishes in probability; formally $\left|\G_{N, z}(\overdotmu_{z}) -  \G_{N, z}(\muhat_{z})\right| = o_{P}\left(1\right)$.
\end{assumption}




\begin{assumption}[Prediction Unbiasedness]\label{asm: prediction unbiasedness}
    $\sum_{i:Z_i=z}\hat{\mu}_z(x_i)= \sum_{i: Z_i=z}y_i(Z_i)$ for $z=0,1$.
\end{assumption}

If $\muhat_{z}$ is the solution to some empirical risk minimization procedure, then a natural candidate for $\overdotmu_{z}$ is the population-level risk minimizer and Assumption~\ref{asm: stability} reflects the standard goal that the empirical risk minimizers approximate the population risk minimizers as the sample size grows.  The functions $\overdotmu_{0}$ and $\overdotmu_{1}$ need not reflect any true relationship between outcomes and covariates.  
Assumption~\ref{asm: error process vanishes} is quite general; in the supplementary material we provide concrete sufficient conditions based upon an entropy bound of \citetrefs{maximalInequalityVDVWellner} or cross-fitting.  Assumption \ref{asm: prediction unbiasedness} holds for many choices of nonlinear models, including generalized linear models. Under Assumption~\ref{asm: prediction unbiasedness} the estimator may be written as a model standardization estimator, with $\hat{\tau}_{gOB} = N^{-1}\sum_{i=1}^N\{\hat{\mu}_1(x_i) - \hat{\mu}_0(x_i)\}$.

Assumptions \ref{asm: stability}, \ref{asm: error process vanishes}, and \ref{asm: prediction unbiasedness} are sufficient to establish consistency and asymptotically linear representations for generalized Oaxaca-Blinder estimators \citeprefs[Theorems 2-3]{generalizedOB}. Further regularity conditions are required to imply asymptotic normality \citeprefs[Corollary 1]{generalizedOB}. Based upon the assumptions of \citetrefs{lin13} and \citetrefs{freedmanRegAdj1} we assume the following about the potential outcomes and the functions $\dot{\mu}_z$:
\begin{assumption}[Limiting Means and Variances]\label{asm: means and covs stabilize}
	 The mean vector and covariance matrix of $(y_i(0), y_i(1), \dot{\mu}_0(x_i), \dot{\mu}_0(x_i))^\T$ have limiting values. For instance, for $z=0,1$ there exists a limiting value $\bar{y}(z)_\infty$ such that $\lim_{N\rightarrow \infty} N^{-1}\sum_{i=1}^Ny_i(z) = \bar{y}(z)_\infty$.
\end{assumption}
\begin{assumption}[Bounded Fourth Moments]\label{asm: bounded fourth moment}
 There exists some $C < \infty$ for which, for all ${z = 0,1}$ and all $N$,
		$N^{-1}\sum_{i = 1}^{N}\{y_{i}(z)\}^{4} < C$ and $N^{-1}\sum_{i = 1}^{N}\left\{\dot{\mu}_z(x_i)\right\}^{4} < C.$
\end{assumption}

\section{Linear Calibration}
Assumptions \ref{asm: stability} - \ref{asm: bounded fourth moment} are not sufficient for $\hat{\tau}_{gOB}$ to be non-inferior to the unadjusted difference in means estimator; see Section \ref{sec: simulations} for an illustration with Poisson regression. We now describe a simple transformation of $\hat{\mu}_z$ that provides a ``do-no-harm" property after nonlinear adjustment. For each unit $i$, create the pseudo-feature vector $\Tilde{x}_{i} = \left(\muhat_{0}(x_{i}), \muhat_{1}(x_{i})\right)^{\T}$ containing the predicted potential outcomes under control and treatment. Then, for $z = 0,1$, define $\muhat_{OLS, z}(x_i)$ as the prediction equation from a least squares regression of $y_{i}(Z_i)$ on $\tilde{x}_i$ along with an intercept for those units $i$ such that $Z_{i} = z$, yielding for $z=0,1$
\begin{align}\label{eq:lincal}
    \hat{\mu}_{OLS,z}(x_i) &= \hat{\alpha}_{z} + \hat{\beta}_{z,0}\hat{\mu}_0(x_i) + \hat{\beta}_{z,1}\hat{\mu}_1(x_i);\\ (\hat{\alpha}_z, \hat{\beta}_{z,0}, \hat{\beta}_{z,1})^\T& \in \underset{(\alpha_z, \beta_{z,0}, \beta_{z,1})^\T}{\arg\min}\;\; \sum_{i: Z_i = z}\{y_i(z) - {\alpha}_{z} - {\beta}_{z,0}\hat{\mu}_0(x_i) - {\beta}_{z,1}\hat{\mu}_1(x_i)\}^2.\nonumber 
\end{align} 
Finally,  form the treatment effect estimator (\ref{eqn: imputation estimator})  using $\muhat_{OLS, 0}$ and $\muhat_{OLS, 1}$. Equivalently, simply calculate Lin's (2013) regression-adjusted estimator with the features $\tilde{x}_i = \left(\muhat_{0}(x_{i}), \muhat_{1}(x_{i})\right)^{\T}$.

We call the resulting estimator the linearly-calibrated Oaxaca-Blinder estimator, denoted by $\tauhat_{cal}$; see the supplementary material for pseudocode. The approach is similar to that of \citetrefs[Equation 8]{generalizedOB}, but importantly differs in that under their approach $\hat{\mu}_{1-z}(x_i)$ is not included as a predictor variable in the regressions for individuals with $Z_i=z$. By including both prediction functions in $\tilde{x}_i$, $\tauhat_{cal}$ attains non-inferiority.

\begin{theorem}\label{thm: reg adj does no harm}
    Suppose that Assumptions \ref{asm: stability}, \ref{asm: error process vanishes}, \ref{asm: means and covs stabilize}, and \ref{asm: bounded fourth moment} hold. Then, $N^{1/2}\left(\tauhat_{cal} - \SATE\right)$ converges in distribution to a mean-zero Gaussian random variable.  Furthermore the asymptotic variance of $N^{1/2}\left(\tauhat_{cal} - \SATE\right)$ is no larger than that of  $N^{1/2}\left(\tauhat_{unadj} - \SATE\right)$.
\end{theorem}

Theorem~\ref{thm: reg adj does no harm} does not require knowledge of the true relationship between outcomes and covariates, and the models $\muhat_{0}$ and $\muhat_{1}$ can be arbitrarily misspecified. Assumption \ref{asm: prediction unbiasedness} is not required for $\hat{\mu}_0$ and $\hat{\mu}_1$ because prediction unbiasedness always holds after applying our procedure due to the inclusion of the intercept terms. The non-inferiority statement in Theorem \ref{thm: reg adj does no harm} is proven in a manner similar to Corollary 1.1 of \citetrefs{lin13};  however, care must be taken to account for randomness in the incorporation of the derived covariates $\tilde{x}_{i} = (\muhat_{0}\left(x_{i}\right), \muhat_{1}\left(x_{i}\right))^\T$. The proof further demonstrates that the sufficient conditions for asymptotic Gaussianity of $N^{1/2}(\tauhat_{gOB}- \SATE)$ provided by Assumptions \ref{asm: stability} - \ref{asm: bounded fourth moment} are also sufficient for the asymptotic Gaussianity of $N^{1/2}(\tauhat_{cal}- \SATE)$.  That is, under these assumptions non-inferiority can be achieved ``for free" through our calibration step. A similar argument yields the following comparisons between $\hat{\tau}_{cal}$ and both the non-calibrated estimator $\hat{\tau}_{gOB}$ under Assumption \ref{asm: prediction unbiasedness} and the singly-calibrated estimator suggested in \citetrefs[Equation 8]{generalizedOB}, denoted $\hat{\tau}_{GBcal}$.
\begin{theorem}\label{thm: calibration beats the gOB}
    Under the assumptions of Theorem~\ref{thm: reg adj does no harm} and for given estimators $\hat{\mu}_0$ and $\hat{\mu}_1$, the linearly-calibrated estimator $N^{1/2}\left(\tauhat_{cal} - \SATE\right)$ has an asymptotic variance that is no larger than that of $N^{1/2}\left(\tauhat_{GBcal} - \SATE\right)$. Further enforcing Assumption \ref{asm: prediction unbiasedness},  $N^{1/2}\left(\tauhat_{cal} - \SATE\right)$ has an asymptotic variance that is no larger than that of $N^{1/2}\left(\tauhat_{gOB} - \SATE\right)$.
\end{theorem}

\section{Further insight into linear calibration}
As the calibration step may appear unusual, it is first worthwhile to consider what occurs when $\hat{\mu}_1$ and $\hat{\mu}_0$ are fit by separate ordinary least squares regressions with intercepts in the treated and control groups.  In this case, $\tauhat_{cal}$, $\hat{\tau}_{gOB}$, and $\hat{\tau}_{GBcal}$ are identical to Lin's (2013) estimator as the resulting prediction equations remain linear in the covariates themselves; see the appendices for a formal proof. Including $\hat{\mu}_0(x_i)$ and $\hat{\mu}_1(x_i)$ as predictors when calculating  $\hat{\tau}_{cal}$ also yields an insightful asymptotic equivalence with an inverse probability weighted (IPW) estimator. Suppose that despite having run a randomized experiment with known assignment probabilities, one fits a logistic regression model for the probability that $Z_i=1$ with $\Tilde{x}_i = (\muhat_{0}\left(x_{i}\right), \muhat_{1}\left(x_{i}\right))^\T$ as covariates along with an intercept. Call the resulting predicted probabilities $\hat{\pi}(\Tilde{x}_i)$. Consider the IPW estimator
$\hat{\tau}_{ipw} = N^{-1}\sum_{i=1}^N Z_iy_i(Z_i)/{\hat{\pi}(\Tilde{x}_i)} - N^{-1}\sum_{i=1}^N(1-Z_i)y_i(Z_i)/\{1-\hat{\pi}(\tilde{x}_i)\}$.
  
\begin{theorem}\label{thm: IPW}
Under the assumptions of Theorem \ref{thm: reg adj does no harm} $N^{1/2}(\hat{\tau}_{cal} - \hat{\tau}_{ipw}) = o_p(1)$. That is, the two estimators are asymptotically equivalent.
\end{theorem}

In light of the proof of Theorem \ref{thm: reg adj does no harm}, the result follows immediately from Corollary 1 of \citetrefs{shen2014} and is omitted; see also \citetrefs{HiranoImbensRidder} for related results on IPW estimators. In randomized experiments, inverse probability weighted estimators adjust for chance imbalances on variables contained within the propensity score model by reweighting individuals using their predicted probabilities of treatment. Imbalances on covariates are problematic only in so far as the covariates are predictive of the potential outcomes. By including both $\hat{\mu}_0$ and $\hat{\mu}_1$ within the propensity score model, $\hat{\tau}_{ipw}$ adjusts for chance imbalances on the predicted values for the potential outcomes under treatment and control. Both \citetrefs{rot12} and \citetrefs{col15} use predicted potential outcomes into propensity score models to establish non-inferior treatment effect estimators under a superpopulation. By the asymptotic equivalence provided by Theorem 3, $\hat{\tau}_{cal}$ can also be viewed in this light.  Importantly, this equivalence does not generally hold for the entire class of Oaxaca-Blinder estimators $\hat{\tau}_{gOB}$. The estimator $\hat{\tau}_{GBcal}$ suggested in \citetrefs[Equation 8]{generalizedOB} is equivalent to a peculiar IPW estimator where the treated and control outcomes are weighted with estimated probabilities stemming from \textit{different} logit models, with the treated (resp. control) outcomes weighted by estimated probabilities where only fitted values under treatment (resp. control) are used as covariates.

\section{Calibration and non-inferiority under superpopulation models}

Our results have viewed the potential outcomes and covariates as fixed, with the only randomness coming from random assignment. As the recruitment process for inclusion in randomized experiments often amounts to a convenience sample, we view this as a natural framework for inference. Alternative frameworks view $(y_i(1), y_i(0), x_i)$ as independent and identically distributed draws from some distribution $P$. Still others view $x_i$ as fixed but $(y_i(1), y_i(0))$ as independently distributed from some conditional distribution $P_{y(1), y(0)\mid x}$.  The corresponding estimands are the population average treatment effect (PATE) and the conditional average treatment effect (CATE), respectively \citeprefs{imb04}.  Theorems \ref{thm: reg adj does no harm} - \ref{thm: IPW} apply with $\bar{\tau}$ replaced by $\PATE$ or $\CATE$ and with our regularity conditions reformulated depending upon the superpopulation framework. See the appendices for details.

When the triples $(y_i(1), y_i(0), x_i)$ are viewed as independent and identically distributed draws from a distribution $P$, alternative estimators leveraging nonlinear adjustment exist which either guarantee non-inferiority under model misspecification, or leverage cross-fitting to attain semiparametric efficiency bounds. Methods guarding against model misspecification in parametric models include \citetrefs{tan10}, \citetrefs{rot12} and \citetrefs{col15}. These are not imputation estimators, require explicitly fitting a separate propensity score model, and in the case of \cite{rot12} require solving a non-convex optimization problem. Our calibration step returns an imputation estimator, and it can be implemented using off-the-shelf statistical software, simply requiring an initial (potentially nonlinear) regression adjustment within each treatment group followed by a linear regression using the fitted values as covariates. Calibration may be deployed in concert with augmented inverse probability weighted (AIPW) and targeted maximum likelihood estimators (TMLE) within randomized experiments. For instance, one may replace $\hat{\mu}_z$ by $\hat{\mu}_{OLS, z}$ of (\ref{eq:lincal}) to produce an AIPW estimator which is guaranteed to be non-inferior to both the uncalibrated AIPW estimator and the unadjusted difference in means. Should the original AIPW estimator achieve the semiparametric efficiency bound, so too will the estimator after calibration. Should the uncalibrated AIPW estimator be based upon a misspecified model, calibration can provide an improvement over both the uncalibrated AIPW estimator and the unadjusted estimator. The relationship between calibration and semiparametric efficiency is explored in detail in the appendices.  Furthermore, in the appendices we discuss the use of sample splitting to acheive finite sample unbiasedness of calibrated estimators under superpopulation models, thereby extending the results of \cite{LoopEstimator} and \cite{rothe}.  We also include simulations to emphasize these points.
\looseness=-1

\section{Illustrating the improvements from linear calibration}\label{sec: simulations}

To illustrate both the benefits of linear calibration with the prediction equations for both potential outcomes and the potential peril of proceeding without our calibration step, we present a simulation study using Poisson regression. The $s$th of $S$ data sets contains $N$ individuals upon whom an experimenter performs a completely randomized experiment with $n_{1} = \lceil p N \rceil$ treated units.  In our simulations $p = 0.8$. Each unit has a scalar covariate $x_i$, generated as independent and identically distributed draws from a Uniform random variable on $[-5, 5]$. We then generate the potential outcomes under treatment and control for each individual independently as $y_i(1)\sim Poisson\{\exp(x_i)\}$ and $y_i(0)\sim Poisson\{72-0.45\exp(x_i)\}$, where $Poisson(\lambda)$ is a Poisson distribution with rate $\lambda$. The Poisson regression model is thus correctly specified for the potential outcomes under treatment, but incorrectly specified for those under control.

For each data set, we draw $B$ treatment assignment allocations.  An experimenter observes  $y_{i}(Z_{i})$ and continuous covariates $x_{i}$ for each unit.  Using the observed responses after each randomized treatment allocation, we estimate the prediction functions $\hat{\mu}_0(x_i)$ and $\hat{\mu}_1(x_i)$ via separate Poisson regressions of $y_i(Z_i)$ on $x_{i}$ in the subgroups where $Z_i=0$ and $Z_i=1$ respectively. With the outcomes and response functions in tow, we form the difference-in-means estimator $\hat{\tau}_{unadj}$, generalized Oaxaca-Blinder estimator $\hat{\tau}_{gOB}$, the singly-calibrated estimator of \citetrefs[Equation 8]{generalizedOB} $\hat{\tau}_{GBcal}$, and our linearly-calibrated estimator $\hat{\tau}_{cal}$.

Table \ref{tab: sampling variances} compares the averages (over $s=1,\ldots,S$) of the ratios of the variances for the adjusted estimators to the unadjusted estimator when setting $S=1000$, $B=1000$, and varying $N$.  Even at $N=10,000$, both $\hat{\tau}_{gOB}$ and $\hat{\tau}_{GBcal}$ have markedly larger variances than the unadjusted difference in means estimator.  Contrast this with our proposed estimator $\hat{\tau}_{cal}$, which in this simulation study provides a substantial reduction in variance relative to the difference in means, $\hat{\tau}_{gOB}$ and $\hat{\tau}_{GBcal}$. This highlights the importance of including the prediction functions for both potential outcomes in the calibration step \eqref{eq:lincal} after nonlinear adjustment.  In the supplementary material we include a simulation using logistic regression which shows the same qualitative phenomena. We also include analysis of real-world data using Poisson regression adjustment investigating the effectiveness of a chemotherapeutic agent.

\begin{table}

\def~{\hphantom{0}}
\caption{Ratios of Monte Carlo variances for $\tauhat_{cal}$, $\tauhat_{GBcal}$, and $\tauhat_{gOB}$ to the difference in means estimator $\hat{\tau}_{unadj}$ for various experiment sizes $N$.  Each variance is based upon $B = 1000$ simulated treatment allocations for a given set of potential outcomes and covariates.  Results are averaged over $S = 1000$ simulated data sets.}
\centering
\begin{tabular}{lccccc}
                & $\hat{\text{var}}(\tauhat_{gOB}) / \hat{\text{var}}(\tauhat_{unadj})$ & $\hat{\text{var}}(\tauhat_{GBcal}) / \hat{\text{var}}(\tauhat_{unadj})$  & $\hat{\text{var}}(\tauhat_{cal}) / \hat{\text{var}}(\tauhat_{unadj})$ &  &  \\
        $N = 200$   &  1.732  &  1.717  &  0.703  &  &  \\
        $N = 500$   &  1.692  &  1.685  &  0.665  &  &  \\
        $N = 1000$  &  1.675  &  1.670  &  0.659  &  &  \\
        $N = 10000$ &  1.660  &  1.657  &  0.654  &  &
\end{tabular}
\label{tab: sampling variances}
\end{table}

\section{Discussion}
Linear calibration maps $\tauhat_{gOB}$ to $\tauhat_{cal}$ in such a way that asymptotic non-inferiority is guaranteed. The linearly-calibrated estimator can provide hypothesis tests and construct confidence intervals using the standard errors proposed in \citetrefs[\S 3.3]{generalizedOB} and a Gaussian approximation; detailed discussion of variance estimators for $\tauhat_{cal}$ is provided in the appendices. Extending \citetrefs{zha20cov}, one can further provide inference that is exact under the sharp null of no effect for any individual while remaining asymptotically valid for the sample average treatment effect. Calibration using only $\muhat_{0}$ and $\muhat_{1}$ fits into a more general class of algorithms wherein calibration is performed on the vectors ${\ddot{x}_{i} = \left(f(x_{i}),\muhat_{0}(x_{i}), \muhat_{1}(x_{i})\right)^{\T}}$ instead of just $\Tilde{x}_{i}$. Taking $f:\R^{k} \rightarrow \R^{\ell}$ adds an additional $\ell$ features to $\Tilde{x}_{i}$. Setting $f(x_i) = x_i$ yields an estimator which is asymptotically no less efficient than $\tauhat_{unadj}$, $\tauhat_{gOB}$, $\hat{\tau}_{cal}$, and Lin's (2013) estimator simultaneously and is akin to the estimator of \citetrefs{col15}; see the appendices for details. 
\looseness=-1

\section*{Supplementary material} \label{SM}
Appendices below includes proofs of Theorems 1 and 2, a discussion of variance estimation, pseudocode, superpopulation results, and a data example.  Code written in \texttt{R} is available at \url{https://github.com/PeterLCohen/OaxacaBlinderCalibration} to implement the method and to reproduce the simulations.


\bibliographystylerefs{apalike}
\bibliographyrefs{bibliography}

\newpage
\begin{center}
    \textbf{\huge{Appendix}}
\end{center}
\appendix

\setcounter{lemma}{0}
\setcounter{assumption}{0}
\setcounter{theorem}{0}
\setcounter{proposition}{0}

\renewcommand{\theassumption}{A.\arabic{assumption}}
\renewcommand{\thecorollary}{A.\arabic{corollary}}
\renewcommand{\thetheorem}{A.\arabic{theorem}}

\section{Extensions and Further Results}
\subsection{Feature Engineering}
While $\tauhat_{cal}$ is no less asymptotically efficient than $\tauhat_{gOB}$ and $\tauhat_{unadj}$, there are no guarantees that $\tauhat_{cal}$ offers an improvement over \citeauthor{lin13}'s \citeyearpar{lin13} regression-adjusted estimator, $\tauhat_{lin}$, in terms of asymptotic variance.  Intuitively, if the prediction functions $\muhat_{0}$ and $\muhat_{1}$ are extremely poor predictors of the potential outcomes, then linear regression on the raw features may substantially outperform even what calibration is able to correct.  Must an experimenter decide \textit{a priori} whether to use $\tauhat_{cal}$ or $\tauhat_{lin}$ when she seeks an estimator of $\SATE$ that is certain to be no less efficient than $\tauhat_{unadj}$?  In fact, there is an estimator which is non-inferior to both $\tauhat_{cal}$ and $\tauhat_{lin}$.

In the calibration algorithm of the main text pseudo-feature vectors were defined as $\Tilde{x}_{i} = \left(\muhat_{0}(x_{i}), \muhat_{1}(x_{i})\right)^{\T}$.  Instead, take $\ddot{x}_{i} = \left(f(x_{i}),\muhat_{0}(x_{i}), \muhat_{1}(x_{i})\right)^{\T}$ where $f:\R^{k} \rightarrow \R^{\ell}$ for some fixed $\ell$.  The function $f$ adds an additional $\ell$ features to $\tilde{x}_{i}$; perhaps incorporating tailored feature engineering guided by domain knowledge.  Define $\tauhat_{cal2}$ to be the calibrated Oaxaca-Blinder estimator based upon $\ddot{x}_{i}$ instead of $\Tilde{x}_{i}$; i.e., incorporate the additional features $f(x_{i})$ in the second-stage linear regression of calibration.

\begin{theorem}\label{supp thm: improves cal and lin}
    Assume the regularity conditions of Theorem~\ref{supp thm: reg adj does no harm}.  So long as the random vectors $\left(f(x_{i}),\muhat_{0}(x_{i}), \muhat_{1}(x_{i})\right)^{\T}$ are sufficiently regular a central limit theorem applies to $N^{1/2}\left(\tauhat_{cal2} - \SATE\right)$ and it is non-inferior to both $N^{1/2}\left(\tauhat_{cal} - \SATE\right)$ and $N^{1/2}\left(\tauhat_{lin} - \SATE\right)$ using the engineered features $f(x_{i})$.
\end{theorem}

\begin{corollary}
    Take $f(x_{i}) = x_{i}$; the resulting estimator $\tauhat_{cal2}$ is asymptotically no less efficient than $\tauhat_{cal}$ and the standard regression adjusted estimator of treatment effect, $\tauhat_{lin}$. 
\end{corollary}

We include proof of Theorem~\ref{supp thm: improves cal and lin} in Section~\ref{sec: proofs} below, along with a more precise statement of regularity conditions.

\subsection{Idempotence}
Let $\mathcal{O}$ denote the set of Oaxaca-Blinder estimators for $\muhat_{0}$ and $\muhat_{1}$ satisfying Assumptions~\ref{supp asm: stability} and \ref{supp asm: simple realizations} of the main text.  Calibration can be thought of as a mapping $\varphi: \mathcal{O} \rightarrow \mathcal{O}$ wherein ${\tauhat_{gOB} \xmapsto{\varphi} \tauhat_{cal}}$; repeated iterations of $\varphi$ induce dynamics on $\mathcal{O}$. A natural question is: do repeated applications of $\varphi$ provide additional improvements in terms of asymptotic efficiency?  
\begin{theorem}\label{supp thm: idempotent}
    $\varphi$ is an idempotent map on $\mathcal{O}$, i.e., $\varphi \circ \varphi = \varphi$.
\end{theorem}
Theorem~\ref{supp thm: idempotent} demonstrates a desirable feature: since $\varphi$ is idempotent, all of the improvement possible through calibration is achieved in one application of $\varphi$.
\begin{proof}
    This proof uses the same line of reasoning as that of Proposition~\ref{prop: lin is special case}, which we include below in Section~\ref{sec: proofs}.  See Remark~\ref{rem: idempotent} for the details of this argument. 
\end{proof}

\section{Regularity Conditions}
In this section we lay out regularity conditions on the potential outcomes and covariates that are sufficient for analyzing the asymptotic distributions of the estimators encountered in the main text.  Finite population inference asymptotics are taken with respect to a sequence of probability spaces which vary with the size of the finite population, $N$.  For each $N$, there are deterministic potential outcomes and covariates for each of the $N$ individuals; randomness enters the model only through the treatment allocation, $Z$.  Our results center around completely randomized experiments; i.e., $Z \sim \text{Unif}(\Omega_{CRE})$.  A basic requirement is that the completely randomized experiments are not asymptotically degenerate.
\begin{assumption}[Non-degeneracy]\label{supp asm: non-degen sampling limit}
	The proportion of treated units, $n_{1} / N$, limits to $p \in (0, 1)$ as $N \rightarrow \infty$.
\end{assumption}

Our main set of assumptions concern the prediction functions $\muhat_{0}$ and $\muhat_{1}$.  These assumptions play into the asymptotic analysis of generalized Oaxaca-Blinder estimators.  The following two regularity conditions match those of \citet{generalizedOB}.

\renewcommand{\theassumption}{\arabic{assumption}}
\setcounter{assumption}{0}
\begin{assumption}[Stability]\label{supp asm: stability}
    There exists a deterministic sequence of functions $\{\overdotmu_{1}^{(N)}\}_{N \in \N}$ such that
    \begin{equation*}
        \left(\frac{1}{N}\sum_{i = 1}^{N}||\overdotmu_{1}^{(N)}(x_{i}) - \muhat_{1}(x_{i})||^{2} \right)^{1/2} = o_{P}(1). 
    \end{equation*}
    The left-hand-side of the formula above satisfies the properties of a norm on functions; this norm is denoted $||\cdot||_{N}$ and so an equivalent statement is that $\left|\left|\muhat_{1} - \overdotmu_{1}^{(N)}\right|\right|_{N} = o_{P}(1)$.
    
    We assume that an analogous sequence, $\{\overdotmu_{0}^{(N)}\}_{N \in \N}$, exists for $\muhat_{0}$.  For notational simplicity we drop the superscripted index and write $\overdotmu_{0}$ and $\overdotmu_{1}$, but the dependence upon $N$ remains an important background detail.  
\end{assumption}

An instructive example of this assumption in practice is when $\muhat_{0}$ and $\muhat_{1}$ are derived via linear regression as in \citet{lin13}.  The deterministic sequence $\{\overdotmu_{1}^{(N)}\}_{N \in \N}$ can be taken as the {population-level} ordinary least squares linear predictor of treated outcome given covariates.  

\begin{assumption}[Vanishing Error Process]\label{app asm: error process vanishes}
    For a function $f: \R^{k} \rightarrow \R$ define\footnote{We use the notation $\G_{N, z}$ to follow that of \citet{generalizedOB}, but the stochastic process $\{\G_{N, z}(\cdot)\}$ indexed by functions $f$ has been implicitly examined elsewhere in the literature, e.g., \citet{rothe}.}
    \begin{equation*}
        \G_{N, z}(f) = N^{-1/2}\sum_{i = 1}^{N}\left(\frac{\indicatorFunction{Z_{i} = z}f(x_{i})}{n_{z}/N} - f(x_{i}) \right).
    \end{equation*}
    Assume that, for $z \in \{0, 1\}$, the error stochastic process $\left|\G_{N, z}(\overdotmu_{z}) -  \G_{N, z}(\muhat_{z})\right|$ vanishes in probability; formally $\left|\G_{N, z}(\overdotmu_{z}) -  \G_{N, z}(\muhat_{z})\right| = o_{P}\left(1\right)$.
\end{assumption}

We include the prediction unbiasedness assumption of \citet{generalizedOB} in order to rigorously discuss asymptotic results for $\tauhat_{gOB}$.

\begin{assumption}[Prediction Unbiasedness]\label{supp asm: prediction unbiasedness}
    For $z=0,1$,   $\sum_{i:Z_i=z}\hat{\mu}_z(x_i) = \sum_{i: Z_i=z}y_i(Z_i)$.
\end{assumption}

The next assumptions constrain the sets of potential outcomes and imputed responses so that central limit behaviour holds for the estimators considered in the main text.  Asymptotic theory for finite population inference builds upon combinatorial analogues for classical probability theory results; see \citet{madowCLT}, \citet{erdosRenyiCLT}, \citet{hajekCLT}, and \citet{hoeffdingCLT} among many others.  \citet{FiniteCLT} developed results for finite population central limit theorems in causal inference; the regularity conditions of their work have become standard in the literature and form the basis for our regularity conditions.  The conditions below naturally generalize those of \citet{lin13} and \citet{freedmanRegAdj2, freedmanRegAdj1} to the nonlinear imputation context of \citet{generalizedOB}.  

\begin{assumption}[Limiting Means and Variances]\label{supp asm: means and covs stabilize}
	  The mean vector and covariance matrix of $(y_i(0), y_i(1), \dot{\mu}_0(x_i), \dot{\mu}_0(x_i))^\T$ have limiting values.  For instance, for ${z=0,1}$ there exists a limiting value $\bar{y}(z)_\infty$ such that $\lim_{N\rightarrow \infty} N^{-1}\sum_{i=1}^Ny_i(z) = \bar{y}(z)_\infty$.
\end{assumption}

\begin{assumption}[Bounded Fourth Moments]\label{supp asm: bounded fourth moment}
 There exists some $C < \infty$ for which, for all ${z = 0,1}$ and all $N$,
		$N^{-1}\sum_{i = 1}^{N}\{y_{i}(z)\}^{4} < C$ and $N^{-1}\sum_{i = 1}^{N}\left\{\dot{\mu}_z(x_i)\right\}^{4} < C.$
\end{assumption}

For comparing $\tauhat_{cal}$ to $\tauhat_{lin}$ we will, at times, need to constrain the covariates $x_{i}$ directly so that \citeauthor{lin13}'s regressions have appropriate asymptotic properties.  For such applications we make the assumption:

\begin{assumption}\label{supp asm: means and covs stabilize for ORIGINAL COVARIATES}
    The mean vector and covariance matrix of $(y_i(0), y_i(1), x_i^{\T})^\T$ have limiting values.  Furthermore, the bounded fourth moment assumption applies componentwise to the raw covariates, i.e., $N^{-1}\sum_{i = 1}^{N}\left\{x_{ij}\right\}^{4} < C.$ where $x_{ij}$ denotes the $j$th coordinate of $x_{i}$.
\end{assumption}

We use the ``bar" notation to denote the mean, i.e., $\bar{{y}}(z) = N^{-1}\sum_{i = 1}^{N}y_{i}(z)$.  Let $\Sigma_{y(0)}$ denote the finite population covariance of the control outcomes, $\Sigma_{y(1)}$ denote its treated analogue, and $\Sigma_{y(z)x}$ the finite population covariance matrix of the joint vectors $(y_{i}(z), x_{i}^{\T})^{\T}$.  In light of this notation, Assumptions~\ref{supp asm: means and covs stabilize} and \ref{supp asm: means and covs stabilize for ORIGINAL COVARIATES} imply that  $\lim_{N \rightarrow \infty}\bar{{y}}(z) = \bar{{y}}_{\infty}(z)$, $\lim_{N \rightarrow \infty}\Sigma_{y(z)} = \Sigma_{y(z),\infty}$, $\lim_{N \rightarrow \infty}\Sigma_{y(z)x} = \Sigma_{y(z)x,\infty}$ for $z \in \{0, 1\}$, etc.  The limiting covariance matrices are assumed to be positive definite.



In the proofs we maintain the implicit assumption that the prediction functions $\muhat_{0}$ and $\muhat_{1}$ are non-collinear and asymptotically almost surely non-constant when evaluated over the set $\{x_{i}\}_{i = 1}^{N}$.  We make the same assumption for $\overdotmu_{0}$ and $\overdotmu_{1}$. Extensions to the rank deficient case are straightforward and are discussed in Section \ref{sec: rank deficient}. Rank deficiency can occur in practice: for instance, if both $\muhat_0$ and $\muhat_1$ are linear regressions and $x_i$ is scalar, then $\overdotmu_0$ and $\overdotmu_1$ are perfectly collinear. Importantly, handling the rank-deficient case does not require additional assumptions, but would complicate the notation used in the proofs. 


For comparison with the assumptions of \citet{generalizedOB} and \citet{rothe} we include the following entropy-based assumption.
\begin{assumption}[Typically Simple Realizations]\label{supp asm: simple realizations}
    There exists a sequence of sets of functions $\tsrC$, which may vary with $N$, such that the random function $\muhat_{0}$ falls into this class asymptotically almost surely.  Formally, $\Prob{\muhat_{0} \in \tsrC} \rightarrow 1$.  Furthermore, the sets of functions are ``small" in the sense that
    \begin{equation*}
        \int_{0}^{1}\sup_{N}\sqrt{\log  \mathscr{N}(\tsrC,||\cdot||_{N}, s)}\, ds < \infty
    \end{equation*}
    where $\mathscr{N}(\tsrC,||\cdot||_{N}, s)$ is the $s$-covering number of $\tsrC$ under the metric induced by $||\cdot||_{N}$.
    An analogous statement holds for $\muhat_{1}$ with a sequence of sets $\tsrT$.
\end{assumption}

The inequality in Assumption~\ref{supp asm: simple realizations} is common in central limit theorems for stochastic processes: it integrates the square root of the maximal entropy without bracketing of $\mathscr{A}_{N, z}$; see \citet[Chapter 11]{ledouxTalagrand}, \citet[Chapter 2]{VanDerVaartWellner}, and \citet{maximalInequalityVDVWellner} for background on the theory and uses of covering numbers and entropy bounds.  The particular upper bound on the region of integration in Assumption~\ref{supp asm: simple realizations} is unimportant, as shown in Lemma~\ref{lem: radius is arbitrary}.  As a result, we say Assumption~\ref{supp asm: simple realizations} holds so long as there is any $D > 0$ for which $\int_{0}^{D}\sup_{N}\sqrt{\log  \mathscr{N}(\mathscr{A}_{N, z},||\cdot||_{N}, s)}\, ds < \infty$ for both $z \in \{0, 1\}$.

\section{Helpful Technical Results}
\renewcommand{\thelemma}{A.\arabic{lemma}}

\begin{lemma}\label{lem: handy variance equation}
    Let $\{a_{i}\}_{i = 1}^{N}$ and $\{b_{i}\}_{i = 1}^{N}$ be two sets of fixed scalars.  Let $Z \sim \textit{Unif}(\Omega_{CRE})$.  The variance of
    \begin{equation}\label{eqn: DiM}
        \frac{1}{n_{1}}\sum_{i = 1}^{N}Z_{i}a_{i} - \frac{1}{n_{0}}\sum_{i = 1}^{N}(1 - Z_{i})b_{i}
    \end{equation}
    is 
    \begin{equation*}
        \frac{n_{0}n_{1}}{N} \cdot \frac{1}{N - 1}\sum_{i = 1}^{N}\left(\frac{a_{i}}{n_{1}} + \frac{b_{i}}{n_{0}} - \overline{\frac{a}{n_{1}} + \frac{b}{n_{0}}} \right)^{2}
    \end{equation*}
    where $\overline{v} = N^{-1}\sum_{j = 1}^{N}v_{i}$ for $v \in \R^{N}$.
\end{lemma}
\begin{proof}
    We start with a simple algebraic manipulation:
    \begin{equation*}
        \frac{1}{n_{1}}\sum_{i = 1}^{N}Z_{i}a_{i} - \frac{1}{n_{0}}\sum_{i = 1}^{N}(1 - Z_{i})b_{i} = \sum_{i = 1}^{N}Z_{i}\left( \frac{a_{i}}{n_{1}} + \frac{b_{i}}{n_{0}} \right) - \sum_{i = 1}^{N}\frac{b_{i}}{n_{0}}.
    \end{equation*}
    Since the second term on the right is constant with respect to $Z$ it plays no part in the variance of \eqref{eqn: DiM} so
    \begin{equation*}
        \Variance{\frac{1}{n_{1}}\sum_{i = 1}^{N}Z_{i}a_{i} - \frac{1}{n_{0}}\sum_{i = 1}^{N}(1 - Z_{i})b_{i}} = \Variance{\sum_{i = 1}^{N}Z_{i}\left( \frac{a_{i}}{n_{1}} + \frac{b_{i}}{n_{0}} \right)}.
    \end{equation*}
    The term on the right is the variance of the population total for drawing a sample of size $n_{1}$ without replacement from the set 
    \begin{equation*}
        \left\{\left(\frac{a_{1}}{n_{1}} + \frac{b_{1}}{n_{0}}\right), \ldots, \left(\frac{a_{N}}{n_{1}} + \frac{b_{N}}{n_{0}} \right)\right\}.
    \end{equation*}
    This variance is well understood in the survey-sampling community; see, for instance, \citet{bootstrapInSurveys}.  Specifically, 
    \begin{equation*}
        \Variance{\sum_{i = 1}^{N}Z_{i}\left( \frac{a_{i}}{n_{1}} + \frac{b_{i}}{n_{0}} \right)} = \frac{n_{0}n_{1}}{N} \cdot \frac{1}{N - 1}\sum_{i = 1}^{N}\left(\frac{a_{i}}{n_{1}} + \frac{b_{i}}{n_{0}} - \overline{\frac{a}{n_{1}} + \frac{b}{n_{0}}} \right)^{2}
    \end{equation*}
\end{proof}
\begin{lemma}\label{lem: regression minimizes asymptotic variance of DiM}
    Consider any two scalars, $w_{0}, v_{0} \in \R$ and any two vectors $W, V \in \R^{\ell}$.  For a set of features $\{\chi_{i}\}_{i = 1}^{N} \subset \R^{\ell}$ define the ``residuals''
    \begin{align*}
        r_{i}(0) &= y_{i}(0) - \left(w_{0} + W^{T}\chi_{i} \right)\\
        r_{i}(1) &= y_{i}(1) - \left(v_{0} + V^{T}\chi_{i} \right).
    \end{align*}
    \begin{enumerate}[i)]
        \item \label{itm: first point}The variance of 
        \begin{equation}\label{eqn: residual DiM}
            \frac{1}{n_{1}}\sum_{i = 1}^{N}Z_{i}r_{i}(1) - \frac{1}{n_{0}}\sum_{i = 1}^{N}(1 - Z_{i})r_{i}(0)
        \end{equation}
        is minimized at
        \begin{align}
            (w_{0}^{*}, W^{*}) &\in \argmin_{w_{0}, W}\left[\sum_{i = 1}^{N} \left\{y_{i}(0) - \left(w_{0} + W^{T}\chi_{i}\right) \right\}^{2}\right]\label{eqn: control reg},\\
            (v_{0}^{*}, V^{*}) &\in \argmin_{v_{0}, V}\left[\sum_{i = 1}^{N} \left\{y_{i}(1) - \left(v_{0} + V^{T}\chi_{i}\right) \right\}^{2}\right]\label{eqn: treated reg}.
        \end{align}
        \item \label{itm: second point} Let $\mathcal{X} = [\chi_{1}, \ldots, \chi_{N}]^{T}$.  The variance of \eqref{eqn: residual DiM} achieves a strict global minimum at $w_{0}^{*}$, $v_{0}^{*}$, $W^{*}$, and $V^{*}$ when the matrix $\mathcal{X}^{T}\mathcal{X}$ is nonsingular.
    \end{enumerate}
    
\end{lemma}
\begin{proof}
    By Lemma~\ref{lem: handy variance equation}, the variance of \eqref{eqn: residual DiM} is proportional to
    \begin{multline*}
        \sum_{i = 1}^{N}\Bigg(\frac{y_{i}(1) - \left(v_{0} + V^{T}\chi_{i} \right)}{n_{1}} + \frac{y_{i}(0) - \left(w_{0} + W^{T}\chi_{i} \right)}{n_{0}} - \\ \overline{\frac{y(1) - \left(v_{0} + V^{T}\chi \right)}{n_{1}} + \frac{y(0) - \left(w_{0} + W^{T}\chi \right)}{n_{0}}} \Bigg)^{2}.
    \end{multline*}
    Rearranging terms gives that the variance of \eqref{eqn: residual DiM} is proportional to
    \begin{multline}\label{eqn: variance as OLS residuals var}
        \sum_{i = 1}^{N}\Bigg(\left(\frac{y_{i}(1)}{n_{1}} + \frac{y_{i}(0)}{n_{0}}\right) - \left(\left(\frac{v_{0} }{n_{1}} + \frac{w_{0}}{n_{0}}\right) + \left(\frac{V}{n_{1}} + \frac{W}{n_{0}}\right)^{T}\chi_{i} \right) - \\ \overline{\left(\frac{y(1)}{n_{1}} + \frac{y(0)}{n_{0}}\right) - \left(\left(\frac{v_{0} }{n_{1}} + \frac{w_{0}}{n_{0}}\right) + \left(\frac{V}{n_{1}} + \frac{W}{n_{0}}\right)^{T}\chi \right)} \Bigg)^{2}.
    \end{multline}
    Since $w_{0}^{*}$ and $W^{*}$ are the intercept and slope, respectively, of the ordinary least squares regression of $y_{i}(0)$ on $\chi_{i}$ it follows that $n_{0}^{-1}w_{0}^{*}$ and $n_{0}^{-1}W^{*}$ are the intercept and slope, respectively, of the ordinary least squares regression of $n_{0}^{-1}y_{i}(0)$ on $\chi_{i}$.  Likewise, $n_{1}^{-1}v_{0}^{*}$ and $n_{1}^{-1}V^{*}$ are the intercept and slope, respectively, of the ordinary least squares regression of $n_{1}^{-1}y_{i}(1)$ on $\chi_{i}$.
    
    Consider now the regression of $n_{1}^{-1}y_{i}(1) + n_{0}^{-1}y_{i}(0)$ on $\chi_{i}$.  Writing the design matrix of this regression as $\mathcal{X}$, the ordinary least squares slope is
    \begin{align*}
        (\mathcal{X}^{T}\mathcal{X})^{-1}\mathcal{X}^{T}\left(\frac{y(1)}{n_{1}} + \frac{y(0)}{n_{0}} \right) &= (\mathcal{X}^{T}\mathcal{X})^{-1}\mathcal{X}^{T}\frac{y(1)}{n_{1}} + (\mathcal{X}^{T}\mathcal{X})^{-1}\mathcal{X}^{T}\frac{y(0)}{n_{0}} \\
        &= \frac{V^{*}}{n_{1}} + \frac{W^{*}}{n_{0}}.
    \end{align*}
    The second equality holds for the following reason: the first term on the right is the slope coefficient in the regression of $n_{1}^{-1}y_{i}(1)$ on $\chi_{i}$ and so it must match $n_{1}^{-1}V^{*}$, similar reasoning applies to the second term.\footnote{In the rank-deficient case that $\mathcal{X}^{T}\mathcal{X}$ is not invertible, these slope terms are not uniquely defined; however, picking a canonical pseudoinverse, e.g., the Moore-Penrose pseudoinverse, obviates this concern.}  Likewise, it follows that the intercept of the regression of $n_{1}^{-1}y_{i}(1) + n_{0}^{-1}y_{i}(0)$ on $\chi_{i}$ is given by $n_{1}^{-1}v_{0}^{*} + n_{0}^{-1}w_{0}^{*}$.  In total, the argument above implies that the ordinary least squares regression predictor of ${n_{1}^{-1}y_{i}(1) + n_{0}^{-1}y_{i}(0)}$ based upon $\chi_{i}$ is
    \begin{equation*}
        \left(\left(\frac{v^{*}_{0} }{n_{1}} + \frac{w^{*}_{0}}{n_{0}}\right) + \left(\frac{V^{*}}{n_{1}} + \frac{W^{*}}{n_{0}}\right)^{T}\chi_{i} \right).
    \end{equation*}
    
    Equation \eqref{eqn: variance as OLS residuals var} is simply the variance of the residuals from attempting to predict ${n_{1}^{-1}y_{i}(1) + n_{0}^{-1}y_{i}(0)}$ via
    \begin{equation*}
        \left(\left(\frac{v_{0} }{n_{1}} + \frac{w_{0}}{n_{0}}\right) + \left(\frac{V}{n_{1}} + \frac{W}{n_{0}}\right)^{T}\chi_{i} \right).
    \end{equation*}
    Since ordinary least squares linear regression minimizes the variance of the residuals, it follows that \eqref{eqn: variance as OLS residuals var} is minimized at $w^{*}_{0}$, $v^{*}_{0}$, $W^{*}$, and $V^{*}$.  Consequently, the variance of \eqref{eqn: residual DiM} is minimized at $w^{*}_{0}$, $v^{*}_{0}$, $W^{*}$, and $V^{*}$, as required to show Part~\ref{itm: first point}.
    
    The result of Part~\ref{itm: second point} follows from the uniqueness of the global minima in the full-rank regression problems \eqref{eqn: control reg} and \eqref{eqn: treated reg}.
\end{proof}

\delineate

\begin{lemma}\label{lem: consistency of OLS coefficients}
    Consider the quantities
    \begin{align*}
        \hat{\beta} &= (\hat{\beta}_{0}, \hat{\beta}_{1}) = \argmin_{\beta_{0}, \beta_{1}} \left[\sum_{i \st Z_{i} = 1} \left\{y_{i}(Z_{i}) - \left(\beta_{0} + \beta_{1}^{T}\begin{bmatrix}\muhat_{0}(x_{i}) \\ \muhat_{1}(x_{i}) \end{bmatrix}\right) \right\}^{2}\right],\\
        \dot{\beta} &= (\dot{\beta}_{0}, \dot{\beta}_{1}) = \argmin_{\beta_{0}, \beta_{1}}\left[\sum_{i = 1}^{N} \left\{y_{i}(1) - \left(\beta_{0} + \beta_{1}^{T}\begin{bmatrix}\dot{\mu}_{0}(x_{i}) \\ \dot{\mu}_{1}(x_{i}) \end{bmatrix}\right) \right\}^{2}\right].
    \end{align*}
Then $||\hat{\beta} - \dot{\beta}||_{2} = o_{P}(1)$ under Assumptions \ref{supp asm: stability} - \ref{supp asm: bounded fourth moment}.\footnote{This lemma generalizes Lemma 7 of \citet{generalizedOB} and also incorporates {both} the imputed outcomes.}
\end{lemma}
\begin{proof}
    Define the design matrices
    \begin{equation*}
        \hat{U}_{1} = \begin{bmatrix}
                        1 & \muhat_{0}(x_{i_1}) & \muhat_{1}(x_{i_1}) \\
                        \vdots & \vdots & \vdots \\
                        1 & \muhat_{0}(x_{i_{n_{1}}}) & \muhat_{1}(x_{i_{n_{1}}})
        \end{bmatrix} \; \text{and} \; \dot{U}_{1} = \begin{bmatrix}
                        1 & \dot{\mu}_{0}(x_{1}) & \dot{\mu}_{1}(x_{1}) \\
                        \vdots & \vdots & \vdots \\
                        1 & \dot{\mu}_{0}(x_{N}) & \dot{\mu}_{1}(x_{N})
        \end{bmatrix}
    \end{equation*}
    where $i_{1}, \ldots, i_{n_1}$ are the indices of the treated units.  Standard ordinary least squares regression theory gives closed-form solutions for $\hat{\beta}$ and $\dot{\beta}$ in terms of $\hat{U}_{1}$ and $\dot{U}_{1}$, respectively, via
    \begin{equation*}
        \hat{\beta} = \left(  \hat{U}_{1}^{\T} \hat{U}_{1} \right)^{-1} \hat{U}_{1}^{\T}\begin{bmatrix}y_{i_{1}}(1) \\ \vdots \\ y_{i_{n_1}}(1) \end{bmatrix} \; \text{and} \; \dot{\beta} = \left(  \dot{U}_{1}^{\T} \dot{U}_{1} \right)^{-1} \dot{U}_{1}^{\T}\begin{bmatrix}y_{1}(1) \\ \vdots \\ y_{N}(1) \end{bmatrix}.
    \end{equation*}
    Conveniently rewriting these regression coefficients by multiplying by one gives
        \begin{equation*}
        \hat{\beta} = \left(\frac{1}{n_{1}}\hat{U}_{1}^{\T} \hat{U}_{1} \right)^{-1} \frac{1}{n_{1}} \hat{U}_{1}^{\T}\begin{bmatrix}y_{i_{1}}(1) \\ \vdots \\ y_{i_{n_1}}(1) \end{bmatrix} \; \text{and} \; \dot{\beta} = \left( \frac{1}{N} \dot{U}_{1}^{\T} \dot{U}_{1} \right)^{-1}\frac{1}{N} \dot{U}_{1}^{\T}\begin{bmatrix}y_{1}(1) \\ \vdots \\ y_{N}(1) \end{bmatrix}.
    \end{equation*}
    
    We first show that 
    \begin{equation}\label{eqn: second part of reg coeffs converges}
        \left| \left| \frac{1}{n_{1}} \hat{U}_{1}^{\T}\begin{bmatrix}y_{i_{1}}(1) \\ \vdots \\ y_{i_{n_1}}(1) \end{bmatrix} - \frac{1}{N} \dot{U}_{1}^{\T}\begin{bmatrix}y_{1}(1) \\ \vdots \\ y_{N}(1) \end{bmatrix} \right| \right|_{2} = o_{P}(1).
    \end{equation}
    These are vectors in $\R^{3}$.  To show \eqref{eqn: second part of reg coeffs converges} we show that the difference in each coordinate is vanishing.  The first coordinate is 
    \begin{equation*}
        \left| \frac{1}{n_{1}}\sum_{Z_{i} = 1} y_{i}(1) - \frac{1}{N}\sum_{i = 1}^{N} y_{i}(1) \right|
    \end{equation*}
    which converges in probability to zero by the weak law of large numbers for finite populations; see for instance \citet[Lemma A.1]{lin13}.  The second term is
    \begin{equation}\label{eqn: second term}
        \left| \frac{1}{n_{1}}\sum_{Z_{i} = 1}\muhat_{0}(x_{i}) y_{i}(1) - \frac{1}{N}\sum_{i = 1}^{N}\dot{\mu}_{0}(x_{i}) y_{i}(1)\right|.
    \end{equation}
    This term vanishes in probability by an application of the triangle inequality, the Cauchy-Schwarz inequality, and a finite population law of large numbers.  The details proceed analogously to those in the proof of Lemma 7 (on page 35) of the arXiv version of \citet{generalizedOB}\footnote{The arXiv version of \citet{generalizedOB} can be found at \url{https://arxiv.org/pdf/2004.11615.pdf}.} except we use $\muhat_{0}(\cdot)$ in place of $\exp(\hat{\theta}_{0}^{\T} \cdot)$; this presents no problem because we have assumed the stability of $\muhat_{0}$.  The third term is 
    \begin{equation*}
        \left| \frac{1}{n_{1}}\sum_{Z_{i} = 1}\muhat_{1}(x_{i}) y_{i}(1) - \frac{1}{N}\sum_{i = 1}^{N}\dot{\mu}_{1}(x_{i}) y_{i}(1)\right|
    \end{equation*}
    which vanishes for exactly the same reason as the second term.  In total, we have shown that \eqref{eqn: second part of reg coeffs converges} holds.
    
    Next, we turn attention to showing that
    \begin{equation}\label{eqn: frobenius norm vanishes}
        \left|\left| \frac{1}{n_{1}}\hat{U}_{1}^{\T} \hat{U}_{1}  -  \frac{1}{N} \dot{U}_{1}^{\T} \dot{U}_{1} \right| \right|_{F} = o_{P}(1).
    \end{equation}
    The matrix in \eqref{eqn: frobenius norm vanishes} is a 3-by-3 symmetric matrix; the term $\frac{1}{n_{1}}\hat{U}_{1}^{\T} \hat{U}_{1}$ takes the form
    \begingroup
        \renewcommand*{\arraystretch}{2} 
        \begin{equation*}
            \begin{bmatrix}\frac{1}{n_{1}}\sum_{Z_{i} = 1}1 & \frac{1}{n_{1}}\sum_{Z_{i} = 1}\muhat_{0}(x_{i}) & \frac{1}{n_{1}}\sum_{Z_{i} = 1}\muhat_{1}(x_{i})\\
            * & \frac{1}{n_{1}}\sum_{Z_{i} = 1}\muhat_{0}(x_{i})^{2} & \frac{1}{n_{1}}\sum_{Z_{i} = 1}\muhat_{0}(x_{i})\muhat_{1}(x_{i})\\
            * & * & \frac{1}{n_{1}}\sum_{Z_{i} = 1}\muhat_{1}(x_{i})^{2}
            \end{bmatrix}
        \end{equation*}
    \endgroup
    where stars indicate symmetry.  The matrix $\frac{1}{N} \dot{U}_{1}^{\T} \dot{U}_{1}$ has an analogous construction where $\overdotmu$ replaces $\muhat$, sums $\sum_{i =1}^{N}$ replace $\sum_{Z_{i} = 1}$, and $N$ replaces $n_{1}$.
    
    In the matrix of \eqref{eqn: frobenius norm vanishes}, the top-left corner element is $\left|\frac{1}{n_{1}}\sum_{Z_{i} = 1}1 - \frac{1}{N}\sum_{i = 1}^{N}1 \right|$ which is deterministically zero.  The remaining two diagonal elements are
    \begin{align*}
        &\left|\frac{1}{n_{1}}\sum_{Z_{i} = 1}\muhat_{0}(x_{i})^{2} - \frac{1}{N}\sum_{i = 1}^{N}\dot{\mu}_{0}(x_{i})^{2} \right|,\\
        &\left|\frac{1}{n_{1}}\sum_{Z_{i} = 1}\muhat_{1}(x_{i})^{2} - \frac{1}{N}\sum_{i = 1}^{N}\dot{\mu}_{1}(x_{i})^{2} \right|.
    \end{align*}
    As before, the convergence of each of these terms in probability to zero is guaranteed by following the argument used to prove that the denominator terms of Lemma 7 from \citet{generalizedOB} vanish.  The only difference is that $\muhat_{z}(\cdot)$ stands in place of $\exp(\hat{\theta}_{z}^{\T} \cdot)$, but this presents no obstacle as we have assumed that the prediction functions $\muhat_{0}(\cdot)$ and $\muhat_{1}(\cdot)$ are stable.  
    
    By the symmetry of the matrix in \eqref{eqn: frobenius norm vanishes}, it suffices to show that 
     \begin{align}
        &\left|\frac{1}{n_{1}}\sum_{Z_{i} = 1}\muhat_{z}(x_{i}) - \frac{1}{N}\sum_{i = 1}^{N}\dot{\mu}_{z}(x_{i}) \right| = o_{P}(1) \quad \text{for } z \in \{0, 1\},\label{eqn: predicted mean vs mudot mean}\\
        &\left|\frac{1}{n_{1}}\sum_{Z_{i} = 1}\muhat_{1}(x_{i})\muhat_{0}(x_{i}) - \frac{1}{N}\sum_{i = 1}^{N}\dot{\mu}_{1}(x_{i})\dot{\mu}_{0}(x_{i}) \right| = o_{P}(1)\label{eqn: cross term}.
    \end{align}
    The same argument used to show that \eqref{eqn: second term} vanishes in probability applies to \eqref{eqn: predicted mean vs mudot mean}; replacing each $y_{i}(1)$ in that argument gives the desired result.  To show, \eqref{eqn: cross term} first add zero to the left-hand-side
    \begin{multline*}
        \left|\frac{1}{n_{1}}\sum_{Z_{i} = 1}\muhat_{1}(x_{i})\muhat_{0}(x_{i}) - \frac{1}{N}\sum_{i = 1}^{N}\dot{\mu}_{1}(x_{i})\dot{\mu}_{0}(x_{i}) \right| = \Bigg|\frac{1}{n_{1}}\sum_{Z_{i} = 1}\muhat_{1}(x_{i})\muhat_{0}(x_{i}) - \frac{1}{n_{1}}\sum_{Z_{i} = 1}\dot{\mu}_{1}(x_{i})\dot{\mu}_{0}(x_{i})  +\\ \frac{1}{n_{1}}\sum_{Z_{i} = 1}\dot{\mu}_{1}(x_{i})\dot{\mu}_{0}(x_{i}) - \frac{1}{N}\sum_{i = 1}^{N}\dot{\mu}_{1}(x_{i})\dot{\mu}_{0}(x_{i}) \Bigg|.
    \end{multline*}
    Applying the triangle inequality for the right-hand-side then bounds above by
    \begin{equation*}
        \underbrace{\Bigg|\frac{1}{n_{1}}\sum_{Z_{i} = 1}\left(\muhat_{1}(x_{i})\muhat_{0}(x_{i}) - \dot{\mu}_{1}(x_{i})\dot{\mu}_{0}(x_{i})\right)\Bigg|}_{\text{Term 1}}  + \underbrace{\Bigg| \frac{1}{n_{1}}\sum_{Z_{i} = 1}\dot{\mu}_{1}(x_{i})\dot{\mu}_{0}(x_{i}) - \frac{1}{N}\sum_{i = 1}^{N}\dot{\mu}_{1}(x_{i})\dot{\mu}_{0}(x_{i}) \Bigg|}_{\text{Term 2}}.
    \end{equation*}
    Just as in \citet[Lemma 7]{generalizedOB} the Cauchy-Schwarz inequality and the stability of $\muhat_{0}(\cdot)$ and $\muhat_{1}(\cdot)$ imply that Term 1 vanishes in probability, and the finite population law of large numbers implies that Term 2 vanishes.  In total, we have shown \eqref{eqn: frobenius norm vanishes}.
    
    The matrix inverse map $M \mapsto M^{-1}$ is continuous over the set of positive definite matrices with the metric defined by the Frobenious norm.  Assuming that the continuous mapping theorem applies \citep[Chapter 18]{asymptoticStats_vdv} (i.e., assume that the matrices $\frac{1}{n_{1}}\hat{U}_{1}^{\T} \hat{U}_{1}$ and $\frac{1}{N} \dot{U}_{1}^{\T} \dot{U}_{1}$ have their minimum eigenvalue bounded away from zero
    ) then \eqref{eqn: frobenius norm vanishes} implies that
    \begin{equation}\label{eqn: inverses also vanish}
        \left|\left| \left(\frac{1}{n_{1}}\hat{U}_{1}^{\T} \hat{U}_{1}\right)^{-1}  -  \left(\frac{1}{N} \dot{U}_{1}^{\T} \dot{U}_{1}\right)^{-1} \right| \right|_{F} = o_{P}(1).
    \end{equation}
    
    Finally, combining \eqref{eqn: second part of reg coeffs converges} with \eqref{eqn: inverses also vanish} 
    yields that  $||\hat{\beta} - \dot{\beta}||_{2} = o_{P}(1)$ as desired.
\end{proof}
Lemma~\ref{lem: consistency of OLS coefficients} applies to the ordinary least squares linear regression coefficients computed within the treated group.  The same logic applies to the ordinary least squares linear regression coefficients computed within the control group. 

\begin{remark}\label{rem: lemma also applies to the second calibrated estimator}
    Suppose that the design matrices used in the proof of Lemma~\ref{lem: consistency of OLS coefficients} were replaced with
     \begin{equation*}
        \hat{U}_{1} = \begin{bmatrix}
                        1 & \muhat_{0}(x_{i_1}) & \muhat_{1}(x_{i_1}) & f(x_{i_1})^{\T}\\
                        \vdots & \vdots & \vdots & \vdots \\
                        1 & \muhat_{0}(x_{i_{n_{1}}}) & \muhat_{1}(x_{i_{n_{1}}}) & f(x_{i_{n_{1}}})^{\T}
        \end{bmatrix} \; \text{and} \; \dot{U}_{1} = \begin{bmatrix}
                        1 & \dot{\mu}_{0}(x_{1}) & \dot{\mu}_{1}(x_{1}) & f(x_{1})^{\T}\\
                        \vdots & \vdots & \vdots & \vdots\\
                        1 & \dot{\mu}_{0}(x_{N}) & \dot{\mu}_{1}(x_{N}) & f(x_{N})^{\T}\\
        \end{bmatrix}.
    \end{equation*}
    As long as the vectors $(\dot{\mu}_{0}(x_{1}), \dot{\mu}_{1}(x_{1}), f(x_{1})^{\T})^{\T}$ satisfy the regularity conditions applied originally (Assumption~\ref{supp asm: means and covs stabilize} and Assumption~\ref{supp asm: bounded fourth moment}) they are amenable to the law of large numbers proofs used in the componentwise analyses for the proof of Lemma~\ref{lem: consistency of OLS coefficients}.  In particular, when $f(\cdot)$ is the identity function, this requirement reduces to Assumption~\ref{supp asm: means and covs stabilize for ORIGINAL COVARIATES} jointly with Assumptions~\ref{supp asm: means and covs stabilize} and \ref{supp asm: bounded fourth moment}.
    
    Consequently, the proof of Lemma~\ref{lem: consistency of OLS coefficients} is equally useful for showing the consistency of the ordinary least squares linear regression coefficients used in $\tauhat_{cal2}$.
\end{remark}

\delineate

\begin{lemma}\label{lem: uniformly bounded N-norm}
    Assumptions~\ref{supp asm: means and covs stabilize} and \ref{supp asm: bounded fourth moment} imply that the quantity
    \begin{equation*}
        \left|\left|\begin{bmatrix}
                                                        \overdotmu_{0} \\ 
                                                        \overdotmu_{1}
                                                    \end{bmatrix}\right|\right|_{N}
    \end{equation*}
    is bounded uniformly in $N$.
\end{lemma}
\begin{proof}
    By Assumption~\ref{supp asm: bounded fourth moment} there exists some finite constant which upper bounds 
	\begin{equation*}
		\frac{\sum_{i = 1}^{N}\left(\overdotmu_{z}(x_{i}) - N^{-1}\sum_{i = 1}^{N}\overdotmu_{z}(x_{i})\right)^{4}}{N}
	\end{equation*}
	for both $z \in \{0, 1\}$ and all $N$.	Bounded fourth central moments imply bounded second central moments, so there exists some constant which upper bounds 
	\begin{equation*}
		\frac{\sum_{i = 1}^{N}\left(\overdotmu_{z}(x_{i}) - N^{-1}\sum_{i = 1}^{N}\overdotmu_{z}(x_{i})\right)^{2}}{N}
	\end{equation*}
	for both $z \in \{0, 1\}$ and all $N$.
	Decomposing the variance into the difference between the second moment and the squared first moment yields that the quantities
	\begin{equation}\label{eqn: variance decomp}
		\frac{1}{N}\sum_{i = 1}^{N}\overdotmu_{z}(x_{i})^{2} - \left\{\frac{1}{N}\sum_{i = 1}^{N}\overdotmu_{z}(x_{i})\right\}^{2}
	\end{equation}
	are uniformly bounded above by a constant.
	
	Assumption~\ref{supp asm: means and covs stabilize} implies that the quantities $N^{-1}\sum_{i = 1}^{N}\overdotmu_{z}(x_{i})$ limit to fixed values for both $z \in \{0, 1\}$; combining this with \eqref{eqn: variance decomp} yields that $\frac{1}{N}\sum_{i = 1}^{N}\overdotmu_{z}(x_{i})^{2}$
	is uniformly bounded above by some constant.  This implies that $\frac{1}{N}\sum_{i = 1}^{N}\left(\overdotmu_{0}(x_{i})^{2} + \overdotmu_{1}(x_{i})^{2}\right)$
	is uniformly bounded.  Lastly,
	\begin{align*}
	    \frac{1}{N}\sum_{i = 1}^{N}\left(\overdotmu_{0}(x_{i})^{2} + \overdotmu_{1}(x_{i})^{2}\right) &= \frac{1}{N}\sum_{i = 1}^{N}\left|\left|\begin{bmatrix}
                            \overdotmu_{0}(x_{i}) \\ 
                            \overdotmu_{1}(x_{i})
                        \end{bmatrix}\right|\right|_{2}^{2}\\
        &= \left|\left|\begin{bmatrix}
            \overdotmu_{0} \\ 
            \overdotmu_{1}
        \end{bmatrix}\right|\right|_{N}^{2}
	\end{align*}
	which concludes the proof.
\end{proof}

\begin{lemma}\label{lem: general prediction unbiaseness in the population for superpopulation models}
    Let $\{\muhat_{z}^{(N)}\}_{N \in \N}$ be a sequence of prediction unbiased functions and assume that there exists some sequence $\{\overdotmu_{z}^{(N)}\}_{N \in \N}$ for which $||\muhat_{z}^{(N)} - \overdotmu_{z}^{(N)}||_{N}$ converges in probability to $0$ with respect to randomness in the superpopulation model (i.e., randomness in $Z$, the potential outcomes and the covariates).  Further assume that $N^{-1}\sum_{i = 1}^{N}\left(\overdotmu_{z}^{(N)}(x_{i}) - y_{i}(z) \right)^{2} = o_{P}(N)$.  Under the preceding conditions, without loss of generality we may assume that $\Expectation{\overdotmu_{z}^{(N)}(x_{i}) - y_{i}(z)} = 0$.
\end{lemma}
\begin{proof}
    Let $\{\overdotmu_{z}^{(N)}\}_{N \in \N}$ for which $||\muhat_{z}^{(N)} - \overdotmu_{z}^{(N)}||_{N}$ converges in probability to $0$ and $N^{-1}\sum_{i = 1}^{N}\left(\overdotmu_{z}^{(N)}(x_{i}) - y_{i}(z) \right)^{2} = o_{P}(n)$.  By the strong law of large numbers,
    $$
        \left| N^{-1}\sum_{i = 1}^{N}\left(\overdotmu_{z}^{(N)}(x_{i})) - y_{i}(z)\right) - \Expectation{\overdotmu_{z}^{(N)}(x_{i}) - y_{i}(z)} \right|
    $$
    converges almost surely to zero.  By the argument of Lemma~3 in the appendix of \citet{generalizedOB}, $\left| N^{-1}\sum_{i = 1}^{N}\left(\overdotmu_{z}^{(N)}(x_{i})) - y_{i}(z)\right)\right| = o_{P}(1)$.  Furthermore, $\Expectation{\overdotmu_{z}^{(N)}(x_{i}) - y_{i}(z)}$ is deterministic; so we must have that $\left|\Expectation{\overdotmu_{z}^{(N)}(x_{i}) - y_{i}(z)}\right| = o(1)$.  Consequently, we may center $\overdotmu_{z}^{(N)}$ by subtracting off $\Expectation{\overdotmu_{z}^{(N)}(x_{i}) - y_{i}(z)}$; defining
    $$
        \dot{\nu}_{z}^{(N)} = \overdotmu_{z}^{(N)} - \Expectation{\overdotmu_{z}^{(N)}(x_{i}) - y_{i}(z)}
    $$
    we retain that $||\muhat_{z}^{(N)} - \dot{\nu}_{z}^{(N)}||_{N}$ converges in probability to $0$ and $N^{-1}\sum_{i = 1}^{N}\left(\dot{\nu}_{z}^{(N)}(x_{i}) - y_{i}(z) \right)^{2} = o_{P}(n)$.
\end{proof}

    \begin{lemma}\label{lem: conditional to unconditional convergence}
        Consider a sequence of random elements $\left\{(\mathcal{A}_{N}, \mathcal{B}_{N})\right\}_{N \in \N}$ and a function $f$ for which $f(\mathcal{A}_{N}, \mathcal{B}_{N}) \in \R$ for all $N \in \N$.  If $f(a_{N}, \mathcal{B}_{N})$ converges in distribution to the random variable $\chi$ for all sequences $\{a_{N}\}_{N \in \N}$ with $a_{N}$ in some measurable set $A_{N}$ such that $\Prob{\mathcal{A}_{N} \in A_{N}} = 1$, then $f(\mathcal{A}_{N}, \mathcal{B}_{N})$ converges in distribution to $\chi$.
    \end{lemma}
    \begin{proof}
        For notation, let $\mathcal{L}(X)$ denote the law of a random variable $X$ and let $\mathcal{L}(X \given Y)$ denote the conditional law of $X$ given $Y$.  The \textit{bounded Lipschitz metric} is a metric on the space of probability measures;
        \begin{equation*}
            d(F, G) := \sup_{f \in \mathcal{F}_{BL}}\left|\int f\,dF - \int f\,dG \right|,
        \end{equation*}
        where $\mathcal{F}_{BL}$ is the class of $1$-Lipschitz functions mapping into $[-1, 1]$.

        The bounded Lipschitz metric metrizes weak convergence of probability measures \citet[Theorem 1.12.4]{VanDerVaartWellner}.  Consequently, the fact that $f(a_{N}, \mathcal{B}_{N})$ converges in distribution to the random variable $\chi$ for $\mathcal{A}_{N}$-almost all $\{a_{N}\}_{N \in \N}$
        \begin{equation*}
            d\left(\mathcal{L}\left(f(a_{N}, \mathcal{B}_{N})\right), \mathcal{L}(\chi)\right) \rightarrow 0
        \end{equation*}
        for $\mathcal{A}_{N}$-almost all $\{a_{N}\}_{N \in \N}$.  Equivalently, the random variable $d\left(\mathcal{L}\left(f(\mathcal{A}_{N}, \mathcal{B}_{N}) \given \mathcal{A}_{N}\right), \mathcal{L}(\chi)\right)$ converges almost surely to $0$ with respect to randomness in $\mathcal{A}_{N}$.  As almost sure convergence implies convergence in probability, $d\left(\mathcal{L}\left(f(\mathcal{A}_{N}, \mathcal{B}_{N}) \given \mathcal{A}_{N}\right), \mathcal{L}(\chi)\right)$ converges in probability to $0$ with respect to randomness in $\mathcal{A}_{N}$.  Theorem 4.1 of \citet{lowDim_HighDim} then implies that $d\left(\mathcal{L}\left(f(\mathcal{A}_{N}, \mathcal{B}_{N})\right), \mathcal{L}(\chi)\right)$ converges to $0$.  Finally, again using that the bounded Lipschitz metric metrizes weak convergence we conclude that unconditionally $f(\mathcal{A}_{N}, \mathcal{B}_{N})$ converges in distribution to $\chi$.
    \end{proof}

\section{Error Processes, Asymptotic Linearity, and Calibration}
For a function $f: \R^{k} \rightarrow \R$ \citet{generalizedOB} define 
\begin{equation*}
    \G_{N, z}(f) = N^{-1/2}\sum_{i = 1}^{N}\left(\frac{\indicatorFunction{Z_{i} = z}f(x_{i})}{n_{z}/N} - f(x_{i}) \right).
\end{equation*}

The collection $\left\{\G_{N, z}\right\}_{f \in F}$ forms a stochastic process indexed by functions $f$ ranging over some family of functions $F$.  Showing that $\left|\G_{N, z}(\overdotmu_{z}) -  \G_{N, z}(\muhat_{z})\right|$ decays quickly in probability is crucial for the central limit theorems undergirding asymptotic inference for imputation-based estimators.  This general principle is true across finite population, fixed covariate, and superpopulation models.
The first consequence of Assumption~\ref{app asm: error process vanishes} is that, under mild regularity conditions, the error process for the calibrated estimator $\muhat_{OLS, z}$ vanishes at the same rate.  Without loss of generality -- for the sake of readability -- we focus only on the prediction equations in the treated group $(z = 1)$ and so we follow the notation of Lemma~\ref{lem: consistency of OLS coefficients}.

Define
\begin{equation*}
    \overdotmu_{OLS, 1}(\cdot) = \dot{\beta}_{0} + \dot{\beta}_{1}^{\T}\begin{bmatrix}\overdotmu_{0}(\cdot) \\ \overdotmu_{1}(\cdot) \end{bmatrix}.
\end{equation*}

\renewcommand{\theproposition}{A.\arabic{proposition}}
\setcounter{proposition}{0}

\begin{proposition}\label{prop: error proces vanishing is inherited by consistent OLS}
    For $z \in \{0, 1\}$ suppose that the random variables $\hat{\beta}_{0}$ and $\hat{\beta}_{1}$ converge in probability to fixed values $\dot{\beta}_{0}$ and $\dot{\beta}_{1}$, respectively.  Under Assumption~\ref{app asm: error process vanishes} it immediately holds that
    \begin{equation*}
        \left|\G_{N, 1}(\overdotmu_{OLS, 1}) -  \G_{N, 1}(\muhat_{OLS, 1})\right| = o_{P}\left(1\right).
    \end{equation*}
    Of course, an analogous result holds for the control quantities as well.
\end{proposition}
\begin{proof}
    Recall the definition 
    \begin{align*}
        \muhat_{OLS, 1}(\cdot) = \hat{\beta}_{0} + \hat{\beta}_{1}^{\T}\begin{bmatrix}\muhat_{0}(\cdot) \\ \muhat_{1}(\cdot) \end{bmatrix}.
    \end{align*}
    Notice that $\G_{N, 1}(\cdot)$ is linear in its argument in the sense that 
    \begin{align*}
        \G_{N, 1}(f + g) &= \G_{N, 1}(f) + \G_{N, 1}(g),\\
        \G_{N, 1}(cf) &= c\G_{N, 1}(f) \text{ for } c\in \R,
    \end{align*}
    so the quantity $ \left|\G_{N, 1}(\overdotmu_{OLS, 1}) -  \G_{N, 1}(\muhat_{OLS, 1})\right|$ decomposes naturally across the linear combination of terms in $\overdotmu_{OLS, 1}$ and $\muhat_{OLS, 1}$.  The assumption of convergence in probability of $\hat{\beta}_{0}$ and $\hat{\beta}_{1}$ to $\dot{\beta}_{0}$ and $\dot{\beta}_{1}$, respectively, combined with Slutsky's theorem and Assumption~\ref{app asm: error process vanishes} concludes the proof.
\end{proof}

In its most abstract sense, the typical analysis of an imputation-based estimator proceeds by writing the estimator in terms of the difference in means of ``residuals" $\left\{\dot{\epsilon}_{i}(1) = y_{i}(1) - \overdotmu_{1}(x_{i})\right\}_{i = 1}^{N}$ and $\left\{\dot{\epsilon}_{i}(0) = y_{i}(0) - \overdotmu_{0}(x_{i})\right\}_{i = 1}^{N}$, possibly some correction factor, and an error term related to $\G_{N, z}(\overdotmu_{z}) -  \G_{N, z}(\muhat_{z})$.  For example, in the finite population context of \citet{generalizedOB}
\begin{multline*}
    \tauhat_{gOB} - \SATE = \underbrace{\frac{1}{n_{1}}\sum_{i \st Z_{i} = 1}\dot{\epsilon}_{i}(1) - \frac{1}{n_{0}}\sum_{i \st Z_{i} = 1}\dot{\epsilon}_{i}(0)}_{\text{Difference in Residual Means}} + \\
    \underbrace{N^{-1/2}\left( \G_{N, 1}(\overdotmu_{1}) -  \G_{N, 1}(\muhat_{1})\right) - N^{-1/2}\left( \G_{N, 0}(\overdotmu_{0}) -  \G_{N, 0}(\muhat_{0})\right)}_{\text{Error Term}}.
\end{multline*}
Likewise, in the context of \citet{rothe}
\begin{multline*}
    \tauhat_{rothe} - \PATE = \underbrace{\frac{1}{n_{1}}\sum_{i \st Z_{i} = 1}\dot{\epsilon}_{i}(1) - \frac{1}{n_{0}}\sum_{i \st Z_{i} = 1}\dot{\epsilon}_{i}(0)}_{\text{Difference in Residual Means}} + \\
    \underbrace{\frac{1}{N}\sum_{i = 1}^{N}\left(\overdotmu_{1}(x_{i}) - \overdotmu_{0}(x_{i})\right) }_{\text{Prediction Unbiasedness Correction Term}} - \PATE + \\ 
    \underbrace{N^{-1/2}\left( \G_{N, 1}(\overdotmu_{1}) -  \G_{N, 1}(\muhat_{1})\right) - N^{-1/2}\left( \G_{N, 0}(\overdotmu_{0}) -  \G_{N, 0}(\muhat_{0})\right)}_{\text{Error Term}}.
\end{multline*}
We provide a formal definition of $\tauhat_{rothe}$ and a detailed analysis of calibrating $\tauhat_{rothe}$ in Section~\ref{sec: semiparametric efficiency}.  

Under Assumption~\ref{app asm: error process vanishes} these error terms are $o_{P}(N^{-1/2})$ and so asymptotic analyses of such estimators can safely ignore the error terms even after scaling by $N^{1/2}$; consequently central limit theorems hold for estimators such as $N^{1/2}(\tauhat_{gOB} - \SATE)$ and $N^{1/2}(\tauhat_{rothe} - \PATE)$ under mild conditions.

\section{Main Proofs}\label{sec: proofs}

In the main text, we remark upon the special case of $\tauhat_{cal}$ when the prediction functions $\muhat_{0}$ and $\muhat_{1}$ are the linear regression prediction functions of outcome based upon covariates, fitted separately in the treated and control groups.  We now formalize this and provide proof.  The identity $\tauhat_{gOB} = \tauhat_{lin}$ in this case is well known; see for instance \citet{missingDataPerspective} or \citet{generalizedOB} and the identity $\tauhat_{GBcal} = \tauhat_{lin}$ follows from similar logic to that presented below.

\begin{proposition}\label{prop: lin is special case}
    Take the original prediction functions to be the linear predictors $\muhat_{1}$ and $\muhat_{0}$ derived from separate regressions in the treated and control groups, respectively; the calibrated estimator $\tauhat_{cal}$ yields exactly $\tauhat_{lin}$.
\end{proposition}
\begin{proof}
    For $z \in \{0, 1\}$ the original linear prediction function $\muhat_{z}$ is defined as ${\muhat_{z}(x) = \beta_{z}^{\T}x + \gamma_{z}}$ where $\beta_{z}$ and $\gamma_{z}$ are the solutions to the $L_{2}$-norm empirical risk minimization problem
    \begin{equation}\label{eqn: first stage OLS}
        \left(\beta_{z}, \gamma_{z}\right) = \argmin_{\substack{\beta_z \in \R^{k} \\ \gamma_{z} \in \R}}\left[ \sum_{i \st Z_{i} = z} \left(y_{i}(z) - \beta_z^{\T}x_{i} - \gamma_{z} \right)^{2}\right].
    \end{equation}
    Let $\Tilde{x}_{i} = \left(\muhat_{0}(x_{i}), \muhat_{1}(x_{i})\right)^{\T}$ .  The calibrated predictor $\muhat_{OLS, z}$ is defined similarly as $\muhat_{OLS, z}(\Tilde{x}) = \beta_{OLS, z}^{\T}\Tilde{x} + \gamma_{OLS, z}$ for
    \begin{equation}\label{eqn: second stage OLS}
        (\beta_{OLS, z}, \gamma_{OLS, z}) = \argmin_{\substack{B_z \in \R^{2}\\ \eta_{z} \in \R}} \left[ \sum_{i \st Z_{i} = z} \left(y_{i}(z) - B_z^{\T}\Tilde{x}_{i} - \eta_{z} \right)^{2}\right].
    \end{equation}
    Writing $B_z \in \R^{2}$ as the vector $(a_z, b_z)^{\T}$ we expand the objective function of \eqref{eqn: second stage OLS} via
    \begin{align}
            \sum_{i \st Z_{i} = z} \left(y_{i}(z) - B_{z}^{\T}\Tilde{x}_{i} - \eta_{z}\right)^{2} &= \sum_{i \st Z_{i} = z} \Big\{y_{i}(z) - \left(a_z\muhat_{0}(x_{i}) +  b_z\muhat_{1}(x_{i}) + \eta_{z}\right) \Big\}^{2}\nonumber\\
        &= \sum_{i \st Z_{i} = z} \Bigg[y_{i}(z) - \Bigg\{a_z\beta_{0}^{\T}x_{i} + b_{z}\beta_{1}^{\T}x_{i} +  (a_z\gamma_{0} + b_{z}\gamma_{1} + \eta_{z})\Bigg\} \Bigg]^{2}\nonumber\\
        &= \sum_{i \st Z_{i} = z} \Bigg\{y_{i}(z) - \left(a_z\beta_{0} + b_z\beta_{1}\right)^{\T}x_{i} - (a_z\gamma_{0} + b_z\gamma_{1} + \eta_{z}) \Bigg\}^{2}\nonumber.
    \end{align}
    Since $a_z\gamma_{0} + b_z\gamma_{1} + \eta_{z}$ ranges over all of $\R$ just as $\eta_{z}$ does, by a change of variables we can write $\muhat_{OLS, z}(\Tilde{x})$ directly as a function of the original covariates.  This reduces to $\textnormal{\ss}_{OLS, z}^{\T}x + \gamma_{OLS, z}$ derived from the solution of
    \begin{equation}\label{eqn: reformulated second stage}
        \left( \textnormal{\ss}_{OLS, z}, \gamma_{OLS, z} \right) =  \argmin_{\substack{\textnormal{\ss}_z \in \text{Span}(\beta_{0}, \beta_{1}) \\ \eta_{z} \in \R}} \left[ \sum_{i \st Z_{i} = z} \left\{y_{i}(z) - \textnormal{\ss}_z^{\T}x_{i} - \eta_{z} \right\}^{2}\right].
    \end{equation}
    where $\text{Span}(\beta_{0}, \beta_{1})$ denotes the linear subspace of $\R^{k}$ generated by $\beta_{0}$ and $\beta_{1}$.
    
    Trivially $\beta_{0}, \beta_{1} \in \text{Span}(\beta_{0}, \beta_{1})$ and, by construction, these offer the solutions to the unconstrained optimization problems of \eqref{eqn: first stage OLS} for $z=0$ and $z=1$, respectively; thus $\beta_{0}$ and $\beta_{1}$ offer feasible optimal solutions to the constrained problem \eqref{eqn: reformulated second stage} as well for $z=0$ and $z=1$, respectively.  Consequently, $\textnormal{\ss}_{OLS, z} = \beta_{z}$ and $\gamma_{OLS, z} = \gamma_{z}$ for $z \in \{0, 1\}$.  Thus, $\muhat_{OLS, 0}(\Tilde{x}_{i}) = \muhat_{0}(x_{i})$ and likewise $\muhat_{OLS, 1}(\Tilde{x}_{i}) = \muhat_{1}(x_{i})$.  From this it follows that $\tauhat_{cal} = \tauhat_{lin}$.
\end{proof}

\begin{remark}\label{rem: idempotent}
    The proof of Proposition~\ref{prop: lin is special case} amounts to demonstrating that when the original prediction algorithms $\muhat_{0}$ and $\muhat_{1}$ are affine functions of the features $x_{i}$ the optimality of the ordinary least squares regression coefficients implies that $\tauhat_{cal}$ must match $\tauhat_{lin}$.  
    
    Suppose that one calibrates $\muhat_{0}$ and $\muhat_{1}$ to form the estimator $\tauhat_{cal}$ and the predictors $\muhat_{OLS, 0}$ and $\muhat_{OLS, 1}$.  In light of Theorem~\ref{supp thm: idempotent}, imagine that one calibrates again: in other words, one forms new predictors $\muhat_{OLS2, 0}$ and $\muhat_{OLS2, 1}$ by regressing upon $\left(\muhat_{OLS, 0}(x_{i}), \muhat_{OLS, 1}(x_{i})\right)^{\T}$ in the treated and control groups separately.  Mirroring the logic of the proof above, we can write $\muhat_{OLS2, 0}$ and $\muhat_{OLS2, 1}$ directly as affine functions of the original prediction functions $\left(\muhat_{0}(x_{i}), \muhat_{1}(x_{i})\right)^{\T}$ and the optimality of the original linear regressions used to form $\muhat_{OLS, 0}$ and $\muhat_{OLS, 1}$ implies that 
    \begin{equation*}
        \muhat_{OLS, 0}(x_{i}) = \muhat_{OLS2, 0}(x_{i}); \quad \muhat_{OLS, 1}(x_{i}) = \muhat_{OLS2, 1}(x_{i}).
    \end{equation*}
    This implies the idempotence of calibration.
\end{remark}

\delineate

\renewcommand*{\theproposition}{A.\arabic{proposition}}
For the purpose of easy computation, we present the following proposition.  Section~\ref{sec: code} below provides further implementation details.

\begin{proposition}\label{prop: fully imputation}
    The calibrated estimator can be written exclusively in terms of the predicted values: ${\tauhat_{cal} = N^{-1}\sum_{i = 1}^{N}\left\{\muhat_{OLS, 1}(x_{i}) - \muhat_{OLS, 0}(x_{i}) \right\}}$. 
\end{proposition}
\begin{proof}
    This follows from Proposition~\ref{prop: linear calibration is prediction unbiased}; which we present below.
\end{proof}

\begin{proposition}\label{prop: linear calibration is prediction unbiased}
    With probability one, the prediction functions $\muhat_{OLS, 0}$ and $\muhat_{OLS, 1}$ satisfy
     \begin{equation}\label{eqn: pred unbiased}
        \sum_{i= 1}^{N}\indicatorFunction{Z_{i} = z}\muhat_{OLS, z}(x_{i}) = \sum_{i= 1}^{N}\indicatorFunction{Z_{i} = z}y_{i}(z).
    \end{equation}
    In the terminology of \citet[Definition 1]{generalizedOB} $\muhat_{OLS, 0}$ and $\muhat_{OLS, 1}$ are {prediction unbiased}.
\end{proposition}
\begin{proof}
    This is a direct consequence of the first order optimality conditions for the intercept term in the regressions defining $\muhat_{OLS, 0}$ and $\muhat_{OLS, 1}$.
\end{proof}

\begin{remark}
    Since the proof of Proposition~\ref{prop: linear calibration is prediction unbiased} relied on the first-order optimality condition for the intercept term in the regressions defining $\muhat_{OLS, 0}$ and $\muhat_{OLS, 1}$ the logic of the proof applies equally well when $\tilde{x}_{i}$ is replaced with $\ddot{x}_{i}$, thereby easily incorporating feature engineering.  
\end{remark}

\delineate 

In what follows we focus attention on the resulting prediction equations in the treated group ($z=1$) after linear calibration, and will at times introduce quantities with the dependence on $z$ suppressed for readability. The proofs for the control group are analogous. Recall from before the quantities
\begin{align}
    \hat{\beta} &= (\hat{\beta}_{0}, \hat{\beta}_{1}) = \argmin_{\beta_{0}, \beta_{1}} \left[\sum_{i \st Z_{i} = 1} \left\{y_{i}(Z_{i}) - \left(\beta_{0} + \beta_{1}^{T}\begin{bmatrix}\muhat_{0}(x_{i}) \\ \muhat_{1}(x_{i}) \end{bmatrix}\right) \right\}^{2}\right],\label{eqn: useful notation sample beta}\\
    \dot{\beta} &= (\dot{\beta}_{0}, \dot{\beta}_{1}) = \argmin_{\beta_{0}, \beta_{1}}\left[\sum_{i = 1}^{N} \left\{y_{i}(1) - \left(\beta_{0} + \beta_{1}^{T}\begin{bmatrix}\dot{\mu}_{0}(x_{i}) \\ \dot{\mu}_{1}(x_{i}) \end{bmatrix}\right) \right\}^{2}\right].\label{eqn: useful notation ``population`` beta}
\end{align}
The quantities $\hat{\beta}_{0}$ and $\hat{\beta}_{1}$ are the ordinary least squares linear regression intercept and slope coefficients, respectively, for the sample regression of treated outcomes on the imputed values $\muhat_{0}(x_{i})$ and $\muhat_{1}(x_{i})$.  Consequently $\hat{\beta}$ can be computed based upon the observed data.  In contrast, $\dot{\beta}_{0}$ and $\dot{\beta}_{1}$ are the population-level ordinary least squares regression intercept and slope coefficients, respectively, for the regression of all treated potential outcomes upon the ``imputed" values $\overdotmu_{0}(x_{i})$ and $\overdotmu_{1}(x_{i})$.  Importantly, $\dot{\beta}$ is not generally computable from the observed data for two reasons:
\begin{enumerate}
    \item The minimization problem defining $\dot{\beta}$ requires knowledge of the treated potential outcomes for each individual, not just those who received treatment.
    \item In practice, the functions $\overdotmu_{0}(x_{i})$ and $\overdotmu_{1}(x_{i})$ are frequently unknown.  For instance, they can often be taken to be solutions to population-level risk minimization procedures; see \citet{generalizedOB} for some examples of this.
\end{enumerate}

Define $\overdotmu_{OLS, 1}(\cdot)$ to be $\dot{\beta}_{0} + \dot{\beta}_{1}^{\T}\begin{bmatrix}\dot{\mu}_{0}(\cdot) \\ \dot{\mu}_{1}(\cdot) \end{bmatrix}$.  


\begin{proposition}\label{prop: stability inherited}
    If the original prediction functions $\muhat_{0}$ and $\muhat_{1}$ are stable in the sense of Assumption~\ref{supp asm: stability}, then under Assumptions~\ref{supp asm: means and covs stabilize} and \ref{supp asm: bounded fourth moment} the calibrated prediction functions $\muhat_{OLS, 0}$ and $\muhat_{OLS, 1}$ are also stable.
\end{proposition}
\begin{proof}
    Our proof focuses on $\muhat_{OLS, 1}$ and uses the notation of \eqref{eqn: useful notation sample beta} and \eqref{eqn: useful notation ``population`` beta}; the proof for $\muhat_{OLS, 0}$ follows the same logic, but requires exchanging control quantities with treated quantities in the obvious places.
    
    To show stability of $\muhat_{OLS, 1}$ we will show that
    \begin{equation*}
        \frac{1}{N}\sum_{i = 1}^{N}\Big|\muhat_{OLS, 1}(x_{i}) - \overdotmu_{OLS, 1}(x_{i})  \Big|^{2} = o_{P}(1).
    \end{equation*}
    This amounts to examining
    \begin{equation}\label{eqn: initial stability equation}
        \frac{1}{N}\sum_{i = 1}^{N}\left|\left(\hat{\beta}_{0} + \hat{\beta}_{1}^{\T}\begin{bmatrix}\muhat_{0}(x_{i}) \\ \muhat_{1}(x_{i}) \end{bmatrix} \right)- \left(\dot{\beta}_{0} + \dot{\beta}_{1}^{\T}\begin{bmatrix}\dot{\mu}_{0}(x_{i}) \\ \dot{\mu}_{1}(x_{i}) \end{bmatrix} \right)  \right|^{2}.
    \end{equation}
    Clearly this quantity is non-negative, so our objective is to provide an $o_{P}(1)$ upper bound.  By the triangle inequality \eqref{eqn: initial stability equation} is upper bounded by
    \begin{equation*}
                \frac{1}{N}\sum_{i = 1}^{N}\left\{\left|\hat{\beta}_{0} - \dot{\beta}_{0} \right| + \left| \hat{\beta}_{1}^{\T}\begin{bmatrix}\muhat_{0}(x_{i}) \\ \muhat_{1}(x_{i}) \end{bmatrix} - \dot{\beta}_{1}^{\T}\begin{bmatrix}\dot{\mu}_{0}(x_{i}) \\ \dot{\mu}_{1}(x_{i}) \end{bmatrix} \right|\right\}^{2}.
    \end{equation*}
    
    By Lemma~\ref{lem: consistency of OLS coefficients} the term $\left|\hat{\beta}_{0} - \dot{\beta}_{0} \right| = o_{P}(1)$ and so  \eqref{eqn: initial stability equation} is upper bounded by
    \begin{equation}\label{eqn: cleared away the first term in upper bound}
                \frac{1}{N}\sum_{i = 1}^{N}\left\{o_{P}(1) + \left| \hat{\beta}_{1}^{\T}\begin{bmatrix}\muhat_{0}(x_{i}) \\ \muhat_{1}(x_{i}) \end{bmatrix} - \dot{\beta}_{1}^{\T}\begin{bmatrix}\dot{\mu}_{0}(x_{i}) \\ \dot{\mu}_{1}(x_{i}) \end{bmatrix} \right|\right\}^{2}.
    \end{equation}
    In light of \eqref{eqn: cleared away the first term in upper bound} we turn our focus to the term 
    \begin{align}
        \left| \hat{\beta}_{1}^{\T}\begin{bmatrix}\muhat_{0}(x_{i}) \\ \muhat_{1}(x_{i}) \end{bmatrix} - \dot{\beta}_{1}^{\T}\begin{bmatrix}\dot{\mu}_{0}(x_{i}) \\ \dot{\mu}_{1}(x_{i}) \end{bmatrix} \right| &= \left| \hat{\beta}_{1}^{\T}\begin{bmatrix}\muhat_{0}(x_{i}) \\ \muhat_{1}(x_{i}) \end{bmatrix} - \hat{\beta}_{1}^{\T}\begin{bmatrix}\overdotmu_{0}(x_{i}) \\ \overdotmu_{1}(x_{i}) \end{bmatrix} +  \hat{\beta}_{1}^{\T}\begin{bmatrix}\overdotmu_{0}(x_{i}) \\ \overdotmu_{1}(x_{i}) \end{bmatrix} - \dot{\beta}_{1}^{\T}\begin{bmatrix}\dot{\mu}_{0}(x_{i}) \\ \dot{\mu}_{1}(x_{i}) \end{bmatrix} \right|\nonumber\\
        &\leq \left| \hat{\beta}_{1}^{\T}\begin{bmatrix}\muhat_{0}(x_{i}) \\ \muhat_{1}(x_{i}) \end{bmatrix} - \hat{\beta}_{1}^{\T}\begin{bmatrix}\overdotmu_{0}(x_{i}) \\ \overdotmu_{1}(x_{i}) \end{bmatrix} \right| + \left| \hat{\beta}_{1}^{\T}\begin{bmatrix}\overdotmu_{0}(x_{i}) \\ \overdotmu_{1}(x_{i}) \end{bmatrix} - \dot{\beta}_{1}^{\T}\begin{bmatrix}\dot{\mu}_{0}(x_{i}) \\ \dot{\mu}_{1}(x_{i}) \end{bmatrix} \right|\nonumber\\
        &\leq \left|\left|\hat{\beta}_{1} \right|\right|_{2}\left|\left| \begin{bmatrix}\muhat_{0}(x_{i}) \\ \muhat_{1}(x_{i}) \end{bmatrix} - \begin{bmatrix}\overdotmu_{0}(x_{i}) \\ \overdotmu_{1}(x_{i}) \end{bmatrix} \right|\right|_{2} + \left| \left| 
        \begin{bmatrix}\dot{\mu}_{0}(x_{i}) \\ \dot{\mu}_{1}(x_{i}) \end{bmatrix}
       \right|\right|_{2}\left|\left| \hat{\beta}_{1}^{\T} - \dot{\beta}_{1}^{\T} \right|\right|_{2}\nonumber\\
        &\leq \left|\left|\hat{\beta}_{1} \right|\right|_{2}\left|\left| \begin{bmatrix}\muhat_{0}(x_{i}) \\ \muhat_{1}(x_{i}) \end{bmatrix} - \begin{bmatrix}\overdotmu_{0}(x_{i}) \\ \overdotmu_{1}(x_{i}) \end{bmatrix} \right|\right|_{2} + \left| \left| 
        \begin{bmatrix}\dot{\mu}_{0}(x_{i}) \\ \dot{\mu}_{1}(x_{i}) \end{bmatrix}
       \right|\right|_{2}o_{P}(1)\nonumber\\
        &\leq O_{P}(1)\left|\left| \begin{bmatrix}\muhat_{0}(x_{i}) \\ \muhat_{1}(x_{i}) \end{bmatrix} - \begin{bmatrix}\overdotmu_{0}(x_{i}) \\ \overdotmu_{1}(x_{i}) \end{bmatrix} \right|\right|_{2} + \left| \left| 
        \begin{bmatrix}\dot{\mu}_{0}(x_{i}) \\ \dot{\mu}_{1}(x_{i}) \end{bmatrix}
       \right|\right|_{2}o_{P}(1)\label{eqn: decoupled bound}.
    \end{align}
    The second line follows from the triangle inequality, the third line from the Cauchy-Schwarz inequality, and the fourth line from Lemma~\ref{lem: consistency of OLS coefficients}.  The last line follows from $\left|\left|\hat{\beta}_{1} \right|\right|_{2} = O_{P}(1)$.  To see why this is the case, notice that by standard theory for ordinary least squares linear regression, under Assumptions~\ref{supp asm: means and covs stabilize} and \ref{supp asm: bounded fourth moment}, $\dot{\beta}_{1}$ converges to a fixed vector in $\R^{2}$; so the consistency of $\hat{\beta}_{1}$ implies that $\left|\left|\hat{\beta}_{1} \right|\right|_{2} = O_{P}(1)$.
    
    Combining \eqref{eqn: cleared away the first term in upper bound} with \eqref{eqn: decoupled bound} gives that \eqref{eqn: initial stability equation} is upper bounded by
    \begin{equation}\label{eqn: upper bound stage 1}
        \frac{1}{N}\sum_{i = 1}^{N}\left\{o_{P}(1) + O_{P}(1)\left|\left| \begin{bmatrix}\muhat_{0}(x_{i}) \\ \muhat_{1}(x_{i}) \end{bmatrix} - \begin{bmatrix}\overdotmu_{0}(x_{i}) \\ \overdotmu_{1}(x_{i}) \end{bmatrix} \right|\right|_{2} + \left| \left| 
        \begin{bmatrix}\dot{\mu}_{0}(x_{i}) \\ \dot{\mu}_{1}(x_{i}) \end{bmatrix}
       \right|\right|_{2}o_{P}(1)\right\}^{2}.
    \end{equation}
    Applying the inequality $2a^{2} + 2b^{2} \geq (a + b)^{2}$ twice yields that \eqref{eqn: upper bound stage 1} is bounded above by
    \begin{multline*}
        \frac{4}{N}\sum_{i = 1}^{N}\left[\left\{O_{P}(1)\left|\left| \begin{bmatrix}\muhat_{0}(x_{i}) \\ \muhat_{1}(x_{i}) \end{bmatrix} - \begin{bmatrix}\overdotmu_{0}(x_{i}) \\ \overdotmu_{1}(x_{i}) \end{bmatrix} \right|\right|_{2}\right\}^{2} + \left\{\left| \left| 
        \begin{bmatrix}\dot{\mu}_{0}(x_{i}) \\ \dot{\mu}_{1}(x_{i}) \end{bmatrix}
       \right|\right|_{2}o_{P}(1)\right\}^{2} + o_{P}(1)\right] = \\
         \underbrace{\frac{O_{P}(1)}{N}\sum_{i = 1}^{N}\left|\left| \begin{bmatrix}\muhat_{0}(x_{i}) \\ \muhat_{1}(x_{i}) \end{bmatrix} - \begin{bmatrix}\overdotmu_{0}(x_{i}) \\ \overdotmu_{1}(x_{i}) \end{bmatrix} \right|\right|_{2}^{2}}_{\text{Term 1}} + \underbrace{\frac{o_{P}(1)}{N}\sum_{i = 1}^{N} \left| \left| 
        \begin{bmatrix}\dot{\mu}_{0}(x_{i}) \\ \dot{\mu}_{1}(x_{i}) \end{bmatrix}
        \right|\right|_{2}^{2}}_{\text{Term 2}} + o_{P}(1).
    \end{multline*}
    
    Term 1 is vanishing in probability by the stability of $\muhat_{0}$ and $\muhat_{1}$.  By Assumptions~\ref{supp asm: means and covs stabilize} and \ref{supp asm: bounded fourth moment}, 
    \begin{equation*}
        \frac{1}{N}\sum_{i = 1}^{N} \left| \left| 
        \begin{bmatrix}\dot{\mu}_{0}(x_{i}) \\ \dot{\mu}_{1}(x_{i}) \end{bmatrix}
        \right| \right|_{2}^{2}
    \end{equation*}
    is limits to a constant, so Term 2 vanishes in probability as well.  In total, we have established that \ref{eqn: initial stability equation} is $o_{P}(1)$ which concludes the proof.
\end{proof}

\begin{remark}\label{rem: stability inherited with feature engineering}
    If Assumptions~\ref{supp asm: means and covs stabilize} and \ref{supp asm: bounded fourth moment} apply to the feature engineered pseudo-covariates $(f(x_{i})^{\T}, \overdotmu_{0}(x_{i}), \overdotmu_{1}(x_{i}))^{\T}$ then the same logic applies to the calibration regressions performed in the formation of $\tauhat_{cal2}$.  Consequently, the calibrated prediction functions used for the feature-engineered estimator $\tauhat_{cal2}$ are prediction unbiased (in the sense of \citet[Definition 1]{generalizedOB}) and stable.
\end{remark}

\delineate




\setcounter{theorem}{0}
\renewcommand*{\thetheorem}{\arabic{theorem}}
We restate the theorems of the main text with the precise regularity conditions required.

\begin{theorem}\label{supp thm: reg adj does no harm}
    Suppose that the completely randomized sampling is non-degenerate (Assumption~\ref{supp asm: non-degen sampling limit}).  If the original prediction functions $\muhat_{0}$ and $\muhat_{1}$ are stable with vanishing error processes (Assumptions~\ref{supp asm: stability} and \ref{app asm: error process vanishes}) and satisfy Assumptions~\ref{supp asm: means and covs stabilize} and \ref{supp asm: bounded fourth moment}, then $\tauhat_{cal}$ is consistent, obeys a central limit theorem, and is asymptotically no less efficient than $\tauhat_{unadj}$.

\end{theorem}
\begin{proof}
    Proposition~\ref{prop: linear calibration is prediction unbiased} establishes that $\muhat_{OLS, 0}$ and $\muhat_{OLS, 1}$ are prediction unbiased (in the sense of \citet[Definition 1]{generalizedOB}).  Proposition~\ref{prop: stability inherited} implies the stability of the prediction functions $\muhat_{OLS, 0}$ and $\muhat_{OLS, 1}$.  Proposition~\ref{prop: error proces vanishing is inherited by consistent OLS} implies that $\muhat_{OLS, 0}$ and $\muhat_{OLS, 1}$ satisfy the vanishing error process assumption.  Thus, the argument of Theorem~2 of \citet{generalizedOB} guarantees the consistency of $\tauhat_{cal}$ as long as
    \begin{equation}
      MSE_{N}(z)= \frac{1}{N}\sum_{i = 1}^{N}\left(y_{i}(z) - \overdotmu_{OLS, z}(x_{i})\right)^{2} = o(N) \quad \text{for } z \in \{0, 1\}
    \end{equation}
    a requirement which is met by Assumptions~\ref{supp asm: means and covs stabilize} and \ref{supp asm: bounded fourth moment}.
    
    The consistency of $\tauhat_{cal}$ does not rely upon whether or not any of the regression models $\muhat_{0}$, $\muhat_{1}$, $\muhat_{OLS, 0}$, and $\muhat_{OLS, 1}$ are ``well-specified" with respect to the data-generating process that gave rise to the potential outcomes and covariates.  In other words, consistency is derived without any special regard for knowing how the data came to be.

    Furthermore, 
    since Proposition~\ref{prop: error proces vanishing is inherited by consistent OLS} implies that $\muhat_{OLS, 0}$ and $\muhat_{OLS, 1}$ have vanishing error processes  the argument of \citet[Theorem~3]{generalizedOB} implies that $\tauhat_{cal}$ has an asymptotically linear reformulation:
    \begin{align}
      \frac{1}{N}\sum_{i = 1}^{N}\left(\hat{y}_{i}(1) - y_{i}(1) \right) &= \frac{1}{n_{1}}\sum_{i \st Z_{i} = 1}\underbrace{\left(y_{i}(1) - \overdotmu_{OLS, 1}(x_i)\right)}_{\dot{\epsilon}_{i}(1)} + o_{P}(N^{-1/2})\label{eqn: linear reformulation treated},\\
      \frac{1}{N}\sum_{i = 1}^{N}\left(\hat{y}_{i}(0) - y_{i}(0) \right) &= \frac{1}{n_{0}}\sum_{i \st Z_{i} = 0}\underbrace{\left(y_{i}(0) - \overdotmu_{OLS, 0}(x_i)\right)}_{\dot{\epsilon}_{i}(0)} + o_{P}(N^{-1/2})\label{eqn: linear reformulation control}.
    \end{align}
    
    Equations~\eqref{eqn: linear reformulation treated} and \eqref{eqn: linear reformulation control} show that the calibrated estimator $\tauhat_{cal} - \SATE$ is equivalent, up to an error of $o_{P}(N^{-1/2})$, to the difference in means estimator for a population where the potential outcomes are given by the residuals $\dot{\epsilon}_{i}(1)$ and $\dot{\epsilon}_{i}(0)$ instead of the original potential outcomes. By standard central limit theorems for the difference in means, e.g., \citet{FiniteCLT}, it follows that under the conditions of Corollary 1 in \citet{generalizedOB} 
    \begin{equation*}
        N^{1/2}\left(\frac{\tauhat_{cal} - \SATE}{\sigma_{N}}\right) \convD \Normal{0}{1},
    \end{equation*}
    where
    \begin{equation}\label{eqn: squared studenizing factor for cal}
        \sigma^{2}_{N} = \left(\frac{1}{n_{1}} \right)MSE_{N}(1) + \left(\frac{1}{n_{0}} \right)MSE_{N}(0) - \frac{1}{N(N-1)}\sum_{i = 1}^{N}\left(\dot{\epsilon}_{i}(1) - \dot{\epsilon}_{i}(0)\right)^{2}.
    \end{equation}
    Assumptions~\ref{supp asm: means and covs stabilize} and \ref{supp asm: bounded fourth moment} are sufficient conditions for Corollary 1 of \citet{generalizedOB}.

    What remains to be shown is that $\tauhat_{cal}$ is asymptotically no less efficient than $\tauhat_{unadj}$.  In other words, we must show that
    \begin{equation}\label{ineq: non-inferior}
        \lim_{N \rightarrow \infty} N\sigma_{N}^{2} \leq \frac{\Sigma_{y(1), \infty}}{p} + \frac{\Sigma_{y(0), \infty}}{1 - p} - \Sigma_{\tau, \infty}
    \end{equation}
    where $\Sigma_{\tau, \infty}$ is the limiting variance of $\{\tau_{i} = y_{i}(1) - y_{i}(0)\}_{i = 1}^{N}$.  The limit of $N\sigma_{N}^{2}$ is guaranteed to exist due to Assumption~\ref{supp asm: means and covs stabilize}.
    
    Consider forming Lin's regression adjusted estimator by regressing the observed outcomes $y_{i}(Z_{i})$ on the deterministic features $\left(\overdotmu_{0}(x_{i}), \overdotmu_{1}(x_{i}) \right)^{\T}$ separately in the treated and control groups, imputing counterfactuals based upon the corresponding regressions, and then taking the difference in means across the imputed populations.  Call this estimator $\hat{T}_{lin}$.  $\hat{T}_{lin}$ differs from $\tauhat_{cal}$ in that the features used for regression are not those estimated from the sample, i.e., $\left(\muhat_{0}(x_{i}), \muhat_{1}(x_{i}) \right)^{\T}$.  Nonetheless, Theorem~1 of \citet{lin13} gives that the asymptotic variance of $N^{1/2}\left(\hat{T}_{lin} - \SATE\right)$ is $\lim_{N \rightarrow \infty} N\sigma_{N}^{2}$.  Corollary~1 of \citet{lin13} then implies the non-inferiority of $\hat{T}_{lin}$ relative to the unadjusted difference in means, which is equivalent to the inequality of \ref{ineq: non-inferior}.  This indirectly shows the non-inferiority of $\tauhat_{cal}$ relative to the unadjusted difference in means $\tauhat_{unadj}$.

\end{proof}

\delineate

\begin{theorem}\label{supp thm: calibration beats the gOB}
Under the regularity conditions of Theorem~\ref{supp thm: reg adj does no harm} and for a given set of prediction functions $\hat{\mu}_0$ and $\hat{\mu}_1$, the calibrated estimator $N^{1/2}\left(\tauhat_{cal} - \SATE\right)$ has an asymptotic variance that is no larger than that of both $N^{1/2}\left(\tauhat_{gOB} - \SATE\right)$ and $N^{1/2}\left(\tauhat_{GBcal} - \SATE\right)$.
\end{theorem}
\begin{proof}
    By Theorem~3 of \citet{generalizedOB}, $N^{1/2}\left(\tauhat_{gOB} - \SATE\right)$ differs by $o_{P}(1)$ from the difference in means statistic
    \begin{equation}\label{eqn: residual formulation of gOB}
        \frac{1}{n_{1}}\sum_{i \st Z_{i} = 1}\dot{\epsilon}_{i}(1) - \frac{1}{n_{0}}\sum_{i \st Z_{i} = 0}\dot{\epsilon}_{i}(0)
    \end{equation}
    where $\dot{\epsilon}_{i}(z) = y_{i}(z) - \overdotmu_{z}(x_{i})$ for $z \in \{0, 1\}$.
    
    In contrast, by \eqref{eqn: linear reformulation treated} and \eqref{eqn: linear reformulation control}, $N^{1/2}\left(\tauhat_{cal} - \SATE\right)$ differs by $o_{P}(1)$ from the difference in means statistic
    \begin{equation}\label{eqn: residual formulation of cal}
        \frac{1}{n_{1}}\sum_{i \st Z_{i} = 1}\dot{\epsilon}_{i}(1) - \frac{1}{n_{0}}\sum_{i \st Z_{i} = 0}\dot{\epsilon}_{i}(0).
    \end{equation}
    
    Unpacking the notation of \eqref{eqn: residual formulation of gOB}
    \begin{align}
        N^{1/2}\left(\tauhat_{gOB} - \SATE\right) &= \frac{1}{n_{1}}\sum_{i \st Z_{i} = 1}\left(y_{i}(1) - \overdotmu_{1}(x_{i})\right) - \frac{1}{n_{0}}\sum_{i \st Z_{i} = 0}\left(y_{i}(0) - \overdotmu_{0}(x_{i})\right) + o_{P}(1)\nonumber\\
        &= \begin{multlined}[t][10.5cm] \frac{1}{n_{1}}\sum_{i \st Z_{i} = 1}\left(y_{i}(1) - \mathbf{e}_{2}^{\T}\begin{bmatrix} \overdotmu_{0}(x_{i}) \\ \overdotmu_{1}(x_{i}) \end{bmatrix}\right) - \\ \frac{1}{n_{0}}\sum_{i \st Z_{i} = 0}\left(y_{i}(0) - \mathbf{e}_{1}^{\T}\begin{bmatrix} \overdotmu_{0}(x_{i}) \\ \overdotmu_{1}(x_{i}) \end{bmatrix}\right) + o_{P}(1) \label{eqn: linear prediction decomp of gOB}
        \end{multlined}
    \end{align}
    where $\mathbf{e}_{i}$ is the $i$\textsuperscript{th} standard basis vector.  The same logic applied to \eqref{eqn: residual formulation of cal} yields that
    \begin{align}
        N^{1/2}\left(\tauhat_{cal} - \SATE\right) &= \frac{1}{n_{1}}\sum_{i \st Z_{i} = 1}\left(y_{i}(1) - \overdotmu_{OLS, 1}(x_{i})\right) - \frac{1}{n_{0}}\sum_{i \st Z_{i} = 0}\left(y_{i}(0) - \overdotmu_{OLS, 0}(x_{i})\right) + o_{P}(1)\nonumber\\
        &= \begin{multlined}[t][10.5cm] \frac{1}{n_{1}}\sum_{i \st Z_{i} = 1}\left\{y_{i}(1) - \left(\dot{\beta}_{0} + \dot{\beta}_{1}^{\T}\begin{bmatrix} \overdotmu_{0}(x_{i}) \\ \overdotmu_{1}(x_{i}) \end{bmatrix}\right)\right\} - \\ \frac{1}{n_{0}}\sum_{i \st Z_{i} = 0}\left\{y_{i}(0) - \left( \dot{B}_{0} + \dot{B}_{1}^{\T}\begin{bmatrix} \overdotmu_{0}(x_{i}) \\ \overdotmu_{1}(x_{i}) \end{bmatrix}\right)\right\} + o_{P}(1) \label{eqn: linear prediction decomp of cal}
        \end{multlined}
    \end{align}
    where $\dot{\beta}_{0}$ is the intercept of $\overdotmu_{OLS, 1}(\cdot)$,  $\dot{\beta}_{1}$ is the slope of $\overdotmu_{OLS, 1}(\cdot)$, and $(\dot{B}_{0}, \dot{B}_{1})$ are defined similarly for $\overdotmu_{OLS, 0}(\cdot)$.
    
    
    
    Notice that both \eqref{eqn: linear prediction decomp of gOB} and \eqref{eqn: linear prediction decomp of cal} are, up to an $o_{P}(1)$ difference, of the form
    \begin{equation}\label{eqn: general linear prediction form}
         \frac{1}{n_{1}}\sum_{i \st Z_{i} = 1}\left\{y_{i}(1) - \left( \beta_{0} + \beta_{1}^{\T}\begin{bmatrix} \overdotmu_{0}(x_{i}) \\ \overdotmu_{1}(x_{i}) \end{bmatrix}\right)\right\} - \\ \frac{1}{n_{0}}\sum_{i \st Z_{i} = 0}\left\{y_{i}(0) - \left( B_{0} + B_{1}^{\T}\begin{bmatrix} \overdotmu_{0}(x_{i}) \\ \overdotmu_{1}(x_{i}) \end{bmatrix}\right)\right\},
    \end{equation}
    with the only difference being that in \eqref{eqn: linear prediction decomp of gOB} we take $\left(\beta_{1}, B_{1}\right) = \left(\mathbf{e}_{2}, \mathbf{e}_{1}\right)$ and $(\beta_{0}, B_{0}) = (0,0)$ whereas in \eqref{eqn: linear prediction decomp of cal} we take $\left(\beta_{1}, B_{1}\right) = \left(\dot{\beta}_{1}, \dot{B}_{1}\right)$ and $(\beta_{0}, B_{0}) = (\dot{\beta}_{0}, \dot{B}_{0})$.  Since, the $o_{P}(1)$ term contributes nothing to the asymptotic variance of either quantity of interest, we neglect it in the remainder of our analysis.  
    
    By the argument presented in Section~4.1 of \citet{lin13} and Lemma~\ref{lem: regression minimizes asymptotic variance of DiM} the variance of \eqref{eqn: general linear prediction form} is minimized when $\left(\beta_{0}, \beta_{1}, B_{0}, B_{1}\right)$ is taken to be the population ordinary least squares linear regression intercepts and slopes.  This is exactly the case for \eqref{eqn: linear prediction decomp of cal}, whereas \eqref{eqn: linear prediction decomp of gOB} presents a feasible, but not necessarily optimal, solution to the ordinary least squares problem.  Consequently, the asymptotic variance of  $N^{1/2}\left(\tauhat_{cal} - \SATE\right)$ does not exceed that of $N^{1/2}\left(\tauhat_{gOB} - \SATE\right)$ and the inequality is strict when the population least squares problem has a unique optimal solution which does not degenerate to $|\tauhat_{cal} - \tauhat_{gOB}| = o_{p}(N^{-1/2})$.
    
    The same style of proof also works to show that  the asymptotic variance of  $N^{1/2}\left(\tauhat_{cal} - \SATE\right)$ does not exceed that of $N^{1/2}\left(\tauhat_{GBcal} - \SATE\right)$.
\end{proof}

\delineate

\renewcommand{\thetheorem}{A.\arabic{theorem}}
Here we restate Theorem~\ref{supp thm: improves cal and lin} with precise regularity conditions and provide its proof.
\setcounter{theorem}{0}
\begin{theorem}
    Assume the regularity conditions of Theorem~\ref{supp thm: reg adj does no harm}.  So long as the random vectors $\left(f(x_{i})^{\T},\muhat_{0}(x_{i}), \muhat_{1}(x_{i})\right)^{\T}$ satisfy Assumptions~\ref{supp asm: means and covs stabilize} and \ref{supp asm: bounded fourth moment}, a central limit theorem applies to $N^{1/2}\left(\tauhat_{cal2} - \SATE\right)$ and it is non-inferior to both $N^{1/2}\left(\tauhat_{cal} - \SATE\right)$ and $N^{1/2}\left(\tauhat_{lin} - \SATE\right)$ using the engineered features $f(x_{i})$.
\end{theorem}
\begin{proof}
    The mechanics of proving consistency and central limit behavior for $\tauhat_{cal2}$ mirror those of the proof for Theorem~\ref{supp thm: reg adj does no harm}.  Stability follows from Proposition~\ref{prop: stability inherited} using the adaptation in Remarks~\ref{rem: stability inherited with feature engineering} and \ref{rem: lemma also applies to the second calibrated estimator}.  Vanishing of the error processes for the prediction functions of $\tauhat_{cal2}$ follows from modifying Proposition~\ref{prop: error proces vanishing is inherited by consistent OLS} in the natural way to account for the additional features $f(x_{i})$.
    
    
    The asymptotic variance of $\tauhat_{cal2}$ is the same as the asymptotic variance of Lin's regression adjusted estimator which regresses upon $(f(x_{i})^{\T}, \overdotmu_{0}(x_{i}), \overdotmu_{1}(x_{i}))$; call this estimator $\hat{T}_{f, \overdotmu}$.  Likewise, the asymptotic variance of $\tauhat_{cal}$ is the same as the same as the asymptotic variance of the regression adjusted estimator which regresses upon $(\overdotmu_{0}(x_{i}), \overdotmu_{1}(x_{i}))$; in the proof of Theorem~\ref{supp thm: reg adj does no harm} we called this estimator $\hat{T}_{lin}$.  By Lemma~\ref{lem: regression minimizes asymptotic variance of DiM} the inclusion of the additional features $f(x_{i})$ guarantees that $\hat{T}_{f, \overdotmu}$ is non-inferior to $\hat{T}_{lin}$.  Similarly, the inclusion of the additional features $(\overdotmu_{0}(x_{i}), \overdotmu_{1}(x_{i}))$ guarantees that $\hat{T}_{f, \overdotmu}$ is non-inferior to $\tauhat_{lin}$ using the engineered features $f(x_{i})$.  Both of these statements rest upon the observation that the asymptotic variance of the regression adjusted estimator is never increased by including additional features; this follows easily from Lemma~\ref{lem: regression minimizes asymptotic variance of DiM} by observing that regressions without the added features are equivalent to regressions with the extra features but with coefficients restricted to zero for such features.
\end{proof}


\section{Implementation}\label{sec: code}
In Algorithm~\ref{alg: cal} we include pseudocode for the construction of $\tauhat_{cal}$ in order to facilitate easy implementation.

		\begin{algorithm}[H] 
			\SetAlgoLined
			\KwIn{An observed treatment allocation $Z$, with observed responses $y_{1}(Z_{1}), \ldots, y_{N}(Z_{N})$ and covariates $x_{1}, \ldots, x_{N}$.}
			\KwResult{The calibrated estimator $\tauhat_{cal}$.}
			\textbf{Step 1: Initial predictions}\\
			Train the initial prediction algorithms $\muhat_{0}$ and $\muhat_{1}$\;
			\For{$i = 1, \ldots, N$}{
			    Impute the outcomes $\muhat_{0}(x_{i})$ and $\muhat_{1}(x_{i})$.
			}

            \textbf{Step 2: Calibrate prediction functions}\\
            \For{$z \in \{0, 1\}$}{
                Let $\muhat_{OLS, z}$ be the linear regression of $y_{i}(Z_{i})$ on $(\muhat_{0}(x_{i}), \muhat_{1}(x_{i}))^{\T}$ in the group $\{i \st Z_{i} = z\}$.
            }
			\Return{ $\tauhat_{cal} = N^{-1}\sum_{i = 1}^{N}\left\{\muhat_{OLS, 1}(x_{i}) - \muhat_{OLS, 0}(x_{i}) \right\}$}
			\caption{Computation of the calibrated Oaxaca-Blinder estimator.}\label{alg: cal}
		\end{algorithm}
Additionally, code written in \texttt{R} is available at \url{https://github.com/PeterLCohen/OaxacaBlinderCalibration} to demonstrate computation of $\tauhat_{cal}$ and $\tauhat_{cal2}$ and to provide examples of their use on real data sets.

\section{Rank-Deficiency}\label{sec: rank deficient}
Throughout the proofs above we worked under the assumption that regressions, both those defining $\hat{\beta}$ and $\dot{\beta}$, are not rank-deficient.  From the perspective of \eqref{eqn: useful notation sample beta} and \eqref{eqn: useful notation ``population`` beta} this assumption guaranteed unique solutions to the empirical risk minimization problems undergirding ordinary least squares regression.  If instead these regressions had linear dependencies between their features, these minimization problems would have uncountably many solutions. This would present a mathematical impediment to the style of proofs used in the preceding sections, but would not invalidate the general results.

In the case that the regression problems \eqref{eqn: useful notation sample beta} and \eqref{eqn: useful notation ``population`` beta} have uncountably many solutions, there is no uniquely identified choice of $\hat{\beta}^{(N)}$ or $\dot{\beta}^{(N)}$ to make the consistency statement ${||\hat{\beta}^{(N)} - \dot{\beta}^{(N)}||_{2} = o_{P}(1)}$ make sense.  However, the predictions $\muhat_{OLS, z}(x_{i})$ and $\overdotmu_{OLS, z}(x_{i})$ are still well-defined for all $x_{i}$.  Because of this, the residuals $\dot{\epsilon}_{i}(z)$ remain well defined; and so intuition derived from the asymptotically linear expansions \eqref{eqn: linear reformulation treated} and \eqref{eqn: linear reformulation control} leads one to expect no change to our main results.

In order to rectify the proofs for stability, vanishing error processes
, and non-inferiority to accord with the uncountable solution-space to the ordinary least squares regression problems \eqref{eqn: useful notation sample beta} and \eqref{eqn: useful notation ``population`` beta} one can pick a canonical representative from the class of optimal solutions.  For this, the linear regression coefficients given by the Moore-Penrose pseudoinverse \citep{penrose_1955} suffice; see \citet[Chapter 5.5]{matrixComputation} for the general relationship between rank-deficient least squares and matrix pseudoinverses.  Specifically, let the Moore-Penrose pseudoinverse of a matrix $M$ be $M^{+}$.  Then take the canonical regression coefficients to be 
\begin{equation*}
    \hat{\beta}^{+} = \left(  \hat{U}_{1}^{\T} \hat{U}_{1} \right)^{+} \hat{U}_{1}^{\T}\begin{bmatrix}y_{i_{1}}(1) \\ \vdots \\ y_{i_{n_1}}(1) \end{bmatrix} \; \text{and} \; \dot{\beta}^{+} = \left(  \dot{U}_{1}^{\T} \dot{U}_{1} \right)^{+} \dot{U}_{1}^{\T}\begin{bmatrix}y_{1}(1) \\ \vdots \\ y_{N}(1) \end{bmatrix}.
\end{equation*}

The proofs for stability, vanishing error processes
, and non-inferiority follow through as before but with $\hat{\beta}^{+}$ and  $\dot{\beta}^{+}$ replacing $\hat{\beta}$ and  $\dot{\beta}$, respectively.  A minor technical consideration is required: the continuous mapping theorem must be applicable as it was in Lemma~\ref{lem: consistency of OLS coefficients}.  Unlike the matrix inverse map $M \mapsto M^{-1}$, the Moore-Penrose pseudoinverse map $M \mapsto M^{+}$ is not a continuous map over the positive semidefinite cone under the metric induced by the Frobenius norm \citep{discontinuityMPinverse}.  The discontinuity of $M \mapsto M^{+}$ can be resolved by stipulating that the there exists some $\aleph \in \N$ such that (almost surely) $\rank\left(\hat{U}_{1}^{\T} \hat{U}_{1} \right)$ and $\rank\left(\dot{U}_{1}^{\T} \dot{U}_{1} \right)$ are equal to a common constant for all populations of size $N \geq \aleph$ \citep{discontinuityResolved}; then the continuous mapping theorem for general metric spaces of \citet[Theorem 18.11]{asymptoticStats_vdv} applies as needed to carry through the rest of the proofs in the same style as before.  

\color{black}
\section{Variance Estimation and Inference under the Finite Population Model}
    The asymptotic Gaussianity of $N^{1/2}(\tauhat_{cal} - \SATE)$ facilitates inference for the sample average treatment effect using $\hat{\tau}_{cal}$.  To proceed, we provide an asymptotically conservative variance estimator for $\tauhat_{cal}$.  Define 
    \begin{align*}
        \hat{\Sigma}_{z} &:= \frac{1}{n_{z} - 1}\sum_{i \st Z_{i} = z}\left(y_{i}(z) - \muhat_{z}(x_{i})\right)^{2},\\
        \hat{V} &= \frac{\hat{\Sigma}_{1}}{n_{1}} + \frac{\hat{\Sigma}_{0}}{n_{0}}.
    \end{align*}
    By Theorem 4 of \citet{generalizedOB}, $\hat{V}$ provides an asymptotically conservative estimate of the variance of the generalized Oaxaca-Blinder estimator using imputation functions $\muhat_{0}$ and $\muhat_{1}$; their result holds for the finite population model and accounts only for variability arising from the stochasticity of the treatment allocation process.  Consequently, the variance of the calibrated estimator, which is a particular form of generalized Oaxaca-Blinder estimator, is estimated by
    \begin{align}
        \hat{\Sigma}_{z, cal} &:= \frac{1}{n_{z} - 1}\sum_{i \st Z_{i} = z}\left(y_{i}(z) - \muhat_{OLS, z}(x_{i})\right)^{2},\nonumber\\
        \hat{V}_{cal} &= \frac{\hat{\Sigma}_{1, cal}}{n_{1}} + \frac{\hat{\Sigma}_{0, cal}}{n_{0}}.\label{eqn: fin pop variance estimator}
    \end{align}
    In a finite population $\hat{V}_{cal}$ is potentially asymptotically conservative relative to the true variance of $\tauhat_{cal}$, and so asymptotically valid -- though potentially conservative -- confidence intervals may be formed based upon the usual normal approximation.  In Section~\ref{sec: Linear Calibration in Alternative Models} we discuss variance estimators which account for stochasticity in the potential outcomes themselves. 
\color{black}

\section{Linear Calibration in Alternative Models}\label{sec: Linear Calibration in Alternative Models}
\subsection{Superpopulation and Fixed-Covariate Models}
The asymptotic non-inferiority of linearly calibrated generalized Oaxaca-Blinder estimators is not tied to the finite population model. Here we detail two alternative common models for which linear calibration guarantees asymptotic non-inferiority of generalized Oaxaca-Blinder estimators: the superpopulation model and the fixed-covariate model.  In the superpopulation model the vector of potential outcomes and covariates  $(y_{i}(0), y_{i}(1), x_{i})$ is an independent and identically distributed draw from some fixed distribution for each $i \in \{1, \ldots, N\}$.  In the fixed-covariate model the $i$th unit's covariates $x_{i}$ are deterministic while the potential outcomes $(y_{i}(0), y_{i}(1))$ are independent draws from some conditional distribution given $x_{i}$.  In the superpopulation framework, the targeted estimand is the population average treatment effect (\textsc{PATE}), $\PATE = E[y_i(1) - y_i(0)]$.  Under the fixed covariate model, the estimand of interest is the conditional average treatment effect (\textsc{CATE}), $\CATE = N^{-1}\sum_{i=1}^NE[y_i(1)-y_i(0)\mid x_i]$.  In both of these models linear calibration functions exactly the same as in the finite population model: estimate the nonlinear prediction functions $\muhat_{0}$ and $\muhat_{1}$ on the observed data and subsequently form \citeauthor{lin13}'s linearly-adjusted estimator using the imputed values $\muhat_{0}(x_{i})$ and $\muhat_{1}(x_{i})$ in lieu of the covariates $x_{i}$.  The natural analogs of Theorem~1 hold so long as the required assumptions are modified to reflect the chosen probabilistic framework.

\subsection{Linear Calibration for the Population Average Treatment Effect}\label{sec: superpop}
Under the superpopulation model, Assumption~\ref{supp asm: stability} (Stability) remains unchanged.  Likewise Assumption~\ref{app asm: error process vanishes} stays the same, though the stochastic process now inherits randomness from the potential outcomes and covariates in addition to the randomness in treatment allocation. Assumptions~\ref{supp asm: means and covs stabilize} and \ref{supp asm: bounded fourth moment} require simple modifications to adapt to randomness in the potential outcomes and covariates.

\begin{assumption}[Superpopulation Limiting Means and Variances]\label{supp asm: superpop means and covs stabilize}
	  The mean vector and covariance matrix of $(y_i(0), y_i(1), \dot{\mu}_0(x_i), \dot{\mu}_0(x_i))^\T$ have limiting values.  For instance, for ${z=0,1}$ there exists a limiting value $\bar{y}(z)_\infty$ such that $\lim_{N\rightarrow \infty} \Expectation{y_i(z)} = \bar{y}(z)_\infty$.
\end{assumption}

\begin{assumption}[Superpopulation Bounded Fourth Moments]\label{supp asm: superpop bounded fourth moment}
     There exists some $C < \infty$ for which, for all ${z = 0,1}$ and all $N$, $\Expectation{y_{i}(z)^{4}} < C$ and $\Expectation{\dot{\mu}_z(x_i)^{4}} < C.$
\end{assumption}

In the superpopulation case, the definition of $\dot{\beta}$ given in Lemma~\ref{lem: consistency of OLS coefficients} requires a minor change to
\begin{align*}
    \dot{\beta} &= (\dot{\beta}_{0}, \dot{\beta}_{1}) = \argmin_{\beta_{0}, \beta_{1}}\Expectation{\left\{y_{i}(1) - \left(\beta_{0} + \beta_{1}^{T}\begin{bmatrix}\dot{\mu}_{0}(x_{i}) \\ \dot{\mu}_{1}(x_{i}) \end{bmatrix}\right) \right\}^{2}}.
\end{align*}
Then $||\hat{\beta} - \dot{\beta}||_{2} = o_{P}(1)$ by the standard consistency of sample ordinary least squares regression coefficients under Assumptions~\ref{supp asm: superpop means and covs stabilize} and \ref{supp asm: superpop bounded fourth moment}; this consistency holds even without assuming the truth of a linear model \citep[Proposition 7]{modelsAsApprox}.

\begin{lemma}[Asymptotically Linear Expansions around the \textsc{SATE}]\label{lem: asymptotic linear expansion around SATE}
    Under Assumptions~\ref{supp asm: non-degen sampling limit}, \ref{supp asm: stability}, and \ref{app asm: error process vanishes} the random variable $N^{-1}\sum_{i = 1}^{N}\left(\muhat_{z}(x_{i}) - y_{i}(z)\right)$ is asymptotically linear in the sense that
    \begin{equation*}
        \frac{1}{N} \sum_{i = 1}^{N}\left(\muhat_{z}(x_{i}) - y_{i}(z)\right) = \frac{1}{n_{z}}\sum_{i \st Z_{i} = z}\dot{\epsilon}_{i}(z) + o_{p}(N^{-1/2})
    \end{equation*}
    where $\dot{\epsilon}_{i}(z) = y_{i}(z) - \overdotmu_{z}(x_i)$.
\end{lemma}
\begin{proof}
    By the exact same reasoning as in \citet[Proof of Theorem 3]{generalizedOB}, rewrite $N^{-1}\sum_{i = 1}^{N}\left(\muhat_{z}(x_{i}) - y_{i}(z)\right) $ as
    \begin{equation*}
        \frac{1}{n_{z}}\sum_{i \st Z_{i} = z}\dot{\epsilon}_{i}(z) + \frac{1}{N}\Bigg(\underbrace{\sum_{i = 1}^{N}\left(\frac{Z_{i}\overdotmu_{z}(x_{i})}{(n_{z} / N)} - \overdotmu_{z}(x_{i})\right)}_{N^{1/2}\G_{N, z}(\overdotmu_{z})} - \underbrace{\sum_{i = 1}^{N}\left(\frac{Z_{i}\muhat_{z}(x_{i})}{(n_{z} / N)} - \muhat_{z}(x_{i})\right)}_{N^{1/2}\G_{N, z}(\muhat_{z})} \Bigg).
    \end{equation*}
    Consequently, the desired result holds so long as $\left|\G_{N, z}(\overdotmu_{z}) -  \G_{N, z}(\muhat_{z})\right| = o_{P}\left(1\right)$; this holds by Assumption~\ref{app asm: error process vanishes}.
\end{proof}

\begin{theorem}\label{supp thm: superpopulation clt for taucal}
    Under the superpopulation model, subject to Assumptions~\ref{supp asm: non-degen sampling limit}, \ref{supp asm: stability}, \ref{app asm: error process vanishes}, \ref{supp asm: superpop means and covs stabilize} and \ref{supp asm: superpop bounded fourth moment}, $N^{1/2}\left(\tauhat_{cal} - \PATE\right)$ obeys a central limit theorem.
\end{theorem}
\begin{proof}
    We start with an algebraic decomposition:
    \begin{align*}
        N^{1/2}\left(\tauhat_{cal} - \PATE\right) &= N^{1/2}\left(\frac{1}{N}\sum_{i = 1}^{N}\left(\muhat_{OLS,1}(x_{i}) - \muhat_{OLS,0}(x_{i}) \right) - \PATE\right)\\
        &= N^{1/2}\Bigg(\frac{1}{N}\sum_{i = 1}^{N}\left(\muhat_{OLS,1}(x_{i}) - \muhat_{OLS,0}(x_{i}) \right) - \frac{1}{N}\sum_{i = 1}^{N}\left(y_{i}(1) - y_{i}(0) \right) +\\
        &\quad\quad \frac{1}{N}\sum_{i = 1}^{N}\left(y_{i}(1) - y_{i}(0) \right) - \PATE\Bigg)\\
        &= N^{1/2}\Bigg(\frac{1}{N}\sum_{i = 1}^{N}\left(\muhat_{OLS,1}(x_{i}) - y_{i}(1) \right) - \frac{1}{N}\sum_{i = 1}^{N}\left(\muhat_{OLS,0}(x_{i}) - y_{i}(0) \right) +\\
        &\quad\quad \frac{1}{N}\sum_{i = 1}^{N}\left(y_{i}(1) - y_{i}(0) \right) - \PATE\Bigg)\\
    \end{align*}
    By Lemma~\ref{lem: asymptotic linear expansion around SATE} this final term equals
    \begin{equation*}
        N^{1/2}\Bigg(\frac{1}{N}\sum_{i = 1}^{N}Z_{i}Nn_{1}^{-1}\dot{\epsilon}_{i}(1) - \frac{1}{N}\sum_{i = 1}^{N}(1 - Z_{i})Nn_{0}^{-1}\dot{\epsilon}_{i}(0) + \frac{1}{N}\sum_{i = 1}^{N}\left(y_{i}(1) - y_{i}(0) \right) - \PATE\Bigg) + o_{P}(1)
    \end{equation*}
    where $\dot{\epsilon}_{i}(z) = y_{i}(z) - \overdotmu_{OLS,z}(x_{i})$.  By the reasoning of Lemma~\ref{lem: general prediction unbiaseness in the population for superpopulation models} we stipulate that $\Expectation{\dot{\epsilon}_{i}(z)} = 0$ for $z \in \{0, 1\}$.
    Rearranging the formula above yields
    \begin{equation*}
        N^{1/2}\Bigg(\frac{1}{N}\sum_{i = 1}^{N}Z_{i}N\left( n_{1}^{-1}\dot{\epsilon}_{i}(1) + n_{0}^{-1}\dot{\epsilon}_{i}(0) \right) - \frac{1}{N}\sum_{i = 1}^{N}Nn_{0}^{-1}\dot{\epsilon}_{i}(0) + \frac{1}{N}\sum_{i = 1}^{N}\left(y_{i}(1) - y_{i}(0) \right) - \PATE\Bigg) + o_{P}(1)
    \end{equation*}

    Consider the random variable 
    \begin{equation*}
        \varphi(Z_{i}) := Z_{i}N\left( n_{1}^{-1}\dot{\epsilon}_{i}(1) + n_{0}^{-1}\dot{\epsilon}_{i}(0) \right) - Nn_{0}^{-1}\dot{\epsilon}_{i}(0) + \left(y_{i}(1) - y_{i}(0) \right);
    \end{equation*}
    to prove the desired result it suffices to show a central limit theorem for $\varphi(Z_{i})$.  To do this we will start by conditioning upon the event $Z = \tilde{z}$ for some $\tilde{z} \in \Omega_{CRE}$; we will show that a central limit theorem applies to $\varphi(\tilde{z}_{i})$ almost surely with respect to the conditioning event $Z = \tilde{z} \in \Omega_{CRE}$; then we will leverage this conditional central limit theorem to deduce that an unconditional central limit theorem applies to $\varphi(Z_{i})$.

    Condition on the event $Z = \tilde{z}$ for some $\tilde{z} \in \Omega_{CRE}$ and recall that $Z_{i}$ is independent of $(y_{i}(0), y_{i}(1), x_{i})$ so this conditioning changes nothing of the distribution of the $\dot{\epsilon}_{i}(z)$ and the $y_{i}(z)$ for all $i \in \{1, \ldots, N\}$ and $z \in \{0, 1\}$.  We examine the conditional expectation $\Expectation{\varphi(Z_{i}) \given Z = \tilde{z}}$.

    \begin{align*}
        \Expectation{\varphi(Z_{i}) \given Z = \tilde{z}} &= \Expectation{Z_{i}N\left( n_{1}^{-1}\dot{\epsilon}_{i}(1) + n_{0}^{-1}\dot{\epsilon}_{i}(0) \right)\given Z = \tilde{z}} - \Expectation{Nn_{0}^{-1}\dot{\epsilon}_{i}(0)\given Z = \tilde{z}} +\\
        &\hspace{1in}  \Expectation{\left(y_{i}(1) - y_{i}(0) \right)\given Z = \tilde{z}}\\
        &=\tilde{z}_{i}\underbrace{\Expectation{N\left( n_{1}^{-1}\dot{\epsilon}_{i}(1) + n_{0}^{-1}\dot{\epsilon}_{i}(0) \right)}}_{ 0} - \underbrace{\Expectation{Nn_{0}^{-1}\dot{\epsilon}_{i}(0)}}_{0} +   \underbrace{\Expectation{\left(y_{i}(1) - y_{i}(0) \right)}}_{ \PATE}\\
        &= \PATE.
    \end{align*}

    The random variables $\varphi(Z_{1}), \ldots, \varphi(Z_{n})$ are a conditionally independent family of random variables given the event $Z = \tilde{z}$; however, they are not identically distributed.  For all $i$ such that $\tilde{z}_{i} = 0$ the random variable $\varphi(\tilde{z}_{i})$ is equal in distribution to $ - Nn_{0}^{-1}\dot{\epsilon}_{i}(0) + \left(y_{i}(1) - y_{i}(0) \right)$.  For the $i$ such that $\tilde{z}_{i} = 1$ the random variable $\varphi(\tilde{z}_{i})$ is equal in distribution to $N\left( n_{1}^{-1}\dot{\epsilon}_{i}(1) + n_{0}^{-1}\dot{\epsilon}_{i}(0) \right) - Nn_{0}^{-1}\dot{\epsilon}_{i}(0) + \left(y_{i}(1) - y_{i}(0) \right)$.  Motivated by this, we study
   \begin{align*}
       N^{-1/2}\left(\sum_{i = 1}^{N}\varphi(\tilde{z}_{i}) - \PATE \right) &= N^{-1/2}\left(\sum_{i = 1}^{N}\varphi(\tilde{z}_{i}) - \Expectation{\varphi(Z_{i}) \given Z = \tilde{z}} \right)\\
       &=\left(N^{-1/2}n_{0}^{1/2}\right)n_{0}^{-1/2}\left(\sum_{i \st \tilde{z}_{i} = 0}\varphi(\tilde{z}_{i}) - \Expectation{\varphi(Z_{i}) \given Z = \tilde{z}} \right) +\\
       &\hspace{1in}\left(N^{-1/2}n_{1}^{1/2}\right)n_{1}^{-1/2}\left(\sum_{i \st \tilde{z}_{i} = 1}\varphi(\tilde{z}_{i}) - \Expectation{\varphi(Z_{i}) \given Z = \tilde{z}} \right)
   \end{align*}
    The first term on the right-hand-side is independent from the second term on the right-hand-side; furthermore:
    \begin{itemize}
        \item The first term is a $n_{0}^{-1/2}$-scaled sum of $n_{0}$ independent terms each equal in distribution to $ - Nn_{0}^{-1}\dot{\epsilon}_{i}(0) + \left(y_{i}(1) - y_{i}(0) \right)$;
        \item The second term is a $n_{1}^{-1/2}$-scaled sum of $n_{1}$ independent terms each equal in distribution to $N\left( n_{1}^{-1}\dot{\epsilon}_{i}(1) + n_{0}^{-1}\dot{\epsilon}_{i}(0) \right) - Nn_{0}^{-1}\dot{\epsilon}_{i}(0) + \left(y_{i}(1) - y_{i}(0) \right)$.
    \end{itemize}

    Recall that $n_{1}/N \rightarrow p$ and $n_{0}/N \rightarrow (1 - p)$; then apply the Lindeberg–L\'{e}vy central limit theorem \citep[Theorem 3.4.1]{durrett} to the two terms $n_{0}^{-1/2}\left(\sum_{i \st \tilde{z}_{i} = 0}\varphi(\tilde{z}_{i}) - \Expectation{\varphi(Z_{i}) \given Z = \tilde{z}} \right)$ and $n_{1}^{-1/2}\left(\sum_{i \st \tilde{z}_{i} = 1}\varphi(\tilde{z}_{i}) - \Expectation{\varphi(Z_{i}) \given Z = \tilde{z}} \right)$ separately.  Denote the limiting variance of the first term by $s_{0}$ and the limiting variance of the second term by $s_{1}$ and notice that these quantities do not depend on the particular choice of $\tilde{z}$.  Then, regardless of $\tilde{z}$'s value we have that $N^{-1/2}\left(\sum_{i = 1}^{N}\varphi(\tilde{z}_{i}) - \Expectation{\varphi(Z_{i}) \given Z = \tilde{z}} \right)$ converges weakly to the random variable $(1 - p)^{-1/2}A + p^{-1/2}B$ where $A \sim \Normal{0}{s_{0}}$ independent of $B \sim \Normal{0}{s_{1}}$.  By the independence of $A$ and $B$ we have that, for all $\tilde{z} \in \Omega_{CRE}$,
    \begin{equation*}
        N^{-1/2}\left(\sum_{i = 1}^{N}\varphi(\tilde{z}_{i}) - \Expectation{\varphi(Z_{i}) \given Z = \tilde{z}} \right) \text{ converges in distribution to } \Normal{0}{\frac{s_{0}}{1 - p} + \frac{s_{1}}{p}}.
    \end{equation*}

    In other words, $N^{-1/2}\left(\sum_{i = 1}^{N}\varphi(Z_{i}) - \PATE \right)$ conditional upon $Z$ converges weakly to $\Normal{0}{(1-p)^{-1}s_{0} + p^{-1}s_{1}}$ almost surely with respect to randomness in $Z$.  Consequently, by Lemma~\ref{lem: conditional to unconditional convergence}, $N^{-1/2}\left(\sum_{i = 1}^{N}\varphi(Z_{i}) - \PATE \right)$ converges in distribution to $\Normal{0}{(1-p)^{-1}s_{0} + p^{-1}s_{1}}$ unconditionally on $Z$.  Unwinding the definition of $\varphi(Z_{i})$ establishes that $N^{1/2}\left(\tauhat_{cal} - \PATE\right)$ obeys a central limit theorem.

\end{proof}

    \begin{theorem}\label{supp thm: superpopulation noninferiority}
        Assume the conditions of Theorem~\ref{supp thm: superpopulation clt for taucal} hold.  Further suppose that the original prediction functions $\muhat_{0}$ and $\muhat_{1}$ are prediction unbiased.  Then $N^{1/2}(\tauhat_{cal} - \PATE)$, $N^{1/2}(\tauhat_{GBcal} - \PATE)$, $N^{1/2}(\tauhat_{gOB} - \PATE)$, and $N^{1/2}(\tauhat_{unadj} - \PATE)$ all obey central limit theorems and $N^{1/2}(\tauhat_{cal} - \PATE)$ has asymptotic variance no greater than any of the other three.
    \end{theorem}
    \begin{proof}
        The central limit theorem for $N^{1/2}(\tauhat_{cal} - \PATE)$ is proved in Theorem~\ref{supp thm: superpopulation clt for taucal}.  The proof for the central limit theorems of $N^{1/2}(\tauhat_{GBcal} - \PATE)$ follows a similar arc to that of Theorem~\ref{supp thm: superpopulation clt for taucal}; the proof of the central limit theorem for $N^{1/2}(\tauhat_{gOB} - \PATE)$ also follows similar reasoning, but relies upon the assumption that the original prediction functions $\muhat_{0}$ and $\muhat_{1}$ are prediction unbiased.  The assumption that the original prediction functions $\muhat_{0}$ and $\muhat_{1}$ are prediction unbiased is not needed for $N^{1/2}(\tauhat_{cal} - \PATE)$ and $N^{1/2}(\tauhat_{GBcal} - \PATE)$ since these two calibration procedures automatically confer prediction unbiasedness due to the first-order optimality conditions of linear regression.  Finally, the central limit theorem for $N^{1/2}(\tauhat_{unadj} - \PATE)$ is a classical consequence of the Lyapunov central limit theorem.

        By the law of total variance, for any measurable event $\mathcal{F}$, the variance of $N^{1/2}(\tauhat_{cal} - \PATE)$ decomposes as
        \begin{multline*}
            \Variance{N^{1/2}(\tauhat_{cal} - \PATE)} = \Expectation{\Variance{N^{1/2}(\tauhat_{cal} - \PATE) \given \mathcal{F}}} +\\ \Variance{\Expectation{N^{1/2}(\tauhat_{cal} - \PATE) \given \mathcal{F}}}.
        \end{multline*}
        Taking $\mathcal{F}$ to be the event $\{(y_{i}(0), y_{i}(1), x_{i}) = (\mathbf{y}_{i}(0), \mathbf{y}_{i}(1), \mathbf{x}_{i})  \, i = 1, \ldots, N\}$ yields the decomposition of the variance of $N^{1/2}(\tauhat_{cal} - \PATE)$ into the expected variance of $N^{1/2}(\tauhat_{cal} - \PATE)$ given a fixed finite population and the variance of the expectation of $N^{1/2}(\tauhat_{cal} - \PATE)$ given that fixed finite population.  Since $\tauhat_{cal} - \SATE$ conditioned on $\mathcal{F}$ has mean on the order of $o_{P}(N^{-1/2})$ it follows that $\lim_{N \rightarrow \infty}\Variance{N^{1/2}(\tauhat_{cal} - \PATE)}$ equals
        \begin{equation}\label{eqn: decomp of variance}
             \underbrace{\lim_{N \rightarrow \infty}\Expectation{\Variance{N^{1/2}(\tauhat_{cal} - \PATE) \given \mathcal{F}}}}_{\text{Term 1}} + \underbrace{\lim_{N \rightarrow \infty}\Variance{N^{1/2}(\SATE - \PATE) }}_{\text{Term 2}}.
        \end{equation}
        Term 1 is the limiting expected finite population variance of $N^{1/2}(\tauhat_{cal} - \PATE)$ while Term 2 is the limiting variance of $N^{1/2}(\SATE - \PATE)$.  In other words, Term 2 is the  asymptotic $N$-scaled variance of the \textsc{SATE}\, around the \textsc{PATE}.  The same variance decomposition idea applies to $N^{1/2}(\tauhat_{GBcal} - \PATE)$, $N^{1/2}(\tauhat_{gOB} - \PATE)$, and $N^{1/2}(\tauhat_{unadj} - \PATE)$; for each different estimator the form of Term 1 adapts to the particular estimator at hand but Term 2 is exactly the same.  Thus, the differences in asymptotic variance under the superpopulation model are controlled only by the differences in finite population variance.  Consequently, the finite population analysis already performed implies that $N^{1/2}(\tauhat_{cal} - \PATE)$ has asymptotic variance no greater than that of $N^{1/2}(\tauhat_{GBcal} - \PATE)$, $N^{1/2}(\tauhat_{gOB} - \PATE)$, and $N^{1/2}(\tauhat_{unadj} - \PATE)$.
    \end{proof}
    
    \color{black}
    The decomposition of variance in \eqref{eqn: decomp of variance} points to an important deficiency of $\hat{V}_{cal}$ of \eqref{eqn: fin pop variance estimator}: while the variance estimator $\hat{V}_{cal}$ of \eqref{eqn: fin pop variance estimator} is guaranteed to be asymptotically conservative in the finite population model, it need not be valid if additional randomness is incorporated into the data generating process.  Indeed, if there is randomness in the potential outcomes themselves, $N\hat{V}_{cal}$ may converge in probability to a constant which is strictly smaller than the limiting variance of $N^{1/2}(\tauhat_{cal} - N^{-1}\sum_{i = 1}^{N}\Expectation{y_{i}(1) - y_{i}(0)})$.  Fortunately, \eqref{eqn: decomp of variance} also suggests a simple rectification of this anti-conservativeness.  By asymptotic linearity and the usual finite population decomposition of variance \citep{decompTreatmentEffectVar}, we have that
    \begin{align*}
        \lim_{N \rightarrow \infty}&\Expectation{\Variance{N^{1/2}(\tauhat_{cal} - \PATE) \given \mathcal{F}}} = \frac{\Sigma_{\dot{\epsilon}(1), \infty}}{p} + \frac{\Sigma_{\dot{\epsilon}(0), \infty}}{1 - p} - \Sigma_{\delta, \infty},\\
        \Sigma_{\dot{\epsilon}(z), \infty} &= \lim_{N \rightarrow \infty} \frac{1}{N - 1}\sum_{i = 1}^{N}\Expectation{\left(\dot{\epsilon}_{i}(z) - \frac{1}{N}\sum_{j = 1}^{N}\Expectation{\dot{\epsilon}_{j}(z)}  \right)^{2}} \quad \text{for } z \in \{0, 1\},\\
        \Sigma_{\delta, \infty} &= \lim_{N \rightarrow \infty} \frac{1}{N - 1}\sum_{i = 1}^{N}\Expectation{\left(\left(\dot{\epsilon}_{i}(1) - \dot{\epsilon}_{i}(0)\right) - \frac{1}{N}\sum_{j = 1}^{N}\Expectation{\dot{\epsilon}_{j}(1) - \dot{\epsilon}_{j}(0)}  \right)^{2}}.
    \end{align*}
    Furthermore,
    \begin{align*}
        \lim_{N \rightarrow \infty}&\Variance{N^{1/2}(\SATE - \PATE) } = \Sigma_{\tau, \infty},\\
        \Sigma_{\tau, \infty} &= \lim_{N \rightarrow \infty}\frac{1}{N - 1}\sum_{i = 1}^{N}\Expectation{\left(\left(y_{i}(1) - y_{i}(0)\right) - \frac{1}{N}\sum_{j = 1}^{N}\Expectation{y_{j}(1) - y_{j}(0)}  \right)^{2}}.
    \end{align*}
    Combining the two previous observations with \eqref{eqn: decomp of variance} yields that
    \begin{equation*}
        \lim_{N \rightarrow \infty}\Variance{N^{1/2}(\tauhat_{cal} - \PATE)} = \frac{\Sigma_{\dot{\epsilon}(1), \infty}}{p} + \frac{\Sigma_{\dot{\epsilon}(0), \infty}}{1 - p} - \Sigma_{\delta, \infty} + \Sigma_{\tau, \infty}.
    \end{equation*}
    
    By the classical partitioning of variance in linear regression the total sum of squares is the sum of the residual sum of squares and the explained sum of squares; in our context this means that 
    \begin{align*}
        \Sigma_{\tau, \infty} &= \Sigma_{\delta, \infty} + \Sigma_{fitted, \infty},\\
        \Sigma_{fitted, \infty} &= \lim_{N \rightarrow \infty}\frac{1}{N - 1}\sum_{i = 1}^{N}\Expectation{\left(\left(\overdotmu_{OLS, 1}(x_{i}) - \overdotmu_{OLS, 0}(x_{i})\right) - \frac{1}{N}\sum_{j = 1}^{N}\Expectation{\overdotmu_{OLS, 1}(x_{j}) - \overdotmu_{OLS, 0}(x_{j})} \right)^{2}}.
    \end{align*}
    In total, we have that 
    \begin{equation*}
        \lim_{N \rightarrow \infty}\Variance{N^{1/2}(\tauhat_{cal} - \PATE)} = \frac{\Sigma_{\dot{\epsilon}(1), \infty}}{p} + \frac{\Sigma_{\dot{\epsilon}(0), \infty}}{1 - p} + \Sigma_{fitted, \infty}.
    \end{equation*}
    
    Consequently, to adapt the variance estimator of \eqref{eqn: fin pop variance estimator} to superpopulation inference, one must incorporate the variance of the predicted treatment effects; this forms the superpopulation variance estimator
    \begin{multline}
        \hat{V}_{cal, sup} = \frac{\hat{\Sigma}_{1, cal}}{n_{1}} + \frac{\hat{\Sigma}_{0, cal}}{n_{0}} + \\
        \frac{1}{N - 1}\sum_{i = 1}^{N}\left(\left(\muhat_{OLS, 1}(x_{i}) - \muhat_{OLS, 0}(x_{i})\right) - \frac{1}{N}\sum_{j = 1}^{N}\left(\muhat_{OLS, 1}(x_{j}) - \muhat_{OLS, 0}(x_{j})\right) \right)^{2}.\label{eqn: superpop variance estimator}
    \end{multline}

    Notice that the decomposition of variance into $p^{-1}\Sigma_{\dot{\epsilon}(1), \infty} + (1 - p)^{-1}\Sigma_{\dot{\epsilon}(0), \infty} + \Sigma_{fitted, \infty}$ relies upon the orthogonality of residuals $y_{i}(z) - \overdotmu_{z}(x_{i})$ and predicted values $\overdotmu_{z}(x_{i})$ due to the first order optimality condition of ordinary least squares linear regression, so such a decomposition is not generally applicable to uncalibrated estimators; this provides another attractive property for calibrated estimators.\footnote{For a generic superpopulation variance estimator in the context of uncalibrated imputation estimators see \citet[Section 3.3]{rothe}.}
    \color{black}

\subsection{Linear Calibration for the Conditional Average Treatment Effect}

We start by presenting several highly general regularity conditions; further on we present an example of a generative model which satisfies the required regularity conditions.  Our regularity conditions in the fixed-covariate model are presented with respect to conditioning upon the potential outcomes in addition to the already implicitly determined covariates.  For each population of size $N$ with deterministic covariates $\{x_{i}\}_{i = 1}^{N}$, consider conditioning upon some realization of the potential outcomes 
\begin{equation}\label{eqn: fixed cov conditioning on outcomes}
    \left\{(y_{i}(0), y_{i}(1)) = (\mathbf{y}_{i}(0), \mathbf{y}_{i}(1))  \, i = 1, \ldots, N\right\}.
\end{equation}  

\begin{assumption}[Fixed-Covariate Limiting Means and Variances]\label{supp asm: fixed cov means and covs stabilize}
    For all conditioning events of the form \eqref{eqn: fixed cov conditioning on outcomes} except for a set of measure zero under the fixed covariate model we require that for ${z=0,1}$ there exists a limiting value $\bar{y}(z)_\infty$ such that $\lim_{N\rightarrow \infty} N^{-1}\sum_{i = 1}^{N}y_i(z) = \bar{y}(z)_\infty$.  Likewise, for almost all conditioning events of the form \eqref{eqn: fixed cov conditioning on outcomes} there exists a common limiting positive semidefinite matrix $\Sigma$ such that
    \begin{multline}\label{eqn: conditional cov mat limit}
        \lim_{N \rightarrow \infty}\frac{1}{N - 1}\sum_{i = 1}^{N}\left(\begin{bmatrix}
        y_{i}(0) \\ y_{i}(1) \\ \overdotmu_{0}(x_{i}) \\ \overdotmu_{1}(x_{i})
        \end{bmatrix} - N^{-1}\sum_{j = 1}^{N}\begin{bmatrix}
        y_{j}(0) \\ y_{j}(1) \\ \overdotmu_{0}(x_{j}) \\ \overdotmu_{1}(x_{j})
        \end{bmatrix} \right) \tensor \\
         \left(\begin{bmatrix}
        y_{i}(0) \\ y_{i}(1) \\ \overdotmu_{0}(x_{i}) \\ \overdotmu_{1}(x_{i})
        \end{bmatrix} - N^{-1}\sum_{j = 1}^{N}\begin{bmatrix}
        y_{j}(0) \\ y_{j}(1) \\ \overdotmu_{0}(x_{j}) \\ \overdotmu_{1}(x_{j})
        \end{bmatrix} \right)  = \Sigma.
    \end{multline}
    
\end{assumption}

\begin{assumption}[Fixed-Covariate Bounded Fourth Moments]\label{supp asm: fixed cov bounded fourth moment}
    For all conditioning events of the form \eqref{eqn: fixed cov conditioning on outcomes} except for a set of measure zero under the fixed covariate model we require that for ${z=0,1}$ there exists some $C < \infty$ such that for all $N \in \N$, $N^{-1}\sum_{i = 1}^{N}y_{i}(z)^{4} < C$ and $N^{-1}\sum_{i = 1}^{N}\overdotmu_z(x_i)^{4} < C.$

\end{assumption}

In the fixed-covariate case, the definition of $\dot{\beta}$ given in Lemma~\ref{lem: consistency of OLS coefficients} requires a minor change to
\begin{align*}
    \dot{\beta} &= (\dot{\beta}_{0}, \dot{\beta}_{1}) = \argmin_{\beta_{0}, \beta_{1}}\sum_{i = 1}^{N}\Expectation{\left\{y_{i}(1) - \left(\beta_{0} + \beta_{1}^{T}\begin{bmatrix}\dot{\mu}_{0}(x_{i}) \\ \dot{\mu}_{1}(x_{i}) \end{bmatrix}\right) \right\}^{2} \; \bigg| \; x_{i}}.
\end{align*}
Assumptions~\ref{supp asm: fixed cov means and covs stabilize} and \ref{supp asm: fixed cov bounded fourth moment} facilitate consistency of ordinary least squares coefficients in the fixed covariate model.

\begin{lemma}[Lindeberg Condition]\label{lem: Lindeberg condition}
    Under Assumptions~\ref{supp asm: fixed cov means and covs stabilize} and \ref{supp asm: fixed cov bounded fourth moment}, the potential outcomes $(y_{i}(0), y_{i}(1))$ given $x_{i}$ jointly satisfy the conditions of Lindeberg's central limit theorem.
\end{lemma}
\begin{proof}
    Define $s_{N}^{2}(z) = \sum_{i =1}^{N}\Variance{y_{i}(z) \given x_{i}}$ where $\Variance{y_{i}(z) \given x_{i}}$ denotes the variance of $y_{i}(z)$ given the covariates $x_{i}$.  We will show that Lyapounov's condition \citep[Equation 11.12]{TestingStatHyp} holds at $\delta = 2$ for the potential outcomes; formally, for $z \in \{0, 1\}$ and $\delta = 2$
    \begin{equation}\label{eqn: lyapounov condition for CATE}
        \lim_{N \rightarrow \infty}\frac{1}{s_{N}^{2 + \delta}(z)}\sum_{i = 1}^{N}\Expectation{\left|y_{i}(z) \right|^{2 + \delta} \given x_{i}} = 0.
    \end{equation}
    
    Rewrite \eqref{eqn: lyapounov condition for CATE} as
    \begin{equation*}
        \lim_{N \rightarrow \infty}\underbrace{\frac{N}{s_{N}^{2 + \delta}(z)}}_{\text{Term 1}}\cdot\underbrace{\frac{1}{N}\sum_{i = 1}^{N}\Expectation{\left|y_{i}(z) \right|^{2 + \delta} \given x_{i}}}_{\text{Term 2}}.
    \end{equation*}
    Term 2 is bounded above by $C$ for all $N$ by Assumption~\ref{supp asm: fixed cov bounded fourth moment} and Kolmogorov's strong law of large numbers for non-identically distributed sequences \textcolor{black}{\citep[Section 10.7]{feller68}}, so it suffices to show that Term 1 vanishes as $N \rightarrow \infty$.  By Assumption~\ref{supp asm: fixed cov means and covs stabilize} and Kolmogorov's strong law, $N^{-1}\sum_{i = 1}^{N}\Variance{y_{i} \given x_{i}}$ limits to a positive constant which we denote $\Sigma_{y(z)}$.\footnote{The use of Kolmogorov's strong law can be replaced by the bounded convergence theorem for both the arguments pertaining to Term 1 and Term 2; Assumption~\ref{supp asm: fixed cov bounded fourth moment} establishes the required bounds.}  Consequently, $Ns_{N}^{-2} \rightarrow \Sigma_{y(z)}^{-1} > 0$ and $s_{N}^{2} = \Theta(N)$.  From this, it is immediate that $Ns_{N}^{-(2+ \delta)} \rightarrow 0$ as $N \rightarrow \infty$ for $\delta = 2$.  In total, this establishes \eqref{eqn: lyapounov condition for CATE}.  Since \eqref{eqn: lyapounov condition for CATE} is sufficient for the Lindeberg condition \citep[Page 427]{TestingStatHyp}, the result follows. 
\end{proof}
\begin{remark}
    Since the residuals $(\dot{\epsilon}_{i}(0), \dot{\epsilon}_{i}(1))$ defined by $\dot{\epsilon}_{i}(z) = y_{i}(z) - \overdotmu_{z}(x_i)$ are deterministic translations of the potential outcomes $(y_{i}(0), y_{i}(1))$ in the fixed-covariate model, Lemma~\ref{lem: Lindeberg condition} immediately implies that the residuals $(\dot{\epsilon}_{i}(0), \dot{\epsilon}_{i}(1))$ jointly satisfy the conditions of Lindeberg's central limit theorem.
\end{remark}

The stability assumption is also taken in accordance with the conditioning events of \eqref{eqn: fixed cov conditioning on outcomes}.

\begin{assumption}[Stability]\label{supp asm: fixed cov conditional stability}
    For all conditioning events of the form \eqref{eqn: fixed cov conditioning on outcomes} except for a set of measure zero under the fixed covariate model there exists a deterministic sequence of functions $\{\overdotmu_{1}^{(N)}\}_{N \in \N}$ such that
    \begin{equation}\label{eqn: fixed cov conditional stability}
        \left(\frac{1}{N}\sum_{i = 1}^{N}||\overdotmu_{1}^{(N)}(x_{i}) - \muhat_{1}(x_{i})||^{2} \right)^{1/2} = o_{P}(1).
    \end{equation}
    We assume that an analogous sequence, $\{\overdotmu_{0}^{(N)}\}_{N \in \N}$, exists for $\muhat_{0}$.  
\end{assumption}
In Assumption~\ref{supp asm: fixed cov conditional stability} the randomness on both sides of \eqref{eqn: fixed cov conditional stability} is only with respect to $Z$.

\begin{remark}
    The regularity conditions Assumptions~\ref{supp asm: fixed cov means and covs stabilize} and \ref{supp asm: fixed cov bounded fourth moment} do not prescribe a particular generative model; they are working-level mathematical ingredients in our subsequent proofs.  In order to complete the picture, we detail a conventional generative model which satisfies Assumptions~\ref{supp asm: fixed cov means and covs stabilize} and \ref{supp asm: fixed cov bounded fourth moment}.
    
    For each finite population the $N$ units have covariates $\{x_{i}\}_{i = 1}^{N}$ and the potential outcomes of unit $i$ are independent of all units $j$ for $j \neq i$.  The pair of potential outcomes $(y_{i}(0), y_{i}(1))$ are distributed according the the conditional distribution $P_{x_{i}}$.  Consequently, for any two individuals who exactly match on their covariates, their potential outcome pairs are independent and identically distributed.  Formally, $x_{i} = x_{j}$ implies that $(y_{i}(0), y_{i}(1))$ is equal in distribution to $(y_{j}(0), y_{j}(1))$.  Let expectations under the distribution $P_{x_{i}}$ be denoted as $\text{E}_{x_{i}}$; similarly, denote variances by $\text{var}_{x_{i}}$. 
    
    Assumption~\ref{supp asm: fixed cov means and covs stabilize} codifies the need for a strong law of large numbers for the means and variances of the realized populations under the fixed covariate model.  Assumption~\ref{supp asm: fixed cov means and covs stabilize} is highly general, but can be implied by moment conditions on the potential outcomes themselves.  In particular suppose that $\lim_{N \rightarrow \infty}N^{-1}\sum_{i = 1}^{N}\text{E}_{x_{i}}[y_{i}(z)]$ and $\lim_{N \rightarrow \infty}N^{-1}\sum_{i = 1}^{N}\text{var}_{x_{i}}(y_{i}(z))$ exist and are finite.  This is sufficient to guarantee the first condition of Assumption~\ref{supp asm: fixed cov means and covs stabilize}; the proof proceeds by application of Kolmogorov's strong law of large numbers for non-identically distributed sequences.  Similar reasoning establishes that if there exists a limiting positive semidefinite matrix $\Sigma$ such that
    \begin{multline*}
        \lim_{N \rightarrow \infty}\frac{1}{N - 1}\sum_{i = 1}^{N}\text{E}_{x_{i}}\Bigg[\left(\begin{bmatrix}
        y_{i}(0) \\ y_{i}(1) \\ \overdotmu_{0}(x_{i}) \\ \overdotmu_{1}(x_{i})
        \end{bmatrix} - N^{-1}\sum_{j = 1}^{N}\text{E}_{x_{j}}\left[\begin{bmatrix}
        y_{j}(0) \\ y_{j}(1) \\ \overdotmu_{0}(x_{j}) \\ \overdotmu_{1}(x_{j})
        \end{bmatrix} \right] \right)  \\
        \tensor \left(\begin{bmatrix}
        y_{i}(0) \\ y_{i}(1) \\ \overdotmu_{0}(x_{i}) \\ \overdotmu_{1}(x_{i})
        \end{bmatrix} - N^{-1}\sum_{j = 1}^{N}\text{E}_{x_{j}}\left[\begin{bmatrix}
        y_{j}(0) \\ y_{j}(1) \\ \overdotmu_{0}(x_{j}) \\ \overdotmu_{1}(x_{j})
        \end{bmatrix} \right] \right)\Bigg] = \Sigma
    \end{multline*}
    and sufficient control of higher order moments (e.g., coordinate-wise fourth moments) is assumed then the second condition of Assumption~\ref{supp asm: fixed cov means and covs stabilize} again follows by Kolmogorov's law of large numbers applied now to the entries of the matrix in \eqref{eqn: conditional cov mat limit}.  
    The moment condition
    \begin{equation*}
        N^{-1}\sum_{i = 1}^{N}\text{E}_{x_{i}}\left[y_{i}(z)^{4}\right] \rightarrow c_{z} \text{ for } z \in \{0, 1\}
    \end{equation*}
    for a constant $c_{z}$ and the requirement that Kolmogorov's condition \citep[Eqn. 7.2]{feller68} holds for the random variables $y_{i}(z)$ establishes Assumption~\ref{supp asm: fixed cov bounded fourth moment} by Kolmogorov's law of large numbers and simultaneously serves in establishing the second condition of Assumption~\ref{supp asm: fixed cov means and covs stabilize}.
    Since the covariates are non-stochastic we make the usual assumption that $N^{-1}\sum_{i = 1}^{N}\overdotmu_z(x_i)^{4} < C$ for $z \in \{0, 1\}$; this exactly mirrors our earlier finite population analysis.
\end{remark}

\begin{lemma}[Asymptotically Linear Expansions around the \textsc{SATE}]\label{lem: asymptotic linear expansion around SATE in fixed cov model}
     Under Assumptions~\ref{supp asm: non-degen sampling limit}, \ref{app asm: error process vanishes}, and \ref{supp asm: fixed cov conditional stability} the random variable $N^{-1}\sum_{i = 1}^{N}\left(\muhat_{z}(x_{i}) - y_{i}(z)\right)$ is asymptotically linear in the sense that, for $\dot{\epsilon}_{i}(z) = y_{i}(z) - \overdotmu_{z}(x_i)$
    \begin{equation*}
        \frac{1}{N} \sum_{i = 1}^{N}\left(\muhat_{z}(x_{i}) - y_{i}(z)\right) = \frac{1}{n_{z}}\sum_{i \st Z_{i} = z}\dot{\epsilon}_{i}(z) + o_{p}(N^{-1/2}).
    \end{equation*}
    
\end{lemma}

The proof of Lemma~\ref{lem: asymptotic linear expansion around SATE in fixed cov model} mostly mirrors that of Lemma~\ref{lem: asymptotic linear expansion around SATE}; the main working ingredient is the vanishing error process argument facilitated by Assumption~\ref{app asm: error process vanishes}.\footnote{In the fixed covariate model, the error process $\G_{N, z}(\overdotmu_{z}) - \G_{N, z}(\muhat_{z})$ inherits randomness only from stochasticity in the outcomes and treatment allocation while the covariates $x_{i}$ are viewed as deterministic vectors.}

For the analysis of central limit theorems under the fixed-covariate model it is mathematically convenient to adopt the probabilistic joint-model-design framework of \citet{twoPhaseFramework}.  Consider a probability space $(\Phi, \mathscr{F}, P)$ from which we form a population of $N$ individuals with potential outcome $y_{i}(z) = \mathcal{Y}_{z}(\omega_{i})$ for $z \in \{0, 1\}$ and covariates $x_{i} = \mathcal{X}(\omega_{i})$ for $\mathcal{Y}_{z}$ and $\mathcal{X}$ measurable functions of ${\omega_{i} \in \Phi}$.  Let $\mathcal{F}_{cov} = \left\{\omega \in \Phi \st \mathcal{X}(\omega_{i}) = \mathbf{x}_{i} \text{ for } i = 1, \ldots, N \right\}$; $\mathcal{F}_{cov}$ is the event that the covariates of the $N$ individuals are given by the deterministic values $\{\mathbf{x}_{i}\}_{i = 1}^{N}$.  Let $P_{\mathcal{F}_{cov}}$ be the conditional probability measure derived from $P$ conditioned on the event $\mathcal{F}_{cov}$.  In the case that $\mathcal{F}_{cov}$ is an event of $P$-measure zero, we tacitly assume that there exists a well-defined regular conditional probability measure and take $P_{\mathcal{F}_{cov}}$ to be this conditional; see \citet[Section 7.2]{chowTeicher} for more details on this technical issue.  Inferences under the fixed-covariate model take $(\Omega, \mathscr{F}, P_{\mathcal{F}_{cov}})$ to generate the outcomes $y_{i}(z) = \mathcal{Y}_{z}(\omega_{i})$ for $z \in \{0, 1\}$ and implicitly constrain the covariates  $x_{i} = \mathcal{X}(\omega_{i}) = \mathbf{x}_{i}$ for $i = 1, \ldots, N$. 

\begin{theorem}\label{supp thm: fixed covariate clt for taucal}
    Under the fixed-covariate model, subject to Assumptions~\ref{supp asm: non-degen sampling limit}, \ref{app asm: error process vanishes}, \ref{supp asm: fixed cov means and covs stabilize}, \ref{supp asm: fixed cov bounded fourth moment}, and \ref{supp asm: fixed cov conditional stability}, $N^{1/2}\left(\tauhat_{cal} - \CATE\right)$ obeys a central limit theorem.
\end{theorem}
\begin{proof}
    We start out with the simple observation that
    \begin{equation*}
        N^{1/2}\left(\tauhat_{cal} - \CATE\right) = N^{1/2}\left(\tauhat_{cal} - \SATE\right) + N^{1/2}\left(\SATE - \CATE\right).
    \end{equation*}
    
    Leveraging the asymptotic linearity of Lemma~\ref{lem: asymptotic linear expansion around SATE in fixed cov model} the first term on the right-hand-side can be replaced with the difference in means of the residuals $\dot{\epsilon}_{i}(1)$ and $\dot{\epsilon}_{i}(0)$ plus an $o_{P}(1)$ error term:
    \begin{multline*}
        N^{1/2}\left(\tauhat_{cal} - \CATE\right) = N^{1/2}\left(\frac{1}{n_{1}}\sum_{i \st Z_{i} = 1}\dot{\epsilon}_{i}(1) - \frac{1}{n_{0}}\sum_{i \st Z_{i} = 0}\dot{\epsilon}_{i}(0) \right) + o_{P}(1) + \\ N^{1/2}\left(\SATE - \CATE\right).
    \end{multline*}

    The $o_{P}(1)$ error term has no impact on the asymptotic distributional behaviour of $N^{1/2}\left(\tauhat_{cal} - \CATE\right)$.  Thus, to show a central limit theorem for $N^{1/2}\left(\tauhat_{cal} - \CATE\right)$ it suffices to show that
    \begin{enumerate}
        \item \label{itm: two CLTs} Conditionally upon the potential outcomes, the term $N^{1/2}\left(\frac{1}{n_{1}}\sum_{i \st Z_{i} = 1}\dot{\epsilon}_{i}(1) - \frac{1}{n_{0}}\sum_{i \st Z_{i} = 0}\dot{\epsilon}_{i}(0) \right)$ converges weakly in probability to a fixed Gaussian distribution and term $N^{1/2}\left(\SATE - \CATE\right)$ obeys a central limit theorem.
        \item \label{itm: asymptotic independence} The terms $ N^{1/2}\left(\tauhat_{cal} - \SATE\right) + o_{P}(1) = N^{1/2}\left(\frac{1}{n_{1}}\sum_{i \st Z_{i} = 1}\dot{\epsilon}_{i}(1) - \frac{1}{n_{0}}\sum_{i \st Z_{i} = 0}\dot{\epsilon}_{i}(0) \right)$ and $N^{1/2}\left(\SATE - \CATE\right)$ are asymptotically independent in the sense that their limiting joint distribution is the product of the two limiting marginal distributions.
    \end{enumerate}
    
    We tackle \ref{itm: two CLTs} first.  By Lemma~\ref{lem: Lindeberg condition} and the Lindeberg central limit theorem \citep[Theorem 11.2.5]{TestingStatHyp} it follows that $N^{1/2}\left(\SATE - \CATE\right)$ converges in distribution to a Gaussian distribution; denote this limiting distribution as $\Normal{0}{s_{m}}$.  
    
    Next, we show that $N^{1/2}\left(\frac{1}{n_{1}}\sum_{i \st Z_{i} = 1}\dot{\epsilon}_{i}(1) - \frac{1}{n_{0}}\sum_{i \st Z_{i} = 0}\dot{\epsilon}_{i}(0) \right)$ converges weakly in probability to a fixed Gaussian distribution. 
    
    Under Assumption~\ref{supp asm: fixed cov conditional stability} and the assumption that $N^{-1}\sum_{i = 1}^{N}\left( \overdotmu_{z}(x_{i}) - y_{i}(z)\right)^{2} = o(N)$ as a numeric sequence for almost all conditioning events of the potential outcomes, by Lemma~3 in the appendix of \citet{generalizedOB} we can, without loss of generality, stipulate that $N^{-1}\sum_{i = 1}^{N}\dot{\epsilon}_{i}(z) = 0$ for $z \in \{0, 1\}$ almost surely with respect to the conditioning \eqref{eqn: fixed cov conditioning on outcomes}.  Under Assumptions~\ref{supp asm: fixed cov means and covs stabilize} and \ref{supp asm: fixed cov bounded fourth moment} the finite population analysis provided in Theorem~\ref{supp thm: reg adj does no harm} shows that $N^{1/2}\left(\frac{1}{N}\sum_{i = 1}^{N}Z_{i}Nn_{1}^{-1}\dot{\epsilon}_{i}(1) - \frac{1}{N}\sum_{i = 1}^{N}(1 - Z_{i})Nn_{0}^{-1}\dot{\epsilon}_{i}(0) \right)$ converges weakly to a centered Gaussian distribution with variance given by the limit of $\sigma_{N}^{2}$ defined in \eqref{eqn: squared studenizing factor for cal}.  This limit exists by Assumption~\ref{supp asm: fixed cov means and covs stabilize} and is common to all conditioning events of the form \eqref{eqn: fixed cov conditioning on outcomes} up to a set of measure zero; we denote it by $s_{d}$.  Consequently, $N^{1/2}\left(\frac{1}{n_{1}}\sum_{i \st Z_{i} = 1}\dot{\epsilon}_{i}(1) - \frac{1}{n_{0}}\sum_{i \st Z_{i} = 0}\dot{\epsilon}_{i}(0) \right)$ converges weakly in probability to a random variable with distribution $\Normal{0}{s_{d}}$.

    Finally, we turn to \ref{itm: asymptotic independence}.  By Theorem 5.1 (iii) of \citet{twoPhaseFramework} it follows that the random vector $\left(N^{1/2}\left(\tauhat_{cal} - \SATE\right) ,N^{1/2}\left(\SATE - \CATE\right) \right)$ converges in distribution to  $(\mathcal{C}, \mathcal{D}) \sim \Normal{0}{s_{d}} \tensor \Normal{0}{s_{m}}$.\footnote{The original work of \citet{twoPhaseFramework} focuses on survey-sampling; however, nothing of their result Theorem 5.1 (iii) relies upon the survey-sampling framework of having only a single potential outcome, so we apply their result to the causal inference context of multiple potential outcomes.}
    
    By the continuous mapping theorem \citep[Theorem 18.11]{asymptoticStats_vdv}, $N^{1/2}\left(\tauhat_{cal} - \SATE\right)  + N^{1/2}\left(\SATE - \CATE\right)$ converges in distribution to $\mathcal{C} + \mathcal{D}$.  Since the sum of independent Gaussian random variables is itself Gaussian we have that $N^{1/2}\left(\tauhat_{cal} - \SATE\right)  + N^{1/2}\left(\SATE - \CATE\right)$ converges in distribution to $\Normal{0}{s_{d} + s_{m}}$.  In turn, this implies that $N^{1/2}\left(\tauhat_{cal} - \CATE\right)$ converges in distribution to $\Normal{0}{s_{d} + s_{m}}$.
\end{proof}

\begin{theorem}
    Assume the conditions of Theorem~\ref{supp thm: fixed covariate clt for taucal} hold.  Further suppose that the original prediction functions $\muhat_{0}$ and $\muhat_{1}$ are prediction unbiased.  Then $N^{1/2}(\tauhat_{cal} - \CATE)$, $N^{1/2}(\tauhat_{GBcal} - \CATE)$, $N^{1/2}(\tauhat_{gOB} - \CATE)$, and $N^{1/2}(\tauhat_{unadj} - \CATE)$ all obey central limit theorems and $N^{1/2}(\tauhat_{cal} - \CATE)$ has asymptotic variance no greater than any of the other three.
\end{theorem}
\begin{proof}
    The proof thematically mirrors that of Theorem~\ref{supp thm: superpopulation noninferiority}.  The first three central limit theorems are justified by Theorem~\ref{supp thm: fixed covariate clt for taucal} and analogous variants for $N^{1/2}(\tauhat_{GBcal} - \CATE)$ and $N^{1/2}(\tauhat_{gOB} - \CATE)$.  The fourth central limit theorem is justified by Lemma~\ref{lem: Lindeberg condition} and the Lindeberg central limit theorem \citep[Theorem 11.2.5]{TestingStatHyp}.
\end{proof}

\textcolor{black}{Variance estimation and the construction of confidence intervals proceeds via the variance estimator of \eqref{eqn: fin pop variance estimator} under analogous reasoning.}

\section{Further Simulations}\label{sec: further simulations}
\subsection{An Example with Logistic Regression}\label{sec: logistic simulation}
In the main text, we provided a simulation study to demonstrate the practical benefits of our calibration procedure while simultaneously showing the risks of uncalibrated estimators.  The main text used an example based upon Poisson regression; to further highlight the concerns of using uncalibrated estimates we now provide a simulation using logistic regression.

The $s$th of $S$ data sets contains $N$ individuals upon whom an experimenter performs a completely randomized experiment with $n_{1} = \lceil pN \rceil$ treated units.  In our simulations $p = 0.8$. Each unit has a scalar covariate $x_i$, generated as independent and identically distributed draws from a Uniform random variable on $[-8, 8]$. We then generate the potential outcomes under treatment and control for each individual independently as $y_i(1)\sim Bern\{f(x_i)\}$ and ${y_i(0)\sim Bern\{-0.4*(f(x_{i}) - 1)\}}$, where $Bern(c)$ is a Bernoulli distribution with probability of success $c$.  We take
\begin{equation*}
    f(x) = \frac{\exp(-3 + 2x)}{1 + \exp(-3 + 2x)}.
\end{equation*}
Consequently, the logistic regression model is correctly specified for the potential outcomes under treatment, but incorrectly specified for those under control.

For each data set, these values are left fixed while the remaining randomness in the simulation arises only from treatment allocation.  An experimenter observes only the count data $y_{i}(Z_{i})$ and continuous covariates $x_{i}$ for each unit.  Using the observed responses after each randomization of treatment allocation, we estimate the prediction functions $\hat{\mu}_0(x_i)$ and $\hat{\mu}_1(x_i)$ via separate logistic regressions of $y_i(Z_i)$ on $x_{i}$ in the subgroups where $Z_i=0$ and $Z_i=1$, respectively. We form the difference-in-means estimator $\hat{\tau}_{unadj}$, generalized Oaxaca-Blinder estimator $\hat{\tau}_{gOB}$, the singly-calibrated estimator of \citet[Equation 8]{generalizedOB} $\hat{\tau}_{GBcal}$, and our linearly-calibrated estimator $\hat{\tau}_{cal}$.

Table~\ref{tab: logistic sampling variances} compares the averages (over $s=1,...,S$) of the ratios of the variances for the adjusted estimators to the unadjusted estimator when setting $S=1000$, $B=1000$, and varying $N$.  Qualitatively, the results of Table~\ref{tab: logistic sampling variances} mirror those of the Poisson regression simulation in the main text: uncalibrated generalized Oaxaca-Blinder estimators can fare worse than the simple difference in means and the singly-calibrated estimator $\tauhat_{GBcal}$ fails to correct this issue; however, our calibration procedure asymptotically improves upon $\hat{\tau}_{unadj}$ while leveraging the desired nonlinear model.  


\begin{table}

\def~{\hphantom{0}}
\centering
\caption{Ratios of Monte Carlo variances for $\tauhat_{cal}$, $\tauhat_{GBcal}$, and $\tauhat_{gOB}$ to the difference in means estimator $\hat{\tau}_{unadj}$ for various experiment sizes $N$.  Each variance is based upon $B = 1000$ simulated treatment allocations for a given set of potential outcomes and covariates.  Results are averaged over $S = 1000$ simulated data sets.}{
\begin{tabular}{lccccc}
                & $\hat{\text{var}}(\tauhat_{gOB}) / \hat{\text{var}}(\tauhat_{unadj})$ & $\hat{\text{var}}(\tauhat_{GBcal}) / \hat{\text{var}}(\tauhat_{unadj})$  & $\hat{\text{var}}(\tauhat_{cal}) / \hat{\text{var}}(\tauhat_{unadj})$ &  &  \\
        $N = 200$   &  1.077  &  1.076  &  1.031  &  &  \\
        $N = 500$   &  1.056  &  1.054  &  0.993  &  &  \\
        $N = 1000$  &  1.050  &  1.047  &  0.981  &  &  \\
        $N = 10000$ &  1.043  &  1.041  &  0.970  &  &
\end{tabular}}
\label{tab: logistic sampling variances}
\end{table}

\subsection{Poisson Regression Calibration in Alternative Models}
To highlight the application of calibration in superpopulation and fixed-covariate models, we recreate the Poisson regression example from the main text at all three levels of inference.  For simulations in the superpopulation, new covariates and potential outcomes are redrawn in each of the $SB$ simulated data sets.  For simulation in the fixed-covariate model new covariates are constructed in the $s$th simulation, but are held fixed -- while potential outcomes are redrawn conditional upon these covariates -- for each randomization of treatment allocation $1, \ldots, B$ in the $s$th simulation.  Each simulation is conducted with $N = 10000$.  Variances are reported for appropriately centered versions $\tauhat_{cal}$, $\tauhat_{GBcal}$, and $\tauhat_{gOB}$.  As an example, for the $s$th simulation in the fixed-covariate case we compute the ratio of the variance of $\tauhat_{cal} - \CATE^{(s)}$ and the variance of $\tauhat_{unadj} - \CATE^{(s)}$ where $\CATE^{(s)}$ denotes the conditional average treatment effect in the $s$th simulated population.  Table~\ref{tab: SATE vs CATE vs PATE sampling variances} summarizes our results.  Even in the fixed-covariate and superpopulation models, uncalibrated Oaxaca-Blinder estimators may suffer from inflated asymptotic variances relative to the unadjusted difference in means and the single calibration of $\tauhat_{GBcal}$ fails to correct the issue. Fortunately, our linear calibration procedure succeeds in all three models, as evidenced by the third column of Table~\ref{tab: SATE vs CATE vs PATE sampling variances}.  The impact of calibration is most noticeable in the finite population framework but is nonetheless profound in all three models.  

\begin{table}

\def~{\hphantom{0}}
\centering
\caption{Ratios of Monte Carlo variances for $\tauhat_{cal}$, $\tauhat_{GBcal}$, and $\tauhat_{gOB}$ to the difference in means estimator $\hat{\tau}_{unadj}$ under different generative models.  Each variance is based upon $B = 1000$ simulated treatment allocations; results are averaged over $S = 1000$ simulated data sets.}{
\begin{tabular}{lccccc}
                & $\hat{\text{var}}(\tauhat_{gOB}) / \hat{\text{var}}(\tauhat_{unadj})$ & $\hat{\text{var}}(\tauhat_{GBcal}) / \hat{\text{var}}(\tauhat_{unadj})$  & $\hat{\text{var}}(\tauhat_{cal}) / \hat{\text{var}}(\tauhat_{unadj})$ &  &  \\
        \textsc{SATE}   &  1.660  &  1.657  &  0.654  &  &  \\
        \textsc{CATE}   &  1.549  &  1.547  &  0.711  &  &  \\
        \textsc{PATE}   &  1.114  &  1.114  &  0.941  &  &
\end{tabular}
}
\label{tab: SATE vs CATE vs PATE sampling variances}
\end{table}

In the third column of Table~\ref{tab: SATE vs CATE vs PATE sampling variances} $\hat{\text{var}}(\tauhat_{cal}) / \hat{\text{var}}(\tauhat_{unadj})$ increases as one goes from inference for the \textsc{SATE}\, to the \textsc{CATE}\, and finally to the \textsc{PATE}.  This trend is a fundamental reflection of asymptotic variances in the three models under consideration.  By the law of total variance, for any measurable event $\mathcal{F}$, the variance of $\tauhat_{unadj}$ decomposes as 
\begin{equation*}
    \Variance{N^{1/2}\tauhat_{unadj}} = \Expectation{\Variance{N^{1/2}\tauhat_{unadj} \given \mathcal{F}}} + N\Variance{\Expectation{\tauhat_{unadj} \given \mathcal{F}}}.
\end{equation*}
Taking $\mathcal{F}$ to be the event $\{(y_{i}(0), y_{i}(1), x_{i}) = (\mathbf{y}_{i}(0), \mathbf{y}_{i}(1), \mathbf{x}_{i})  \, i = 1, \ldots, N\}$ yields the decomposition of the variance of $\tauhat_{unadj}$ into the expected variance of $\tauhat_{unadj}$ given a fixed finite population and the variance of the expectation of $\tauhat_{unadj}$ given that fixed finite population.  Since the difference in means is unbiased for the sample average treatment effect given $\mathcal{F}$ it follows that the second term is just the variance of the sample average treatment effect, $\Variance{\SATE}$.  This decomposition is valid regardless of the underlying model of the data: superpopulation, fixed-covariate, or finite population.  For the sake of explanation, we discuss the difference between the \textsc{SATE}\, row and the \textsc{PATE}\, row, but the same reasoning applies to the differences between the \textsc{SATE}\, and \textsc{CATE}\, rows and the \textsc{CATE}\, and \textsc{PATE}\, rows of Table~\ref{tab: SATE vs CATE vs PATE sampling variances}.

In the finite population model there is no variance of the sample average treatment effect whatsoever since it depends only on fixed values; in other words, for the top row of $N\Variance{\Expectation{\tauhat_{unadj} \given \mathcal{F}}} = 0$.  In contrast with the finite population case, under a superpopulation model the sample average treatment effect may vary; in all but the most degenerate cases $N\Variance{\SATE} > 0$.  Furthermore, under the finite population conditions of \citet{lin13} and standard fourth-moment regularity conditions in the superpopulation model $\Expectation{\Variance{N^{1/2}\tauhat_{unadj} \given \mathcal{F}}}$ has the same limit in the finite population model and the superpopulation model.

The same variance decomposition can be applied to the estimators $\tauhat_{gOB}$, $\tauhat_{GBcal}$, and $\tauhat_{cal}$.  Without loss of generality, we discuss the case of $\tauhat_{cal}$.  Since $\tauhat_{cal}$ is asymptotically unbiased in all three models, the term $N\Variance{\Expectation{\tauhat_{cal} \given \mathcal{F}}}$ limits to the $N$-scaled variance of the sample average treatment effect in all three models.  As before, under mild regularity conditions $\Expectation{\Variance{N^{1/2}\tauhat_{cal} \given \mathcal{F}}}$ has the same limit in the finite population model and the superpopulation model.  Thus, the difference in asymptotic variances between $\tauhat_{cal}$ and $\tauhat_{unadj}$ is driven by the limiting difference between $\Expectation{\Variance{N^{1/2}\tauhat_{cal} \given \mathcal{F}}}$ and $\Expectation{\Variance{N^{1/2}\tauhat_{unadj} \given \mathcal{F}}}$.  By the reasoning above, this limiting difference is the same in both the finite population model and in a superpopulation model.  Taking the results for large $N$ as reflective of their asymptotic behavior the difference between the top-right and bottom-right elements of Table~\ref{tab: SATE vs CATE vs PATE sampling variances} can be explained as: the difference between the numerator and denominator remains the same while the magnitude of both the numerator and the denominator are larger in the \textsc{PATE}\, row than in the \textsc{SATE}\, row.  Analogous reasoning applies to the \textsc{SATE}\, versus \textsc{CATE}\, rows and the \textsc{CATE}\, versus \textsc{PATE}\, rows and across the other columns as well.

\section{A Case Study on Tumor Recurrence of Bladder Cancers}

The Veterans Administration Cooperative Urological Research Group (VACURG) conducted a completely randomized clinical trial to examine the effectiveness of treatment against recurrence of bladder cancers.  Each patient enrolled in the study had superficial bladder tumors are the start of the study; the tumors were removed transurethrally before the patients were assigned to one of three treatment conditions: placebo pills, pyridoxine (vitamin $B_6$) pills, or periodic treatment with thiotepa (a chemotherapeutic agent).
The patients returned for follow-up visits and the existence of recurrent tumors observed in these follow-ups was tabulated; although -- at times -- more than one tumor was observed during a follow-up appointment the number of such tumors is not the primary object of study, only their presence or not is recorded as a binary outcome in each follow-up.  After a recurrent tumor was observed, it was removed and the treatment regimen assigned to that individual was continued.  For further details see \citet[Chapter 45]{dataBook}.  Our analysis focuses upon the placebo group, with 47 individuals, and the thiotepa treatment group, with 38 individuals.

The primary outcome of the study is the count of the number of recurrences, so Poisson regression is a natural adjustment model.  Covariate information collected at the start of the experiment includes the initial number of tumors and the diameter of the largest of these.  The number of months over which the patient attended follow-up appointments was recorded as well as the survival status of the patient at the conclusion of the study.  We control for the $\log$-number of follow-up months, the number of initial tumors, and the diameter of the largest initial tumor.  We compare the unadjusted difference in means, the uncalibrated generalized Oaxaca-Blinder estimator of \citet{generalizedOB}, the singly-calibrated estimator of \citet{generalizedOB}, and our calibrated estimator.  Table~\ref{tab: case study} displays the point estimate of treatment effect and the corresponding estimated variance. The variances displayed along the right-hand column of Table~\ref{tab: case study} demonstrate the substantial benefit of controlling for features.  Only $\tauhat_{cal}$ is guaranteed to be non-inferior to $\tauhat_{unadj}$; the performance of $\tauhat_{gOB}$ and $\tauhat_{GBcal}$ is not generally guaranteed.  Moreover, $\tauhat_{cal}$ never has asymptotic variance which exceeds that of $\tauhat_{gOB}$ and $\tauhat_{GBcal}$; this is observed even in this sample.

\begin{table}

\def~{\hphantom{0}}
\centering
\caption{Point estimates and estimated variances of $\tauhat_{unadj}$, $\tauhat_{gOB}$, and $\tauhat_{GBcal}$, and $\hat{\tau}_{cal}$ on the {VACURG} bladder tumor recurrence data set.}{
\begin{tabular}{lcc}
                            & Point Estimate    & Variance  \\
        $\tauhat_{unadj}$ 	&   -0.667                 &   0.190     \\
    $\tauhat_{gOB}$     	&   -0.775                 &   0.123     \\
    $\tauhat_{GBcal}$   	&   -0.784                 &   0.122     \\
    $\tauhat_{cal}$     	&   -0.778                 &   0.120
\end{tabular}} 
\label{tab: case study}
\end{table}

\color{black}
\section{Calibration and Semiparametric Efficiency}\label{sec: semiparametric efficiency}
In superpopulation models, a great deal of regression adjustment literature has focused upon semiparametric efficiency of estimators.  Below we include a brief survey of some of this literature and demonstrate the relationship between semiparametric efficient estimators and the calibration procedure.

\citet{Hahn98} takes a superpopulation approach to inference; his formulation aligns with the superpoplation framework of \citet{twoPhaseFramework}.  Units of the population are $N$ independent and identically distributed tuples $(y_{i}(0), y_{i}(1), x_{i})$ and the object of inference is the population average treatment effect $\Expectation{y_{1}(1) - y_{1}(0)}$.  We present his main semiparametric efficiency bound adapted to the context of completely randomized experiments below; it can be envisioned in the same light as the Cram\'{e}r-Rao bound as it provides a lower bound on the asymptotic variance of any regular estimator of $\Expectation{y_{1}(1) - y_{1}(0)}$.

\begin{theorem}[\citet{Hahn98}, Theorem 1]
    In a completely randomized experiment, the asymptotic variance of any regular estimator sequence for $\Expectation{y_{1}(1) - y_{1}(0)}$ is bounded below by
    \begin{equation}\label{eqn: hahn lower bound}
        \Expectation{\frac{\Variance{y_{i}(1) \given x_{i}}}{p} + \frac{\Variance{y_{i}(0) \given x_{i}}}{1 - p} + \left(\Expectation{y_{1}(1) - y_{1}(0) \given x_{i}} - \Expectation{y_{1}(1) - y_{1}(0)} \right)^{2}}.
    \end{equation}
\end{theorem}

Any sequence of regular estimators for $\Expectation{y_{1}(1) - y_{1}(0)}$ which achieves an asymptotic variance of \eqref{eqn: hahn lower bound} in the limit is said to be (asymptotically) \textit{semiparametric efficient}.  Write the conditional expectation of the outcomes given the covariates as
\begin{equation*}
    \condExp_{z}(x) = \Expectation{y_{i}(z) \given x_{i} = x} \quad \text{for } z \in \{0, 1\}.
\end{equation*}
Citing a result by \citet{Hahn98}, \citet{rothe} remarks that any semiparametric efficient regular estimator of $\Expectation{y_{1}(1) - y_{1}(0)}$ is necessarily of the form $N^{-1}\sum_{i = 1}^{N} \psi_{i}(\condExp_{0}, \condExp_{1}) + o_{p}(N^{-1/2})$ where
\begin{multline}\label{eqn: general form of semiparametric efficient estimator}
    \psi_{i}(\condExp_{0}, \condExp_{1}) := \left(\condExp_{1}(x_{i}) - \condExp_{0}(x_{i})\right) + \\
    \frac{Z_{i}\left(y_{i}(Z_{i}) - \condExp_{1}(x_{i})\right) }{(n_{1} / N)} - \frac{(1 - Z_{i})\left(y_{i}(Z_{i}) - \condExp_{0}(x_{i})\right) }{(n_{0} / N)}.
\end{multline}

We show that any estimator of the form $N^{-1}\sum_{i = 1}^{N} \psi_{i}(\muhat_{0}, \muhat_{1}) + o_{p}(N^{-1/2})$ obtains non-inferiority after calibration; even though $N^{-1}\sum_{i = 1}^{N} \psi_{i}(\muhat_{0}, \muhat_{1}) + o_{p}(N^{-1/2})$ originally could have been asymptotically inferior to the unadjusted difference in means.  Consequently, any estimator which has hope of semiparametric efficiency can be made non-inferior to the difference in means via our calibration procedure.  If it was the case that the original estimator sequence happened to be semiparametric efficient, then the asymptotic variance of the calibrated estimator will coincide with the semiparametric efficiency bound; however for regular estimators which do not achieve the semiparametric efficiency bound calibration automatically yields asymptotic non-inferiority to the difference in means.  We formalize this statement below.

\citet{rothe} considers estimators of the form $\tauhat_{rothe} = N^{-1}\sum_{i = 1}^{N} \psi_{i}(\muhat_{0}, \muhat_{1})$ where $\muhat_{0}$ and $\muhat_{1}$ attempt to estimate the true conditional expectation functions $\condExp_{0}$ and $\condExp_{1}$, respectively.  Consider calibrating $\tauhat_{rothe}$ by defining $\muhat_{OLS, 0}$ and $\muhat_{OLS, 1}$ according to the usual linear calibration formula of the main manuscript and forming 
\begin{equation*}
    \tauhat_{rothe, cal} = N^{-1}\sum_{i = 1}^{N} \psi_{i}(\muhat_{OLS, 0}, \muhat_{OLS, 1}).
\end{equation*}

\begin{theorem}\label{supp thm: comparison to Rothe}
    Subject to Assumptions~\ref{supp asm: stability}, \ref{app asm: error process vanishes}, \ref{supp asm: superpop means and covs stabilize}, and \ref{supp asm: superpop bounded fourth moment}:
    \begin{itemize}
        \item $\tauhat_{rothe, cal}$ is asymptotically linear in the sense that 
        \begin{equation}\label{eqn: calibrated rothe asymptotic linearity}
            N^{1/2}\left(\tauhat_{rothe, cal} - \PATE\right) = N^{-1/2}\sum_{i = 1}^{N} \left(\psi_{i}(\overdotmu_{OLS, 0}, \overdotmu_{OLS, 1}) - \PATE\right) + o_{P}(1).
        \end{equation}
        
        \item The asymptotic variance of $N^{1/2}\left(\tauhat_{rothe, cal} - \PATE\right)$ lies in the closed interval $[\Sigma_{l}, \Sigma_{u}]$ where $\Sigma_{l}$ is the semiparametric efficiency bound of \eqref{eqn: hahn lower bound} and $\Sigma_{u}$ is the asymptotic variance of the $N^{1/2}\left(\tauhat_{unadj} - \PATE\right)$.  Furthermore, the asymptotic variance of $N^{1/2}\left(\tauhat_{rothe, cal} - \PATE\right)$ is no greater than the asymptotic variance of $N^{1/2}\left(\tauhat_{rothe} - \PATE\right)$, so in the case that $\tauhat_{rothe}$ is more asymptotically precise than $\tauhat_{unadj}$ we can replace $\Sigma_{u}$ with the limiting variance of $N^{1/2}\left(\tauhat_{rothe} - \PATE\right)$ and thereby shrink the interval even further.
    \end{itemize}
\end{theorem}
\begin{proof}
    The stability assumption (Assumption~\ref{supp asm: stability}) coupled with the moment assumptions Assumptions~\ref{supp asm: superpop means and covs stabilize}, and \ref{supp asm: superpop bounded fourth moment} establish that the calibrated imputation functions $\muhat_{OLS, 0}$ and $\muhat_{OLS, 1}$ are stable themselves; the argument mirrors the finite population case as before.  Consequently, $\muhat_{OLS, 0}$ and $\muhat_{OLS, 1}$ satisfy Assumption~1 of \citet{rothe} (which is basically a variant of our stability assumption).  
    Furthermore, 
    the vanishing error process assumption (Assumption~\ref{app asm: error process vanishes}) coupled with Assumptions~\ref{supp asm: stability}, \ref{supp asm: superpop means and covs stabilize}, and \ref{supp asm: superpop bounded fourth moment} establishes
    $$
        N^{1/2}\left(\tauhat_{rothe, cal} - \PATE\right) = N^{-1/2}\sum_{i = 1}^{N} \left(\psi_{i}(\overdotmu_{OLS, 0}, \overdotmu_{OLS, 1}) - \PATE\right) + o_{P}(1).
    $$
    
    Now we turn attention to the second claim of the theorem.  The lower bound is an automatic consequence of \citet[Corollary 1]{rothe} and \citet[Theorem 1]{Hahn98}; thus we turn to the upper bound. By the earlier asymptotic linearity result (see \eqref{eqn: calibrated rothe asymptotic linearity} for the calibrated case and \citet[Theorem 1]{rothe} for the uncalibrated case), the asymptotic variance of $N^{1/2}\left(\tauhat_{rothe} - \PATE\right)$ equals that of $N^{-1/2}\sum_{i = 1}^{N} \left(\psi_{i}(\overdotmu_{0}, \overdotmu_{1}) - \PATE\right)$; consequently, it suffices to examine the asymptotic variance of $N^{-1/2}\sum_{i = 1}^{N} \left(\psi_{i}(\overdotmu_{0}, \overdotmu_{1}) - \PATE\right)$.  We leverage the law of total variance via a conditioning argument; as in Section~\ref{sec: further simulations} let $\mathcal{F}$ to be the event $\{(y_{i}(0), y_{i}(1), x_{i}) = (\mathbf{y}_{i}(0), \mathbf{y}_{i}(1), \mathbf{x}_{i})  \, i = 1, \ldots, N\}$.
    \begin{multline*}
        \Variance{N^{-1/2}\sum_{i = 1}^{N} \left(\psi_{i}(\overdotmu_{0}, \overdotmu_{1}) - \PATE\right)} = \Expectation{\Variance{N^{-1/2}\sum_{i = 1}^{N} \left(\psi_{i}(\overdotmu_{0}, \overdotmu_{1}) - \PATE\right) \,\Big|\, \mathcal{F}}} + \\
        \Variance{\Expectation{N^{-1/2}\sum_{i = 1}^{N} \left(\psi_{i}(\overdotmu_{0}, \overdotmu_{1}) - \PATE\right) \,\Big|\, \mathcal{F}}}
    \end{multline*}
    Given $\mathcal{F}$ the only randomness in $N^{-1/2}\sum_{i = 1}^{N} \left(\psi_{i}(\overdotmu_{0}, \overdotmu_{1}) - \PATE\right)$ comes from the random allocation of treatment assignment, $Z$, and since 
    \begin{equation}\label{eqn: definition of psi for generic overdotmus}
        \psi_{i}(\overdotmu_{0}, \overdotmu_{1}) = \overdotmu_{1}(x_{i}) - \overdotmu_{0}(x_{i}) + \frac{Z_{i}\left(y_{i}(Z_{i}) - \overdotmu_{1}(x_{i})\right) }{(n_{1} / N)} - \frac{(1 - Z_{i})\left(y_{i}(Z_{i}) - \overdotmu_{0}(x_{i})\right) }{(n_{0} / N)},
    \end{equation}
    it follows that
    \begin{equation*}
        \Expectation{N^{-1/2}\sum_{i = 1}^{N} \left(\psi_{i}(\overdotmu_{0}, \overdotmu_{1}) - \PATE\right) \,\Big|\, \mathcal{F}} = N^{1/2}\left(\frac{1}{N}\sum_{i = 1}^{N}(y_{i}(1) - y_{i}(0)) - \PATE \right).
    \end{equation*}
    Consequently, the term $\Variance{\Expectation{N^{-1/2}\sum_{i = 1}^{N} \left(\psi_{i}(\overdotmu_{0}, \overdotmu_{1}) - \PATE\right) \,\Big|\, \mathcal{F}}}$ has no dependence upon $\overdotmu_{0}$ and $\overdotmu_{1}$; it is just determined by the variance of the sample average treatment effect.  Formally,
    \begin{equation*}
        \Variance{\Expectation{N^{-1/2}\sum_{i = 1}^{N} \left(\psi_{i}(\overdotmu_{0}, \overdotmu_{1}) - \PATE\right) \,\Big|\, \mathcal{F}}} = \Variance{N^{1/2}\left(\SATE - \PATE \right)}.
    \end{equation*}
    
    Furthermore, by inspection of \eqref{eqn: definition of psi for generic overdotmus}, the conditional variance of $\psi_{i}(\overdotmu_{0}, \overdotmu_{1})$ given $\mathcal{F}$ is only dependent upon variability in $\frac{Z_{i}\left(y_{i}(Z_{i}) - \overdotmu_{1}(x_{i})\right) }{(n_{1} / N)} - \frac{(1 - Z_{i})\left(y_{i}(Z_{i}) - \overdotmu_{0}(x_{i})\right) }{(n_{0} / N)}$ inherited from randomness in the $Z_{i}$, so
    \begin{multline*}
        \Variance{N^{-1/2}\sum_{i = 1}^{N} \left(\psi_{i}(\overdotmu_{0}, \overdotmu_{1}) - \PATE\right) \,\Big|\, \mathcal{F}} = \\
        \Variance{N^{-1/2}\sum_{i = 1}^{N} \left(\frac{Z_{i}\left(y_{i}(Z_{i}) - \overdotmu_{1}(x_{i})\right) }{(n_{1} / N)} - \frac{(1 - Z_{i})\left(y_{i}(Z_{i}) - \overdotmu_{0}(x_{i})\right) }{(n_{0} / N)} - \PATE\right) \,\Big|\, \mathcal{F}}.
    \end{multline*}
    In total, we have shown that
    \begin{multline}\label{eqn: asymptotic variance decomp for psi}
         \Variance{N^{-1/2}\sum_{i = 1}^{N} \left(\psi_{i}(\overdotmu_{0}, \overdotmu_{1}) - \PATE\right)} = \\
         \Expectation{\Variance{N^{-1/2}\sum_{i = 1}^{N} \left(\frac{Z_{i}\left(y_{i}(Z_{i}) - \overdotmu_{1}(x_{i})\right) }{(n_{1} / N)} - \frac{(1 - Z_{i})\left(y_{i}(Z_{i}) - \overdotmu_{0}(x_{i})\right) }{(n_{0} / N)} - \PATE\right) \,\Big|\, \mathcal{F}}} +\\
         \Variance{N^{1/2}\left(\SATE - \PATE \right)}.
    \end{multline}
    
    Comparing \eqref{eqn: asymptotic variance decomp for psi} to the variance decomposition of Theorem~\ref{supp thm: superpopulation noninferiority} and taking the limit as $N \rightarrow \infty$ yields that the asymptotic variance of $N^{1/2}\left(\tauhat_{rothe} - \PATE\right)$ matches that of the $N^{1/2}\left(\tauhat_{gOB} - \PATE\right)$ where the generalized Oaxaca-Blinder estimator is computed using the same imputation functions $\muhat_{0}$ and $\muhat_{1}$ that $\tauhat_{rothe}$ uses.  Consequently, the non-inferiority results for $N^{1/2}\left(\tauhat_{cal} - \PATE\right)$ proven in Theorem~\ref{supp thm: superpopulation noninferiority} translate to $N^{1/2}\left(\tauhat_{rothe, cal} - \PATE\right)$ as well.  This establishes the upper bound on the asymptotic variance of $\tauhat_{rothe, cal}$ and thereby completes the proof.
\end{proof}

\begin{remark}
    The second claim of Theorem~\ref{supp thm: comparison to Rothe} substantially improves upon the result of \citet[Corollary 1, Part ii]{rothe} which provides only the lower bound.  The upper bound of Theorem~\ref{supp thm: comparison to Rothe} guarantees two things:
    \begin{enumerate}
        \item A practitioner is certain that their calibrated estimator is asymptotically no less efficient than the difference in means (or the uncalibrated estimator) regardless of model misspecification. 
        \item If the original estimator of $\tauhat_{rothe}$ was indeed a semiparamteric efficient estimator then so too is the new calibrated estimator $\tauhat_{rothe, cal}$.
    \end{enumerate}
\end{remark}

Informally stated, the result of \citet{Hahn98} establishes that any estimator which has hope of semiparamteric efficiency must be of the form $\tauhat_{rothe}$ for some choice of $\muhat_{0}$ and $\muhat_{1}$ and Theorem~\ref{supp thm: comparison to Rothe} goes on to establish that any such estimator can be imbued with non-inferiority via the calibration procedure.

\section{Cross-Fitting and Calibration}
Cross-fitting is an algorithmic procedure based upon randomly splitting the sample into multiple portions, often called ``folds", computing prediction functions on some portion of these folds, and then applying the prediction functions to data from the other folds.  Since the prediction function is trained on different data than it is applied to, it is independent of the data it is applied to under standard superpopulation models.  This independence provides numerous benefits from a theoretical angle and has established cross-fitting as a common and powerful tool in statistical literature.  In particular, \citet{doubleMachineLearning} demonstrated that an appropriate use of cross-fitting could achieve strong statistical inference guarantees while eschewing classical Donsker-style entropy conditions.  Here we demonstrate that cross-fitting is compatible with calibration to yield non-inferiority results for a wide array of prediction functions $\muhat_{0}$ and $\muhat_{1}$ which need not satisfy the ``typically simple realizations" entropy condition.  Our discussion centers around superpopulation inference in the style of Section~\ref{sec: superpop}.

We begin with some new notation.  Superscripts of $(-i)$ indicate that the associated random function is independent of the $i$th data point; for example $\muhat_{1}^{(-i)}(\cdot)$ is a prediction function of treated outcomes which is independent of the $i$th observation $(y_{i}(Z_{i}), x_{i})$.  In practice, $\muhat_{1}^{(-i)}(\cdot)$ is usually computed by fitting the random prediction function $\muhat_{1}$ on the data set of $N - 1$ individuals which excludes the $i$th individual.  In line with the estimator $\tauhat_{rothe}$ and the work of \citet{highDimRegAdj}, define the ``leave-one-out" estimator
\begin{multline*}
    \tauhat_{loo} = \frac{1}{N}\sum_{i = 1}^{N}\left(\muhat_{1}^{(-i)}(x_{i}) - \muhat_{0}^{(-i)}(x_{i}) \right) + \\
    \frac{1}{n_{1}}\sum_{i \st Z_{i} = 1}\left(y_{i}(Z_{i}) - \muhat_{1}^{(-i)}(x_{i}) \right) - \frac{1}{n_{0}}\sum_{i \st Z_{i} = 0}\left(y_{i}(Z_{i}) - \muhat_{0}^{(-i)}(x_{i}) \right).
\end{multline*}
To define the calibrated version of $\tauhat_{loo}$, some care is needed to account for the sample splitting in both the original prediction functions $(\muhat_{0}, \muhat_{1})$ and in the calibrated prediction functions.  The leave-one-out calibrated prediction function is defined as

\begin{align}\label{eq: leave-one-out lincal}
    \hat{\mu}_{OLS,z}^{(-i)}(x_i) &= \hat{\alpha}_{z}^{(-i)} + \hat{\beta}_{z,0}^{(-i)}\hat{\mu}_0^{(-i)}(x_i) + \hat{\beta}_{z,1}^{(-i)}\hat{\mu}_1^{(-i)}(x_i);\\
    (\hat{\alpha}_z^{(-i)}, \hat{\beta}_{z,0}^{(-i)}, \hat{\beta}_{z,1}^{(-i)})^\T& \in \underset{(\alpha_z, \beta_{z,0}, \beta_{z,1})^\T}{\arg\min}\;\; \sum_{\substack{j: Z_j = z\\ j \neq i}}\{y_j(z) - {\alpha}_{z} - {\beta}_{z,0}\hat{\mu}_0^{(-i)}(x_j) - {\beta}_{z,1}\hat{\mu}_1^{(-i)}(x_j)\}^2.\nonumber 
\end{align} 

The calibrated prediction function is then
\begin{multline*}
    \tauhat_{loo, cal} = \frac{1}{N}\sum_{i = 1}^{N}\left(\muhat_{OLS, 1}^{(-i)}(x_{i}) - \muhat_{OLS, 0}^{(-i)}(x_{i}) \right) + \\
    \frac{1}{n_{1}}\sum_{i \st Z_{i} = 1}\left(y_{i}(Z_{i}) - \muhat_{OLS, 1}^{(-i)}(x_{i}) \right) - \frac{1}{n_{0}}\sum_{i \st Z_{i} = 0}\left(y_{i}(Z_{i}) - \muhat_{OLS, 0}^{(-i)}(x_{i}) \right).
\end{multline*}

In the context of leave-one-out estimation, we slightly modify the notation of the error process in Assumption~\ref{app asm: error process vanishes}.  Specifically, write
\begin{equation*}
    \G_{N, z}(\muhat_{z}) = N^{-1/2}\sum_{i = 1}^{N}\left(\frac{\indicatorFunction{Z_{i} = z}\muhat^{(-i)}_{z}(x_{i})}{n_{z}/N} - \muhat^{(-i)}_{z}(x_{i}) \right).
\end{equation*}

This modification of $\G_{N, z}(\muhat_{z})$ is done to take into account the fact that different prediction functions may be used for different individuals in the population.  In the case of leave-one-out estimation, there are $2N$ different prediction functions $\left(\muhat_{0}^{(-1)}, \ldots, \muhat_{0}^{(-N)}, \muhat_{1}^{(-1)}, \ldots, \muhat_{1}^{(-N)}\right)$.  This change does not modify any of the structure of our previous proofs.  In fact, we could have originally defined the process $\G_{N, z}$ to account for different prediction functions at each $i$, but this would have introduced needless notational burden for the previous proofs wherein the functions $\muhat_{0}$ and $\muhat_{1}$ do not depend upon $i$.

\begin{theorem}\label{supp thm: crossfitting to avoid entropy conditions}
    Consider the superpopulation model model of Section~\ref{sec: superpop}.  
    Suppose that Assumptions~\ref{supp asm: non-degen sampling limit}, \ref{app asm: error process vanishes}, \ref{supp asm: superpop means and covs stabilize}, and \ref{supp asm: superpop bounded fourth moment} hold; then $\tauhat_{loo, cal}$ obeys a central limit theorem and the asymptotic variance of $N^{1/2}\left(\tauhat_{loo, cal} - \PATE\right)$ lies in the closed interval $[\Sigma_{l}, \Sigma_{u}]$ where $\Sigma_{l}$ is the semiparametric efficiency bound of \eqref{eqn: hahn lower bound} and $\Sigma_{u}$ is the asymptotic variance of the $N^{1/2}\left(\tauhat_{unadj} - \PATE\right)$.  Furthermore, the asymptotic variance of $N^{1/2}\left(\tauhat_{loo, cal} - \PATE\right)$ is no greater than that of $N^{1/2}\left(\tauhat_{loo} - \PATE\right)$; so in the case that $\tauhat_{loo}$ is more asymptotically precise than $\tauhat_{unadj}$ we can replace $\Sigma_{u}$ with the limiting variance of $N^{1/2}\left(\tauhat_{loo} - \PATE\right)$ and thereby shrink the interval even further.
\end{theorem}
\begin{proof}
    As in our previous proofs, we begin by showing that the estimators $\tauhat_{loo}$ and $\tauhat_{loo, cal}$ are asymptotically linear under the assumed regularity conditions.  Quite similarly to \citet{highDimRegAdj}
    \begin{equation*}
        \tauhat_{loo} = \frac{1}{N}\left(\overdotmu_{1}(x_{i}) - \overdotmu_{0}(x_{i}) \right) + \frac{1}{n_{1}}\sum_{i \st Z_{i} = 1}\left(y_{i}(Z_{i}) - \overdotmu_{1}(x_{i}) \right) - \frac{1}{n_{0}}\sum_{i \st Z_{i} = 0}\left(y_{i}(Z_{i}) - \overdotmu_{0}(x_{i}) \right) + R
    \end{equation*}
    where $R$ is defined as
    \begin{equation*}
        R = \sum_{i = 1}^{N}\frac{(-1)^{Z_{i}}}{n_{Z_{i}}}\left(\frac{n_{0}}{N}\left(\muhat_{1}^{(-i)}(x_{i}) - \overdotmu_{1}(x_{i}) \right) + \frac{n_{1}}{N}\left(\muhat_{0}^{(-i)}(x_{i}) - \overdotmu_{0}(x_{i}) \right)\right).
    \end{equation*}
    
    The quantity $\tauhat_{loo} - R$ exactly agrees with $N^{-1}\sum_{i = 1}^{N}\psi_{i}(\overdotmu_{0}, \overdotmu_{1})$, and so analysis of $\tauhat_{loo}$ reduces to analysis of $N^{-1}\sum_{i = 1}^{N}\psi_{i}(\overdotmu_{0}, \overdotmu_{1})$ so long as $R$ vanishes asymptotically at a sufficiently fast rate.  Algebraic manipulation shows that
    \begin{equation*}
        R = \frac{1}{N^{1/2}}\left(\G_{N, 1}(\overdotmu_{1}) - \G_{N, 1}(\muhat_{1}) \right) - \frac{1}{N^{1/2}}\left(\G_{N, 0}(\overdotmu_{0}) - \G_{N, 0}(\muhat_{0}) \right).
    \end{equation*}
    Consequently, by the vanishing error process assumption, $|R| = o_{P}(N^{-1/2})$.  Similarly, by Lemma~\ref{prop: error proces vanishing is inherited by consistent OLS} since $|R| = o_{P}(N^{-1/2})$ it follows that $|R_{cal}| = o_{P}(N^{-1/2})$ as well where
    \begin{equation*}
        R_{cal} = \sum_{i = 1}^{N}\frac{(-1)^{Z_{i}}}{n_{Z_{i}}}\left(\frac{n_{0}}{N}\left(\muhat_{OLS, 1}^{(-i)}(x_{i}) - \overdotmu_{OLS, 1}(x_{i}) \right) + \frac{n_{1}}{N}\left(\muhat_{OLS, 0}^{(-i)}(x_{i}) - \overdotmu_{OLS, 0}(x_{i}) \right)\right).
    \end{equation*}    
    
    Consequently,
    \begin{align*}
        \tauhat_{loo} &= N^{-1}\sum_{i = 1}^{N}\psi_{i}(\overdotmu_{0}, \overdotmu_{1}) + o_{P}(N^{-1/2}),\\
        \tauhat_{loo, cal} &= N^{-1}\sum_{i = 1}^{N}\psi_{i}(\overdotmu_{OLS, 0}, \overdotmu_{OLS, 1}) + o_{P}(N^{-1/2}),
    \end{align*}
    and so the conclusions of Theorem~\ref{supp thm: comparison to Rothe} hold automatically for $\tauhat_{loo}$ and $\tauhat_{loo, cal}$ as well.

\end{proof}

Below we present sufficient conditions for the theorem. We begin with two definitions related to those of \citet{highDimRegAdj}.

\begin{definition}\label{defn: jackknife}
    An estimator $\muhat_{z}$ is ``jackknife compatible" if
    \begin{equation*}
        \Expectation{\sum_{i \st Z_{i} = z}\left(\muhat_{z}^{(-i)}(x_{new}) - \muhat_{z}(x_{new}) \right)^{2}} = o(n_{z})
    \end{equation*}
    for a new data point $x_{new}$ drawn independently of the observed data.
\end{definition}

\begin{definition}\label{defn: risk consistent}
    An estimator $\muhat_{z}$ is ``risk consistent" if 
    \begin{equation*}
        \frac{1}{N}\sum_{i = 1}^{N}\left(\muhat_{z}^{(-i)}(x_{i}) - \overdotmu_{z}(x_{i}) \right)^{2} = o_{P}(1).
    \end{equation*}
\end{definition}

\begin{proposition}
    Assume that $\muhat_{0}$ and $\muhat_{1}$ are jackknife compatible and risk consistent, then Assumption~\ref{app asm: error process vanishes} holds.
\end{proposition}
\begin{proof}
    As before, define
    \begin{equation*}
        R = \sum_{i = 1}^{N}\frac{(-1)^{Z_{i}}}{n_{Z_{i}}}\left(\frac{n_{0}}{N}\left(\muhat_{1}^{(-i)}(x_{i}) - \overdotmu_{1}(x_{i}) \right) + \frac{n_{1}}{N}\left(\muhat_{0}^{(-i)}(x_{i}) - \overdotmu_{0}(x_{i}) \right)\right).
    \end{equation*}
    The argument of \citet[Proof of Theorem 5]{highDimRegAdj} shows that $\Expectation{R^{2}} = o(N^{-1})$.\footnote{In fact, a small modification is needed to adapt the argument of \citet{highDimRegAdj} to our context; simply replace their $\mu^{(z)}$ with $\overdotmu_{z}$ for $z \in \{0, 1\}$.}  In fact, a more detailed examination of their argument shows that both $\Expectation{R_{0}^{2}} = o(N^{-1})$ and $\Expectation{R_{1}^{2}} = o(N^{-1})$ for 
    \begin{align*}
        R_{0} = \sum_{i = 1}^{N}\frac{(-1)^{Z_{i}}}{n_{Z_{i}}}\left(\frac{n_{0}}{N}\left(\muhat_{1}^{(-i)}(x_{i}) - \overdotmu_{1}(x_{i}) \right) \right).\\
        R_{1} = \sum_{i = 1}^{N}\frac{(-1)^{Z_{i}}}{n_{Z_{i}}}\left( \frac{n_{1}}{N}\left(\muhat_{0}^{(-i)}(x_{i}) - \overdotmu_{0}(x_{i}) \right)\right).
    \end{align*}
    Take $z \in \{0, 1\}$.  For any $\varepsilon > 0$ Chebyshev's inequality implies that $${\Prob{N^{1/2}|R_{z}| > \varepsilon} \leq \Expectation{R_{z}^{2}}N\varepsilon^{-2}}.$$  Since $\Expectation{R_{z}^{2}} = o(N^{-1})$ the right-hand-side vanishes as $N \rightarrow \infty$ and so $R_{z} = o_{P}(N^{-1/2})$. Algebraic rearrangement yields that $R_{z} = \frac{1}{N^{1/2}}\left(\G_{N, z}(\overdotmu_{z}) - \G_{N, z}(\muhat_{z}) \right)$; so $\G_{N, z}(\overdotmu_{z}) - \G_{N, z}(\muhat_{z}) = o_{P}(1)$ which concludes the proof.

\end{proof}

\begin{remark}
    Our definition of risk consistency differs sharply from that of \citet{highDimRegAdj} in that we are concerned only with the squared distance between the leave-one-out prediction $\muhat_{z}^{(-i)}(x_{i})$ and some fixed function $\overdotmu_{z}(x_{i})$ where $\overdotmu_{z}$ need not be the true conditional mean of $y_{i}(z)$ given the covariates.  Consequently, even under arbitrary model misspecification $\tauhat_{loo, cal}$ achieves asymptotic non-inferiority to both $\tauhat_{loo}$ and $\tauhat_{unadj}$.  Moreover, the leave-one-out cross-fitting procedure allows one to avoid entropy conditions (cf. \citet{generalizedOB, rothe}).  This allows practitioners to use complex machine learning algorithms -- subject to Definitions~\ref{defn: jackknife} and \ref{defn: risk consistent} -- to form the initial estimators $\muhat_{0}$ and $\muhat_{1}$.  See \citet{highDimRegAdj} for examples of such estimators; one such example is the subsampled random forest estimator of \citet{BART}.
\end{remark}


\begin{remark}
    A further advantage of leave-one-out estimation is in its control of finite sample bias; we discuss this issue in a superpopulation setting and focus on estimators of the form $\tauhat_{rothe} = N^{-1}\sum_{i = 1}^{N} \psi_{i}(\muhat_{0}, \muhat_{1})$.  In general, imputation-based estimators can introduce biases in the sense that $\Expectation{\tauhat_{rothe} - \PATE} \neq 0$.  The central limit theorem for estimators of the form $ N^{-1}\sum_{i = 1}^{N} \psi_{i}(\muhat_{0}, \muhat_{1})$ implies that this bias is asymptotically vanishing.  However, \citet{LoopEstimator} highlights that practical cases exist for which finite sample biases are unacceptable. Under mild conditions, leave-one-out estimation of the form $\tauhat_{loo}$ -- which is equivalent to the LOOP estimator of \citet{LoopEstimator} -- achieves finite sample exact unbiasedness; see \citet{LoopEstimator} or \citet[Corollary 3]{rothe} for proof.  Furthermore, our analyses of $\tauhat_{loo, cal}$ and $\tauhat_{rothe, cal}$ demonstrate that calibration can be applied to simultaneously obtain non-inferiority guarantees for such estimators without any assumption of correct model specification.  Together this implies that under the conditions of \citet[Corollary 3]{rothe} or those of \citet{LoopEstimator} the leave-one-out calibrated estimator $\tauhat_{loo, cal}$ has no finite sample bias and is asymptotically non-inferior to the difference in means.
    
    We include simulations which mirror the superpopulation Poisson regression simulations of Table~\ref{tab: SATE vs CATE vs PATE sampling variances}; we evaluate the performance of two sample-splitting estimators $\tauhat_{loo}$ and $\tauhat_{loo, cal}$.  The results are summarized in Table~\ref{tab: poisson leave-one-out sampling MSEs}.  The first columns highlight that -- in this particular simulation setting -- $\tauhat_{gOB}$ and $\tauhat_{loo}$ are less efficient than the unadjusted difference in means $\tauhat_{unadj}$.  However, calibration of the leave-one-out estimator results in a new estimator $\tauhat_{loo, cal}$ which is asymptotically no less efficient than the difference in means.  The final column highlights that $\tauhat_{loo, cal}$ and $\tauhat_{cal}$ are asymptotically equivalent under the conditions of Theorem~\ref{supp thm: crossfitting to avoid entropy conditions}.  Both $\tauhat_{loo}$ and $\tauhat_{loo, cal}$ are guaranteed to have zero bias \citep{LoopEstimator, rothe}, so the third column of Table~\ref{tab: poisson leave-one-out sampling MSEs} demonstrates that via calibration we can construct an estimator with both desirable robustness properties: unbiasedness and asymptotic non-inferiority.
    
    \begin{table}[h]
    \def~{\hphantom{0}}
    \caption{Ratios of Monte Carlo mean-square-errors for $\tauhat_{cal}$, $\tauhat_{gOB}$, $\tauhat_{loo}$, and $\tauhat_{loo, cal}$ to the difference in means estimator $\hat{\tau}_{unadj}$ for various experiment sizes $N$.  Each mean-square-error is based upon $B = 100$ simulated experiments.  Results are averaged over $S = 100$ simulations.}{
    \begin{tabular}{lccccc}
                & $\frac{\hat{\text{MSE}}(\tauhat_{gOB})}{ \hat{\text{MSE}}(\tauhat_{unadj})}$ & $\frac{\hat{\text{MSE}}(\tauhat_{loo}) }{ \hat{\text{MSE}}(\tauhat_{unadj})}$  & $\frac{\hat{\text{MSE}}(\tauhat_{loo, cal}) }{ \hat{\text{MSE}}(\tauhat_{unadj})}$ &  $\frac{\hat{\text{MSE}}(\tauhat_{cal}) }{ \hat{\text{MSE}}(\tauhat_{unadj})}$&  \\
        $N = 200$   &  1.131  &  1.139  &  0.950  & 0.949  &  \\
        $N = 500$   &  1.119  &  1.121  &  0.948  & 0.948  &
    \end{tabular}
    }
    \label{tab: poisson leave-one-out sampling MSEs}
    \end{table}
\end{remark}

To illustrate the benefits of calibration even when flexible machine learning methods are used for imputation, below we present a simulation using random forests as the original imputation functions $\muhat_{0}$ and $\muhat_{1}$.  In this simulation, we used the \texttt{Boston} data set from \citet{MASS}.  In each simulation of we uniformly sampled $N$ observations with replacement from the \texttt{Boston} data set; crime was used as the outcome of interest and the remaining information was taken as features.  Independently, we drew a treatment allocation $Z \sim \textit{Unif}(\Omega)$ with $n_{1}=\lfloor pN\rfloor$ for $p = 0.6$.  By construction, Fisher's sharp null holds.  We use a sample splitting estimation procedure wherein prediction functions $\muhat_{0}$ and $\muhat_{1}$ are trained on half of the data and then applied to the other half to form the treatment effect estimator; the roles of the training and testing sets are interchanged and the results are averaged together to form the final estimator.  This follows the 2-fold procedure of \citet{highDimRegAdj}.  The original prediction functions $\muhat_{0}$ and $\muhat_{1}$ are random forests.  We compare the results against the analogous calibrated estimator which linearly calibrates the prediction functions $\muhat_{0}$ and $\muhat_{1}$ within the training sets; Table~\ref{tab: random forest sample splitting variances_HF} summarizes the results.  In particular, although both estimators achieve the semiparametric efficiency bound in the limit the calibrated estimator displays smaller variances across the simulations of Table~\ref{tab: random forest sample splitting variances_HF}.  Table~\ref{tab: random forest sample splitting variances_HN} repeats the simulations of Table~\ref{tab: random forest sample splitting variances_HF} but incorporates treatment effect heterogeneity by adding independent exponential noise with rate 1 to the treated outcomes and subtracting independent exponential noise with rate 1 from the control outcomes.  The ties between the variances of the calibrated and uncalibrated estimators presented in Table~\ref{tab: random forest sample splitting variances_HN} demonstrates the benign effect of calibration when the original uncalibrated estimator was performing well to begin with. 
    \begin{table}[h]
    \centering
    \begin{tabular}{lcc}
                    & $\hat{\text{var}}(\sqrt{N}\tauhat_{SSRF})$ & $\hat{\text{var}}(\sqrt{N}\tauhat_{SSRF, cal})$  \\
            $N = 500$       &  137.27   &  134.98  \\
            $N = 1000$      &  118.19   &  116.13  \\
            $N = 2000$      &  83.02    &  80.50  \\
    \end{tabular}
    \caption{\textbf{(Sharp Null Simulations)} Monte Carlo variances for the sample-split random forest imputation estimator, $\hat{\text{var}}(\sqrt{N}\tauhat_{SSRF})$, and its calibrated analogue, $\hat{\text{var}}(\sqrt{N}\tauhat_{SSRF, cal})$, for various experiment sizes $N$.  Each variance is based upon $1000$ simulated experiments.}
    \label{tab: random forest sample splitting variances_HF}
    \end{table}
    
   \begin{table}[h]
    \centering
    \begin{tabular}{lcc}
                    & $\hat{\text{var}}(\sqrt{N}\tauhat_{SSRF})$ & $\hat{\text{var}}(\sqrt{N}\tauhat_{SSRF, cal})$  \\
            $N = 500$       &  139.73   &  139.08  \\
            $N = 1000$      &  119.91   &  119.47  \\
            $N = 2000$      &  90.11    &  90.16  \\
    \end{tabular}
    \caption{\textbf{(Weak Null Simulations)} Monte Carlo variances for the sample-split random forest imputation estimator, $\hat{\text{var}}(\sqrt{N}\tauhat_{SSRF})$, and its calibrated analogue, $\hat{\text{var}}(\sqrt{N}\tauhat_{SSRF, cal})$, for various experiment sizes $N$.  Each variance is based upon $1000$ simulated experiments.}
    \label{tab: random forest sample splitting variances_HN}
    \end{table}
\color{black}

\section{An Alternative Framework via Entropy Conditions}
Assumption~\ref{app asm: error process vanishes} is sufficient for our results, but may be challengingly abstract.  For the sake of clarity in our explication, we include some auxiliary results based upon Assumption~\ref{supp asm: simple realizations} which demonstrate how to work directly with entropy conditions in the proofs of this paper.  In the process, we establish several results which demonstrate that entropy conditions and Vapnik–Chervonenkis dimension conditions are sufficient for Assumption~\ref{app asm: error process vanishes}.  We start with a few technical lemmas.
\subsection{Some technical lemmas on entropy conditions}
\begin{lemma}\label{lem: radius is arbitrary}
    For a sequence of function classes $\left\{\mathscr{F}_{N} \right\}_{N \in \N}$ suppose that  $\int_{0}^{1}\sup_{N}\sqrt{\log  \mathscr{N}(\mathscr{F}_{N},||\cdot||_{N}, s)}\, ds < \infty$.  Then it follows that for any $D > 0$ 
    \begin{equation*}
        \int_{0}^{1}\sup_{N}\sqrt{\log  \mathscr{N}\left(\mathscr{F}_{N},||\cdot||_{N}, \frac{s}{D}\right)}\, ds < \infty.
    \end{equation*}
\end{lemma}
\begin{proof}
    By the change of variables $s = D^{-1}t$ 
    \begin{equation*}
        \int_{0}^{1}\sup_{N}\sqrt{\log  \mathscr{N}\left(\mathscr{F}_{N},||\cdot||_{N}, \frac{s}{D}\right)}\, ds = \frac{1}{D}\int_{0}^{D}\sup_{N}\sqrt{\log  \mathscr{N}\left(\mathscr{F}_{N},||\cdot||_{N}, t\right)}\, dt.
    \end{equation*}
    We break the proof into two pieces: $D \geq 1$ versus $D \in (0,1)$.  For now we focus on the first case, so assume that $D \geq 1$.  The result is trivial when $D = 1$, so take $D > 1$ and let $\lfloor D \rfloor$ denote the greatest integer below $D$.  It follows that 
    \begin{multline*}
        \int_{0}^{D}\sup_{N}\sqrt{\log  \mathscr{N}\left(\mathscr{F}_{N},||\cdot||_{N}, t\right)}\, dt = \int_{0}^{1}\sup_{N}\sqrt{\log  \mathscr{N}\left(\mathscr{F}_{N},||\cdot||_{N}, t\right)}\, dt +\\ \int_{1}^{\lfloor D \rfloor}\sup_{N}\sqrt{\log  \mathscr{N}\left(\mathscr{F}_{N},||\cdot||_{N}, t\right)}\, dt + \\
         \int_{\lfloor D \rfloor}^{D}\sup_{N}\sqrt{\log  \mathscr{N}\left(\mathscr{F}_{N},||\cdot||_{N}, t\right)}\, dt.
    \end{multline*}
    The first term is guaranteed to be finite by assumption.  To bound the second term, notice that $\mathscr{N}\left(\mathscr{F}_{N},||\cdot||_{N}, t\right)$ is a non-increasing function of $t$ and so \begin{equation}\label{ineq: 0 to 1 is biggest}
        \int_{0}^{1}\sup_{N}\sqrt{\log  \mathscr{N}\left(\mathscr{F}_{N},||\cdot||_{N}, t\right)}\, dt \geq \int_{\ell}^{\ell + 1}\sup_{N}\sqrt{\log  \mathscr{N}\left(\mathscr{F}_{N},||\cdot||_{N}, t\right)}\, dt.
    \end{equation}
    for any $\ell \in \N$.  Thus, the second term is no greater than ${(\lfloor D \rfloor - 1)\int_{0}^{1}\sup_{N}\sqrt{\log  \mathscr{N}\left(\mathscr{F}_{N},||\cdot||_{N}, t\right)}\, dt}$ which is certain to be finite.  Lastly, because $\mathscr{N}\left(\mathscr{F}_{N},||\cdot||_{N}, t\right) \geq 1$ it follows that the integrand of the third term is non-negative and so 
    \begin{equation*}
        \int_{\lfloor D \rfloor}^{D}\sup_{N}\sqrt{\log  \mathscr{N}\left(\mathscr{F}_{N},||\cdot||_{N}, t\right)}\, dt \leq \int_{\lfloor D \rfloor}^{\lfloor D \rfloor  + 1}\sup_{N}\sqrt{\log  \mathscr{N}\left(\mathscr{F}_{N},||\cdot||_{N}, t\right)}\, dt;
    \end{equation*}
    the right-hand-side of this inequality is finite by \eqref{ineq: 0 to 1 is biggest}.  Consequently, when $D \geq 1$ the desired result holds.
    
    Now suppose that $D \in (0,1)$.  As noted before, the integrand $\sup_{N}\sqrt{\log  \mathscr{N}\left(\mathscr{F}_{N},||\cdot||_{N}, t\right)}$ is non-negative and so
    \begin{equation*}
        \int_{0}^{D}\sup_{N}\sqrt{\log  \mathscr{N}\left(\mathscr{F}_{N},||\cdot||_{N}, t\right)}\, dt \leq \int_{0}^{1}\sup_{N}\sqrt{\log  \mathscr{N}\left(\mathscr{F}_{N},||\cdot||_{N}, t\right)}\, dt ;
    \end{equation*}
    the right-hand-side of this inequality is finite by assumption.  Thus, when $D \in (0, 1)$ the desired result holds, thereby concluding the proof.
\end{proof}

\begin{lemma}\label{lem: marginally simple realizations implies jointly simple}
    Suppose that the prediction functions $\muhat_{0}(\cdot)$ and $\muhat_{1}(\cdot)$ have typically simple realizations (Assumption~\ref{supp asm: simple realizations}). Then, so too does the joint prediction function
    \begin{equation*}
        x \mapsto \begin{bmatrix}\muhat_{0}(x) \\ \muhat_{1}(x) \end{bmatrix}.
    \end{equation*}
\end{lemma}
\begin{proof}
    Say that $\muhat_{0}$ satisfies Assumption~\ref{supp asm: simple realizations} for the sequence of function classes $\left\{\tsrC \right\}_{N \in \N}$ and $\muhat_{1}$ satisfies Assumption~\ref{supp asm: simple realizations} for the sequence of function classes $\left\{\tsrT \right\}_{N \in \N}$.
    
    Consider the class of functions $\left\{\mathscr{C}_{N} \right\}_{N \in \N}$ for $\mathscr{C}_{N} = \tsrC \times \tsrT$.  
    
    By assumption
    \begin{equation*}
        \Prob{\muhat_{0} \in \tsrC } \rightarrow 1 \quad\text{and}\quad \Prob{\muhat_{1} \in \tsrT } \rightarrow 1.
    \end{equation*}
    Because $\Prob{\begin{bmatrix}\muhat_{0}(\cdot) \\ \muhat_{1}(\cdot) \end{bmatrix} \in \mathscr{C}_{N}} \geq 1 - \Prob{\muhat_{0} \not\in \tsrC } - \Prob{\muhat_{1} \not\in \tsrT }$ it follows that
    $$
        \Prob{\begin{bmatrix}\muhat_{0}(\cdot) \\ \muhat_{1}(\cdot) \end{bmatrix} \in \mathscr{C}_{N}} \rightarrow 1.
    $$

    Now, we show that 
    $$
        \int_{0}^{1}\sup_{N}\sqrt{\log  \mathscr{N}(\mathscr{C}_{N},||\cdot||_{N}, s)}\, ds < \infty.
    $$
    
    Consider $\begin{bmatrix}{\mu}_{0}(\cdot) \\ {\mu}_{1}(\cdot) \end{bmatrix},  \begin{bmatrix}{\nu}_{0}(\cdot) \\ {\nu}_{1}(\cdot) \end{bmatrix} \in \mathscr{C}_{n}$; by definition 
    \begin{align}
        \left| \left| \begin{bmatrix}{\mu}_{0}(\cdot) \\ {\mu}_{1}(\cdot) \end{bmatrix} -  \begin{bmatrix}{\nu}_{0}(\cdot) \\ {\nu}_{1}(\cdot) \end{bmatrix} \right|\right|_{N} &= \left(\frac{1}{N}\sum_{i = 1}^{N} \left| \left| \begin{bmatrix}{\mu}_{0}(x_i) \\ {\mu}_{1}(x_i) \end{bmatrix} -  \begin{bmatrix}{\nu}_{0}(x_i) \\ {\nu}_{1}(x_i) \end{bmatrix} \right|\right|_{2}^{2} \right)^{1/2}\nonumber\\
        &= \left[\frac{1}{N}\sum_{i = 1}^{N} \left\{ \left({\mu}_{0}(x_i) -  {\nu}_{0}(x_i) \right)^{2} + \left({\mu}_{1}(x_i) - {\nu}_{1}(x_i) \right)^{2} \right\}\right]^{1/2}\nonumber\\
        &= \left[\frac{1}{N}\sum_{i = 1}^{N} \left\{{\mu}_{0}(x_i) -  {\nu}_{0}(x_i) \right\}^{2} + \frac{1}{N}\sum_{i = 1}^{N}\left\{{\mu}_{1}(x_i) - {\nu}_{1}(x_i) \right\}^{2}\right]^{1/2}\nonumber\\
        &\leq \left[\frac{1}{N}\sum_{i = 1}^{N} \left\{{\mu}_{0}(x_i) -  {\nu}_{0}(x_i) \right\}^{2}\right]^{1/2} + \left[\frac{1}{N}\sum_{i = 1}^{N}\left\{{\mu}_{1}(x_i) - {\nu}_{1}(x_i) \right\}^{2}\right]^{1/2}\label{ineq: subadditive sqrt}\\
        &= \left| \left| {\mu}_{0}(\cdot) -  {\nu}_{0}(\cdot)\right| \right|_{N} +  \left| \left| {\mu}_{1}(\cdot) -  {\nu}_{1}(\cdot)\right| \right|_{N}\nonumber
    \end{align}
    where \eqref{ineq: subadditive sqrt} follows by the subadditivity of the square root.
    
    Consequently, $\mathscr{N}(\mathscr{C}_{N},||\cdot||_{N}, s) \leq \mathscr{N}(\tsrC,||\cdot||_{N}, s/2) \mathscr{N}(\tsrT,||\cdot||_{N}, s/2)$ and so the monotonicity of the logarithm gives that
    \begin{align*}
        \int_{0}^{1}\sup_{N}\sqrt{\log  \mathscr{N}(\mathscr{C}_{N},||\cdot||_{N}, s)}\, ds &\leq \int_{0}^{1}\sup_{N}\sqrt{\log  \left(\mathscr{N}(\tsrC,||\cdot||_{N}, s/2) \mathscr{N}(\tsrT,||\cdot||_{N}, s/2)\right)}\, ds\\
        &= \int_{0}^{1}\sup_{N}\sqrt{\log  \mathscr{N}(\tsrC,||\cdot||_{N}, s/2) + \log \mathscr{N}(\tsrT,||\cdot||_{N}, s/2)}\, ds\\
        &\leq \int_{0}^{1}\sup_{N}\left( \sqrt{\log  \mathscr{N}(\tsrC,||\cdot||_{N}, s/2)} + \sqrt{\log \mathscr{N}(\tsrT,||\cdot||_{N}, s/2)}\right)\, ds\\
        &\leq \begin{multlined}[t][10.5cm] \int_{0}^{1}\sup_{N}\sqrt{\log  \mathscr{N}(\tsrC,||\cdot||_{N}, s/2)} +\\ \int_{0}^{1}\sup_{N} \sqrt{\log \mathscr{N}(\tsrT,||\cdot||_{N}, s/2)}\, ds.
        \end{multlined}
    \end{align*}
    The last line is guaranteed to be finite exactly because $\muhat_{0}$ satisfies Assumption~\ref{supp asm: simple realizations} for the sequence of function classes $\left\{\tsrC \right\}_{N \in \N}$ and $\muhat_{1}$ satisfies Assumption~\ref{supp asm: simple realizations} for the sequence of function classes $\left\{\tsrT \right\}_{N \in \N}$.
\end{proof}

\begin{remark}[Multiplicative Bounds for Covering Numbers]\label{rem: mult bound for covering number}
    In the proof of Lemma~\ref{lem: marginally simple realizations implies jointly simple} we remarked that
    \begin{equation}\label{ineq: norm bound}
         \left| \left| \begin{bmatrix}{\mu}_{0}(\cdot) \\ {\mu}_{1}(\cdot) \end{bmatrix} -  \begin{bmatrix}{\nu}_{0}(\cdot) \\ {\nu}_{1}(\cdot) \end{bmatrix} \right|\right|_{N} \leq  \left| \left| {\mu}_{0}(\cdot) -  {\nu}_{0}(\cdot)\right| \right|_{N} +  \left| \left| {\mu}_{1}(\cdot) -  {\nu}_{1}(\cdot)\right| \right|_{N}
    \end{equation}
    implies that 
    \begin{equation}\label{ineq: covering number bound}
        \mathscr{N}(\mathscr{C}_{N},||\cdot||_{N}, s) \leq \mathscr{N}(\tsrC,||\cdot||_{N}, s/2) \mathscr{N}(\tsrT,||\cdot||_{N}, s/2).
    \end{equation}
    This reasoning plays an important role in our proofs and is of independent interest since covering numbers play an significant role in numerous areas of probability.
    
    To avoid triviality, assume that both $\tsrC$ and $\tsrT$ are non-empty.  Suppose that the points $a_{1}, \ldots, a_{\ell}$ provide a minimal cardinality $(s/2)$-cover of $\tsrC$ and the points $b_{1}, \ldots, b_{\ell'}$ provide a minimal cardinality $(s/2)$-cover of $\tsrT$.  Consider the set of points $C = \{a_{1}, \ldots, a_{\ell}\} \times \{b_{1}, \ldots, b_{\ell'}\}$; this set has cardinality $\ell \ell'$.  Fix a point $c = (a,b) \in \mathscr{C}_{N}$.  Since $a_{1}, \ldots, a_{\ell}$ and $b_{1}, \ldots, b_{\ell'}$ are $(s/2)$-covers there exists at least one point $(a_{i}, b_{j})$ for which
    \begin{align*}
        \left| \left| a -  a_{i}\right| \right|_{N} &\leq \frac{s}{2}\\
        \left| \left| b -  b_{j}\right| \right|_{N} &\leq \frac{s}{2}.
    \end{align*}
    By \eqref{ineq: norm bound} $||c - (a_{i}, b_{j})^{\T}||_{N}$ is bounded above by 
    \begin{equation*}
        \left| \left| a -  a_{i}\right| \right|_{N} +  \left| \left| b -  b_{j}\right| \right|_{N}
    \end{equation*}
    which is bounded above by $s$ due to our choice of $a_{i}$ and $b_{j}$.  This implies that $C$ is a valid $s$-cover of $\mathscr{C}_{N}$ which has cardinality $\ell\ell'$.  Since $C$ is a feasible $s$-cover it follows that the minimal cardinality $s$-cover of $\mathscr{C}_{N}$ must have cardinality no greater than $\ell\ell'$.  However, since $\{a_{1}, \ldots, a_{\ell}\}$ is a minimal cardinality $(s/2)$-cover of $\tsrC$ it follows that $\ell = \mathscr{N}(\tsrC,||\cdot||_{N}, s/2)$; likewise $\ell' = \mathscr{N}(\tsrT,||\cdot||_{N}, s/2)$.  Consequently \eqref{ineq: covering number bound} holds.
    
    Iterating the logic above implies the following theorem.
    
    \renewcommand{\thetheorem}{A.\arabic{theorem}} 
    \begin{theorem}\label{supp thm: submultiplicative covering numbers}
        Suppose that the set $T \subseteq T_{1} \times \cdots \times T_{\ell}$; let $\pi_{i}: T \rightarrow T_{i}$ be the $i$\textsuperscript{th} coordinate projection.  Suppose that $(T, d)$ is a metric space such that 
        \begin{equation*}
            d(a, b) \leq \sum_{i = 1}^{\ell}d_{i}\left(\pi_{i}(a), \pi_{i}(b)\right)
        \end{equation*}
        where $d_{i}(\cdot, \cdot)$ denotes some metric on $T_{i}$.  Then
        \begin{equation}
            \mathscr{N}(T,d, s) \leq \prod_{i = 1}^{\ell} \mathscr{N}\left(T_{i},d_{i}, \frac{s}{\ell}\right).
        \end{equation}
    \end{theorem}
\end{remark}

\begin{proposition}\label{prop: typically simple realizations are inherited}
    If the prediction functions $\muhat_{0}$ and $\muhat_{1}$ are stable, have typically simple realizations, and satisfy Assumptions~\ref{supp asm: means and covs stabilize} and \ref{supp asm: bounded fourth moment}, then the prediction functions $\muhat_{OLS, 0}$ and $\muhat_{OLS, 1}$ also have typically simple realizations.
\end{proposition}
\begin{proof}
    Say that $\muhat_{0}$ satisfies Assumption~\ref{supp asm: simple realizations} for the sequence of function classes $\left\{\tsrC \right\}_{N \in \N}$ and $\muhat_{1}$ satisfies Assumption~\ref{supp asm: simple realizations} for the sequence of function classes $\left\{\tsrT \right\}_{N \in \N}$.  Let  ${\mathscr{C}_{N} = \tsrC \times \tsrT}$.  Define the sequence of function classes $\left\{\mathscr{F}_{N} \right\}_{N \in \N}$ via
    \begin{multline}\label{eqn: definition of F_n}
        \mathscr{F}_{N} = \Bigg\{\mu_{\beta_{0}, \beta_{1}}({x}) = \beta_{0} + \beta_{1}^{\T}\begin{bmatrix}
                        \mu_{0}({x})\\
                        \mu_{1}({x})\\
                        \end{bmatrix} \; \Bigg| \; |\beta_{0} - \dot{\beta}_{0}^{(N)}| \leq 1; \; ||\beta_{1} - \dot{\beta}_{1}^{(N)}||_{2} \leq 1;\\ \left| \left| \begin{bmatrix}
                        \mu_{0}\\
                        \mu_{1}\\
                        \end{bmatrix} - \begin{bmatrix}
                        \dot{\mu}_{0}\\
                        \dot{\mu}_{1}\\
                        \end{bmatrix} \right|\right|_{N} \leq 1;\;
                        \begin{bmatrix}
                        \mu_{0}\\
                        \mu_{1}\\
                        \end{bmatrix}\in \mathscr{C}_{N}\Bigg\}
    \end{multline}
    
    First, we show that  $\muhat_{OLS, 0}$ and $\muhat_{OLS, 1}$ are asymptotically almost surely elements of $\mathscr{F}_{N}$.  Since the proofs are basically the same for both $\muhat_{OLS, 0}$ and $\muhat_{OLS, 1}$, we present the proof only for $\muhat_{OLS, 1}$ and use the notation of \eqref{eqn: useful notation sample beta} and \eqref{eqn: useful notation ``population`` beta}.  By the consistency of the ordinary least squares linear regression coefficients, Lemma~\ref{lem: consistency of OLS coefficients}, it follows that 
    \begin{equation}\label{eqn: OLS coefficients are asymptotically almost surely close}
        \lim_{N \rightarrow \infty}\Prob{|\hat{\beta}_{0} - \dot{\beta}_{0}^{(N)}| > 1\;\; \& \;\;||\hat{\beta}_{1} - \dot{\beta}_{1}^{(N)}||_{2} > 1} = 0.
    \end{equation}
    
    By the joint stability of $\muhat_{0}$ and $\muhat_{1}$ 
    \begin{equation}\label{eqn: predictor functions are asymptotically almost surely close}
        \lim_{N \rightarrow \infty}\Prob{\left| \left| \begin{bmatrix}
                        \muhat_{0}\\
                        \muhat_{1}\\
                        \end{bmatrix} - \begin{bmatrix}
                        \dot{\mu}_{0}\\
                        \dot{\mu}_{1}\\
                        \end{bmatrix} \right|\right|_{N} > 1} = 0.
    \end{equation}

    Since $\muhat_{0}$ and $\muhat_{1}$ have are typically simple with respect to $\left\{\tsrC \right\}_{N \in \N}$ and $\left\{\tsrT \right\}_{N \in \N}$, respectively, it follows from Lemma~\ref{lem: marginally simple realizations implies jointly simple} that 
    \begin{equation}\label{eqn: joint stability}
        \lim_{N \rightarrow \infty} \Prob{\begin{bmatrix}\muhat_{0}(\cdot) \\ \muhat_{1}(\cdot) \end{bmatrix} \in \mathscr{C}_{N}} = 1.
    \end{equation}
    
    By Boole's inequality
    \begin{multline*}
        \Prob{\muhat_{OLS, 1} \in \mathscr{F}_{N}} \geq 1 - \Prob{|\hat{\beta}_{0} - \dot{\beta}_{0}^{(N)}| > 1\;\; \& \;\;||\hat{\beta}_{1} - \dot{\beta}_{1}^{(N)}|| > 1} - \\ \Prob{\left| \left| \begin{bmatrix}
                        \muhat_{0}\\
                        \muhat_{1}\\
                        \end{bmatrix} - \begin{bmatrix}
                        \dot{\mu}_{0}\\
                        \dot{\mu}_{1}\\
                        \end{bmatrix} \right|\right|_{N} > 1} -       \Prob{\begin{bmatrix}\muhat_{0}(\cdot) \\ \muhat_{1}(\cdot) \end{bmatrix} \not\in \mathscr{C}_{N}}
    \end{multline*}
    for each $N \in \N$, so it follows from \eqref{eqn: OLS coefficients are asymptotically almost surely close}, \eqref{eqn: predictor functions are asymptotically almost surely close}, and \eqref{eqn: joint stability} that $\lim_{N \rightarrow \infty}\Prob{\muhat_{OLS, 1} \in \mathscr{F}_{N}} = 1$.  A mirrored proof yields that $\lim_{N \rightarrow \infty}\Prob{\muhat_{OLS, 0} \in \mathscr{F}_{N}} = 1$.

    All that remains to be shown is that
    \begin{equation*}
        \int_{0}^{1}\sup_{N}\sqrt{\log  \mathscr{N}(\mathscr{F}_{N},||\cdot||_{N}, s)}\, ds < \infty.
    \end{equation*}
    To start, we examine two functions $f, g \in \mathscr{F}_{N}$ defined by
    \begin{align*}
        f({x}) &= \beta_{0f} + \beta_{1f}^{\T}\begin{bmatrix}
                                                        \mu_{0f}(x) \\ 
                                                        \mu_{1f}(x)
                                                    \end{bmatrix} \\
        g({x}) &= \beta_{0g} + \beta_{1g}^{\T}\begin{bmatrix}
                                                        \mu_{0g}(x) \\ 
                                                        \mu_{1g}(x)
                                                    \end{bmatrix} .                  
    \end{align*}
    The norm of their difference is
    \begin{align}
        \left|\left|f - g \right|\right|_{N} &= \left|\left|\left(\beta_{0f} - \beta_{0g} \right) + \left(\beta_{1f}^{\T}\begin{bmatrix}
                                                        \mu_{0f} \\ 
                                                        \mu_{1f}
                                                    \end{bmatrix} - \beta_{1g}^{\T}\begin{bmatrix}
                                                        \mu_{0g} \\ 
                                                        \mu_{1g}
                                                    \end{bmatrix}\right) \right|\right|_{N}\nonumber\\
                                            &\leq \left|\left|\beta_{0f} - \beta_{0g}\right|\right|_{N} + \left|\left| \beta_{1f}^{\T}\begin{bmatrix}
                                                        \mu_{0f} \\ 
                                                        \mu_{1f}
                                                    \end{bmatrix} - \beta_{1g}^{\T}\begin{bmatrix}
                                                        \mu_{0g} \\ 
                                                        \mu_{1g}
                                                    \end{bmatrix} \right|\right|_{N}\label{ineq: triangle ineq}\\
                                            &= \left|\beta_{0f} - \beta_{0g}\right| + \left|\left| \beta_{1f}^{\T}\begin{bmatrix}
                                                        \mu_{0f} \\ 
                                                        \mu_{1f}
                                                    \end{bmatrix} - \beta_{1g}^{\T}\begin{bmatrix}
                                                        \mu_{0g} \\ 
                                                        \mu_{1g}
                                                    \end{bmatrix} \right|\right|_{N}\label{eqn: defn of n norm}
    \end{align}
    where \eqref{ineq: triangle ineq} is due to the triangle inequality and \eqref{eqn: defn of n norm} follows from the definition of $||\cdot||_{N}$ for constant functions.  Furthermore
    \begin{align}
        &\left|\left| \beta_{1f}^{\T}\begin{bmatrix}
                        \mu_{0f} \\ 
                        \mu_{1f}
                    \end{bmatrix} - \beta_{1g}^{\T}\begin{bmatrix}
                        \mu_{0g} \\ 
                        \mu_{1g}
                    \end{bmatrix} \right|\right|_{N} \nonumber \\
                    &= \left|\left| \beta_{1f}^{\T}\begin{bmatrix}
                                                        \mu_{0f} \\ 
                                                        \mu_{1f}
                                                    \end{bmatrix} - \beta_{1f}^{\T}\begin{bmatrix}
                                                        \mu_{0g} \\ 
                                                        \mu_{1g}
                                                    \end{bmatrix} + \beta_{1f}^{\T}\begin{bmatrix}
                                                        \mu_{0g} \\ 
                                                        \mu_{1g}
                                                    \end{bmatrix} - \beta_{1g}^{\T}\begin{bmatrix}
                                                        \mu_{0g} \\ 
                                                        \mu_{1g}
                                                    \end{bmatrix} \right|\right|_{N}\nonumber\\
                &\leq \left|\left| \beta_{1f}^{\T}\begin{bmatrix}
                                                        \mu_{0f} \\ 
                                                        \mu_{1f}
                                                    \end{bmatrix} - \beta_{1f}^{\T}\begin{bmatrix}
                                                        \mu_{0g} \\ 
                                                        \mu_{1g}
                                                    \end{bmatrix}\right|\right|_{N} + \left|\left|\beta_{1f}^{\T}\begin{bmatrix}
                                                        \mu_{0g} \\ 
                                                        \mu_{1g}
                                                    \end{bmatrix} - \beta_{1g}^{\T}\begin{bmatrix}
                                                        \mu_{0g} \\ 
                                                        \mu_{1g}
                                                    \end{bmatrix} \right|\right|_{N}\label{ineq: triangle ineq again}\\
                &\leq \left|\left|\beta_{1f}\right|\right|_{N}\left|\left| \begin{bmatrix}
                                                        \mu_{0f} \\ 
                                                        \mu_{1f}
                                                    \end{bmatrix} - \begin{bmatrix}
                                                        \mu_{0g} \\ 
                                                        \mu_{1g}
                                                    \end{bmatrix}\right|\right|_{N} + \left|\left|\begin{bmatrix}
                                                        \mu_{0g} \\ 
                                                        \mu_{1g}
                                                    \end{bmatrix}\right|\right|_{N}\left|\left|\beta_{1f} - \beta_{1g}\right|\right|_{N}\label{eqn: CS ineq}\\
                &= \left|\left|\beta_{1f}\right|\right|_{2}\left|\left| \begin{bmatrix}
                                                        \mu_{0f} \\ 
                                                        \mu_{1f}
                                                    \end{bmatrix} - \begin{bmatrix}
                                                        \mu_{0g} \\ 
                                                        \mu_{1g}
                                                    \end{bmatrix}\right|\right|_{N} + \left|\left|\begin{bmatrix}
                                                        \mu_{0g} \\ 
                                                        \mu_{1g}
                                                    \end{bmatrix}\right|\right|_{N}\left|\left|\beta_{1f} - \beta_{1g}\right|\right|_{2}\label{eqn: final line of inequality}
    \end{align}
    where \eqref{ineq: triangle ineq again} is due to the triangle inequality, \eqref{eqn: CS ineq} is due to the Cauchy-Schwarz inequality, and \eqref{eqn: final line of inequality} is due to the definition of $||\cdot||_{N}$ for constant functions.  Because $f, g \in \mathscr{F}_{N}$ 
    \begin{equation*}
        \left|\beta_{1f} - \dot{\beta}_{1}^{(N)}\right| \leq 1  \quad \text{and} \quad \left|\left|\begin{bmatrix}
                                                        \mu_{0g} \\ 
                                                        \mu_{1g}
                                                    \end{bmatrix} - \begin{bmatrix}
                                                        \overdotmu_{0} \\ 
                                                        \overdotmu_{1}
                                                    \end{bmatrix}\right|\right|_{N} \leq 1
    \end{equation*}
    so 
    \begin{equation*}
        \left|\left|\beta_{1f}\right|\right|_{2} \leq 1 + \left|\left|\dot{\beta}_{1}^{(N)}\right|\right|_{2} \quad \text{and} \quad \left|\left|\begin{bmatrix}
                                                        \mu_{0g} \\ 
                                                        \mu_{1g}
                                                    \end{bmatrix}\right|\right|_{N} \leq 1 + \left|\left|\begin{bmatrix}
                                                        \overdotmu_{0} \\ 
                                                        \overdotmu_{1}
                                                    \end{bmatrix}\right|\right|_{N}.
    \end{equation*}
    Thus \eqref{eqn: final line of inequality} can be bounded above by
    \begin{equation*}
        \left(1 + \left|\left|\dot{\beta}_{1}^{(N)}\right|\right|_{2}\right)\left|\left| \begin{bmatrix}
                                                        \mu_{0f} \\ 
                                                        \mu_{1f}
                                                    \end{bmatrix} - \begin{bmatrix}
                                                        \mu_{0g} \\ 
                                                        \mu_{1g}
                                                    \end{bmatrix}\right|\right|_{N} 
            + \left(1 + \left|\left|\begin{bmatrix}
                \overdotmu_{0} \\ 
                \overdotmu_{1}
            \end{bmatrix}\right|\right|_{N} \right)\left|\left|\beta_{1f} - \beta_{1g}\right|\right|_{2}.
    \end{equation*}
    By Assumptions~\ref{supp asm: means and covs stabilize} and \ref{supp asm: bounded fourth moment} and standard ordinary least squares regression results 
    $\left|\left|\dot{\beta}_{1}^{(N)}\right|\right|_{2}$ is bounded uniformly in $N$.  By Assumptions~\ref{supp asm: means and covs stabilize} and \ref{supp asm: bounded fourth moment} and Lemma~\ref{lem: uniformly bounded N-norm} the quantity $\left|\left|\begin{bmatrix}
                                                        \overdotmu_{0} \\ 
                                                        \overdotmu_{1}
                                                    \end{bmatrix}\right|\right|_{N}$ is also bounded uniformly in $N$, it follows that \eqref{eqn: final line of inequality} is bounded above by
    \begin{equation}\label{eqn: upper bound second term}
        \kappa\left(\left|\left| \begin{bmatrix}
                                                        \mu_{0f} \\ 
                                                        \mu_{1f}
                                                    \end{bmatrix} - \begin{bmatrix}
                                                        \mu_{0g} \\ 
                                                        \mu_{1g}
                                                    \end{bmatrix}\right|\right|_{N} + \left|\left|\beta_{1f} - \beta_{1g}\right|\right|_{2}\right)
    \end{equation}
    for some $\kappa$ which does not depend upon $N$.  Combining \eqref{eqn: defn of n norm} with \eqref{eqn: upper bound second term} yields that
    \begin{align}
         \left|\left|f - g \right|\right|_{N} &\leq 
                                                D\left|\beta_{0f} - \beta_{0g}\right| + D\left|\left|\beta_{1f}^{\T} - \beta_{1g}^{\T} \right|\right|_{2}+  D\left|\left| \begin{bmatrix}
                                                        \mu_{0f} \\ 
                                                        \mu_{1f}
                                                    \end{bmatrix} - \begin{bmatrix}
                                                        \mu_{0g} \\ 
                                                        \mu_{1g}
                                                    \end{bmatrix}\right|\right|_{N} \label{ineq: distance between functions bound}
    \end{align}
    for $D = \max\{1, \kappa\}$ which does not depend upon $N$.
    
    Using \eqref{ineq: distance between functions bound} and Theorem~\ref{supp thm: submultiplicative covering numbers} we can bound the $s$-covering number of $\mathscr{F}_{N}$ as  
    \begin{multline}\label{ineq: covering number initial bound}
        \mathscr{N}(\mathscr{F}_{N},||\cdot||_{N}, s) \leq \mathscr{N}\left(\mathcal{B}(\dot{\beta}_{0}^{(N)}),|\cdot|, \frac{s}{3D}\right) \times\\ \mathscr{N}\left(\mathcal{B}(\dot{\beta}_{1}^{(N)}),||\cdot||_{2}, \frac{s}{3D}\right) \times \mathscr{N}\left(\mathscr{C}_{N},||\cdot||_{N}, \frac{s}{3D}\right)
    \end{multline}
    where $\mathcal{B}(\dot{\beta}_{0}^{(N)})$ is the unit ball around $\dot{\beta}_{0}^{(N)}$ and $\mathcal{B}(\dot{\beta}_{1}^{(N)})$ is the unit ball around $\dot{\beta}_{1}^{(N)}$.  Since the $s/(3D)$-covering number of the unit ball in $\R^{m}$ under the $\ell^{2}$-norm is bounded above by $\left(1 + 6D/s\right)^{m} $ \citep[Example 2]{generalizedOB} it follows from \eqref{ineq: covering number initial bound} that for all $s \in (0,1)$
    \begin{equation}\label{ineq: after using Euclidean ball bounds}
        \mathscr{N}(\mathscr{F}_{N},||\cdot||_{N}, s) \leq \left(1 + 6D/s\right) \times \left(1 + 6D/s\right)^{2} \times \mathscr{N}\left(\mathscr{C}_{N},||\cdot||_{N}, s/(3D)\right).
    \end{equation}
    By the monotonicity of the logarithm \eqref{ineq: after using Euclidean ball bounds} implies that
    \begin{align}
        \int_{0}^{1}\sup_{N}\sqrt{\log  \mathscr{N}(\mathscr{F}_{N},||\cdot||_{N}, s)}\, ds &\leq \int_{0}^{1}\sup_{N}\sqrt{\log \left(\left(1 + 6D/s\right)^{3} \mathscr{N}\left(\mathscr{C}_{N},||\cdot||_{N}, s/(3D)\right) \right) }\, ds \nonumber\\
                 &= \int_{0}^{1}\sup_{N}\sqrt{3\log \left(1 + 6D/s\right) +  \log\left( \mathscr{N}\left(\mathscr{C}_{N},||\cdot||_{N}, s/(3D)\right) \right) }\, ds \nonumber
        \end{align}
        By the subadditivity of the square-root, the last line can be bounded above by
        \begin{align}
                 &\int_{0}^{1}\sup_{N}\left(\sqrt{3\log \left(1 + 6D/s\right)} + \sqrt{ \log\left( \mathscr{N}\left(\mathscr{C}_{N},||\cdot||_{N}, s/(3D)\right) \right) }\right)\, ds \nonumber\\
                &= \underbrace{\int_{0}^{1}\sqrt{3\log \left(1 + 6D/s\right)}\,ds}_{a} +   \underbrace{\int_{0}^{1}\sup_{N}\sqrt{ \log\left( \mathscr{N}\left(\mathscr{C}_{N},||\cdot||_{N}, s/(3D)\right) \right) }\, ds}_{b}.
        \end{align}
        The term $a$ is finite for all $D > 0$.  The term $b$ is finite by Lemmas~\ref{lem: radius is arbitrary} and \ref{lem: marginally simple realizations implies jointly simple}.  Thus, $\int_{0}^{1}\sup_{N}\sqrt{\log  \mathscr{N}(\mathscr{F}_{N},||\cdot||_{N}, s)}\, ds < \infty$ and so $\muhat_{OLS, 0}$ and $\muhat_{OLS, 1}$ have typically simple realizations.

\end{proof}

\begin{remark}\label{rem: simple realizations extend to feature engineering}
    For a deterministic sequence of functions $\left\{f^{(N)}\right\}_{N \in \N}$ which map from $\R^{k}$ to $\R^{\ell}$ define the function class
    \begin{multline}\label{eqn: definition of F_n with engineered features}
        \mathscr{F}_{N} = \Bigg\{\mu_{\beta_{0}, \beta_{1}}({x}) = \beta_{0} + \beta_{1}^{\T}\begin{bmatrix}
                        \mu_{0}({x})\\
                        \mu_{1}({x})\\
                        \end{bmatrix} + \beta_{2}^{\T}f^{(N)}(x) \; \Bigg|
                        \; |\beta_{0} - \dot{\beta}_{0}^{(N)}| \leq 1; \; ||\beta_{1} - \dot{\beta}_{1}^{(N)}||_{2} \leq 1;\\ 
                        ||\beta_{2} - \dot{\beta}_{2}^{(N)}||_{2} \leq 1; \;  \left| \left| \begin{bmatrix}
                        \mu_{0}\\
                        \mu_{1}\\
                        \end{bmatrix} - \begin{bmatrix}
                        \dot{\mu}_{0}\\
                        \dot{\mu}_{1}\\
                        \end{bmatrix} \right|\right|_{N} \leq 1;\; \begin{bmatrix}
                        \mu_{0}\\
                        \mu_{1}\\
                        \end{bmatrix}\in \mathscr{C}_{N}\Bigg\} 
    \end{multline}
    where $\dot{\beta}^{(N)}$ is derived from the population-level ordinary least squares linear regression including the engineered features $f^{(N)}(x_{i})$; i.e., for the case of calibration in the treated population replace \eqref{eqn: useful notation ``population`` beta} with
    \begin{equation*}
        \dot{\beta}^{(N)} = (\dot{\beta}_{0}^{(N)}, \dot{\beta}_{1}^{(N)}, \dot{\beta}_{2}^{(N)}) = \argmin_{\beta_{0}, \beta_{1}, \beta_{2}}\left[\sum_{i = 1}^{N} \left\{y_{i}(1) - \left(\beta_{0} + \beta_{1}^{T}\begin{bmatrix}\dot{\mu}_{0}(x_{i}) \\ \dot{\mu}_{1}(x_{i}) \end{bmatrix} + \beta_{2}^{\T}f^{(N)}(x_{i})\right) \right\}^{2}\right].
    \end{equation*}
    Suppose that the engineered feature vectors $f(x_{i})$ satisfy Assumptions~\ref{supp asm: means and covs stabilize} and \ref{supp asm: bounded fourth moment}, and the required linear regressions are not ill-defined.  Under Assumptions~\ref{supp asm: means and covs stabilize} and \ref{supp asm: bounded fourth moment} following the same line of reasoning used in the proof of Proposition~\ref{prop: typically simple realizations are inherited} yields that 
    \begin{multline*}
        \int_{0}^{1}\sup_{N}\sqrt{\log  \mathscr{N}(\mathscr{F}_{N},||\cdot||_{N}, s)}\, ds \leq \int_{0}^{1}\sqrt{(3 + \ell)\log \left(1 + 8D/s\right)}\,ds +  \\
        \int_{0}^{1}\sup_{N}\sqrt{ \log\left( \mathscr{N}\left(\mathscr{C}_{N},||\cdot||_{N}, s/(4D)\right) \right) }\, ds < \infty.
    \end{multline*}
    Consequently, the prediction functions undergirding $\tauhat_{cal2}$ have typically simple realizations.
\end{remark}

\subsection{Entropy analysis in finite population models}
The work of \citet{generalizedOB} proceeds under a finite population model and the entropy condition Assumption~\ref{supp asm: simple realizations}, so we direct the interested reader to their explication.  We highlight a two points that are relevant in the context of calibration; namely
\begin{itemize}
    \item Lemma~\ref{lem: marginally simple realizations implies jointly simple} establishes that the mapping from covariates to ``pseudo-covariates" $(\muhat_{0}(x_{i}), \muhat_{1}(x_{i}))$ automatically inherits typically simple realizations from the two original prediction functions $\muhat_{0}$ and $\muhat_{1}$;
    \item Proposition~\ref{prop: typically simple realizations are inherited} leverages Lemma~\ref{lem: marginally simple realizations implies jointly simple} to conclude that the prediction functions $\muhat_{OLS, 0}$ and $\muhat_{OLS, 1}$ inherit typically simple realizations from the two original prediction functions $\muhat_{0}$ and $\muhat_{1}$.
\end{itemize}

\subsection{Entropy analysis in superpopulation models}
In the superpopulation model of Section~\ref{sec: superpop}, the entropy condition Assumption~\ref{supp asm: simple realizations} requires mild modification to account for randomness in potential outcomes and covariates.

\begin{assumption}[Superpopulation Typically Simple Realizations]\label{supp asm: superpop simple realizations}
    There exists a sequence of sets of functions $\tsrC$, which may vary with $N$, such that the random function $\muhat_{0}$ falls into this class asymptotically almost surely.  Formally, $\Prob{\muhat_{0} \in \tsrC} \rightarrow 1$.  Furthermore, the sets of functions are ``small" in the sense that
    \begin{equation*}
        \int_{0}^{1}\Expectation{\sup_{N}\sqrt{\log  \mathscr{N}(\tsrC,||\cdot||_{N}, s)}}\, ds < \infty
    \end{equation*}
    where $\mathscr{N}(\tsrC,||\cdot||_{N}, s)$ is the $s$-covering number of $\tsrC$ under the random metric induced by $||\cdot||_{N}$, and the expectation is taken with respect to randomness in $\{x_{i}\}_{i = 1}^{N}$.
    An analogous statement holds for $\muhat_{1}$ with a sequence of sets $\tsrT$.
\end{assumption}

The following lemma reproduces Lemma~\ref{lem: asymptotic linear expansion around SATE} but directly uses the entropy condition Assumption~\ref{supp asm: superpop simple realizations} instead of assuming \textit{a priori} that the error process vanishes as Assumption~\ref{app asm: error process vanishes} does.  A byproduct of this this result is that Assumption~\ref{supp asm: superpop simple realizations} is a sufficient condition for Assumption~\ref{app asm: error process vanishes} in a superpopulation model.

\begin{lemma}[Superpopulation Linear Expansions Via Entropy Bounds]\label{lem: asymptotic linear expansion around SATE via entropy bounds}
    Under Assumptions~\ref{supp asm: non-degen sampling limit}, \ref{supp asm: stability}, and \ref{supp asm: superpop simple realizations} the random variable $N^{-1}\sum_{i = 1}^{N}\left(\muhat_{z}(x_{i}) - y_{i}(z)\right)$ is asymptotically linear in the sense that
    \begin{equation*}
        \frac{1}{N} \sum_{i = 1}^{N}\left(\muhat_{z}(x_{i}) - y_{i}(z)\right) = \frac{1}{n_{z}}\sum_{i \st Z_{i} = z}\dot{\epsilon}_{i}(z) + o_{p}(N^{-1/2})
    \end{equation*}
    where $\dot{\epsilon}_{i}(z) = y_{i}(z) - \overdotmu_{z}(x_i)$.
\end{lemma}
\begin{proof}
    By the exact same reasoning as in \citet[Proof of Theorem 3]{generalizedOB}, rewrite $N^{-1}\sum_{i = 1}^{N}\left(\muhat_{z}(x_{i}) - y_{i}(z)\right) $ as
    \begin{equation*}
        \frac{1}{n_{1}}\sum_{i \st Z_{i} = z}\dot{\epsilon}_{i}(z) + \frac{1}{N}\Bigg(\underbrace{\sum_{i = 1}^{N}\left(\frac{Z_{i}\overdotmu_{z}(x_{i})}{(n_{z} / N)} - \overdotmu_{z}(x_{i})\right)}_{N^{1/2}\mathbb{G}(\overdotmu_{z})} - \underbrace{\sum_{i = 1}^{N}\left(\frac{Z_{i}\muhat_{z}(x_{i})}{(n_{z} / N)} - \muhat_{z}(x_{i})\right)}_{N^{1/2}\mathbb{G}(\muhat_{z})} \Bigg).
    \end{equation*}
    Consequently, the desired result holds so long as we can show that $|\mathbb{G}(\overdotmu_{z}) - \mathbb{G}(\muhat_{z})| = o_{p}(1)$.  Unlike the proof in \cite{generalizedOB}, the processes $\mathbb{G}(\overdotmu_{z})$ and $\mathbb{G}(\muhat_{z})$ inherit randomness from more than just $Z$; randomness enters through $Z$, $\{y_{i}(z)\}_{i = 1}^{N}$, and $\{x_{i}\}_{i = 1}^{N}$.  For this, we use a Massart-like inequality \citep[pg. 73-109]{massartCitation} similar to that of Lemma S1 in \citet{lasso}.  In the original formulation of \citeauthor{lasso} randomness is only due to $Z$.  Consequently, for some fixed universal constant $\kappa > 0$, by the logic of \citet[Proof of Theorem 3]{generalizedOB}
    \begin{multline}\label{eqn: superpop conditional covering number bound}
        \Prob{\sup_{\substack{\mu \in \tsrz \\ ||\mu - \overdotmu_{z}||_{N} \leq r}} \left|\mathbb{G}(\overdotmu_{z}) - \mathbb{G}(\muhat_{z}) \right| > \varepsilon \,\Bigg|\, \{x_{i}, y_{i}(0), y_{i}(1)\}_{i = 1}^{N}} \leq\\ \kappa\frac{r}{\varepsilon} + \kappa\int_{0}^{r}\sup_{N}\left( \log\mathscr{N}(\tsrz,||\cdot||_{N}, s) \right)^{1/2}\,ds.
    \end{multline}
    The norm $||\cdot||_{N}$ on the right-hand-side of \eqref{eqn: superpop conditional covering number bound} is random since it depends upon the realization of  $\{x_{i}\}_{i = 1}^{N}$.  By the law of iterated expectation,
    \begin{multline}\label{eqn: superpop unconditional covering number bound}
        \Prob{\sup_{\substack{\mu \in \tsrz \\ ||\mu - \overdotmu_{z}||_{N} \leq r}} \left|\mathbb{G}(\overdotmu_{z}) - \mathbb{G}(\muhat_{z}) \right| > \varepsilon} \leq \kappa\frac{r}{\varepsilon} + \kappa\Expectation{\int_{0}^{r}\sup_{N}\left( \log\mathscr{N}(\tsrz,||\cdot||_{N}, s) \right)^{1/2}\,ds},
    \end{multline}
    where the expectation on the right is with respect to randomness in $\{x_{i}\}_{i = 1}^{N}$.\footnote{By the law of iterated expectation, the expectation on the right-hand-side of \eqref{eqn: superpop unconditional covering number bound} is with respect to randomness in $\{x_{i}, y_{i}(0), y_{i}(1)\}_{i = 1}^{N}$; however, the right-hand-side of the inequality has no dependence upon $\{y_{i}(0), y_{i}(1)\}_{i = 1}^{N}$ and so the expectation can be taken to be only over $\{x_{i}\}_{i = 1}^{N}$ without loss of generality.}  The Fubini-Tonelli theorem \citep[Theorem 1.7.2]{durrett} justifies exchanging the expectation with the integral, yielding the upper bound of
    \begin{equation*}
        \kappa\frac{r}{\varepsilon} + \kappa\int_{0}^{r}\Expectation{\sup_{N}\left( \log\mathscr{N}(\tsrz,||\cdot||_{N}, s) \right)^{1/2}\,ds}.
    \end{equation*}
    The remaining logic of  \citet[Proof of Theorem 3]{generalizedOB} combined with the superpopulation entropy condition of Assumption~\ref{supp asm: superpop simple realizations} yields that $|\mathbb{G}(\overdotmu_{z}) - \mathbb{G}(\muhat_{z})| = o_{p}(1)$; asymptotic linearity follows immediately.
\end{proof}

\begin{remark}
    The condition of Assumption~\ref{supp asm: superpop simple realizations} is implied by the uniform entropy bound of Equation (2.5.1) in \citet{VanDerVaartWellner}:
    \begin{equation*}
        \int_{0}^{\infty}\sup_{N}\sup_{Q}\left(\log \mathscr{N}(\tsrz,L_{2}(Q), s\left|\left|\tsrz\right|\right|_{Q, 2}) \right)^{1/2}\,ds < \infty
    \end{equation*}
    where the supremum over $Q$ ranges over all finitely discrete probability measures on $\R^{k}$.  The $\log$-covering number $\log\mathscr{N}(\tsrz,L_{2}(Q), s\left|\left|\tsrz\right|\right|_{Q, 2})$ vanishes when $s$ exceeds one, since $\tsrz$, which must be of finite diameter, can be covered by a single ball with diameter greater than $\left|\left|\tsrz\right|\right|_{Q, 2}$.  Thus, the integral in question can be rewritten as
    \begin{equation}\label{eqn: reformulated uniform entropy integral}
        \int_{0}^{1}\sup_{N}\sup_{Q}\left(\log \mathscr{N}(\tsrz,L_{2}(Q), s\left|\left|\tsrz\right|\right|_{Q, 2}) \right)^{1/2}\,ds.
    \end{equation}
    Uniform entropy conditions are tightly related to Vapnik-Chervonenkis dimension \citep[Theorem 2.6.7]{VanDerVaartWellner} for $s \in (0, 1)$
    \begin{equation*}
        \mathscr{N}(\tsrz,L_{2}(Q), s\left|\left|\tsrz\right|\right|_{Q, 2}) \leq \tilde{C}\mathcal{VC}(\tsrz)(16e)^{\mathcal{VC}(\tsrz)} \left(\frac{1}{s} \right)^{2(\mathcal{VC}(\tsrz) - 1)},
    \end{equation*}
    where $\tilde{C}$ is a constant and $\mathcal{VC}(\tsrz)$ denotes the Vapnik-Chervonenkis dimension of the subgraphs of functions in the function class $\tsrz$.  Consequently
    \begin{multline*}
        \log\mathscr{N}(\tsrz,L_{2}(Q), s\left|\left|\tsrz\right|\right|_{Q, 2}) \leq \log \tilde{C} + \log \mathcal{VC}(\tsrz) + \\ \mathcal{VC}(\tsrz)\log(16e) +  \log \left(\left(\frac{1}{s} \right)^{2(\mathcal{VC}(\tsrz) - 1)} \right).
    \end{multline*}
    The finiteness of \eqref{eqn: reformulated uniform entropy integral} is guaranteed as long as
    \begin{multline*}
        \int_{0}^{1}\sup_{N}\Bigg(\log \tilde{C} + \log \mathcal{VC}(\tsrz) + \\ \mathcal{VC}(\tsrz)\log(16e) +  \log \left(\left(\frac{1}{s} \right)^{2(\mathcal{VC}(\tsrz) - 1)} \right) \Bigg)^{1/2}\,ds < \infty.
    \end{multline*}
    Assume that $\mathcal{VC}(\tsrz)$ is bounded above by some constant independent of $N$.  By subadditivity of the square root, a sufficient condition for the required finiteness of the integral above is to require that
    \begin{equation*}
        \int_{0}^{1}\sup_{N}\left(-2(\mathcal{VC}(\tsrz) - 1)\log s\right)^{1/2}\,ds = \sup_{N}\left(\mathcal{VC}(\tsrz) - 1\right)^{1/2}\left(\frac{\pi}{2}\right)\,ds < \infty.
    \end{equation*}
    This holds so long as $\sup_{N}\mathcal{VC}(\tsrz) \geq 1$, which is true by definition.

    \begin{assumption}\label{supp asm: VC dimension}
        For $z \in \{0, 1\}$, the Vapnik-Chervonenkis dimension of the sub-graphs of the function class $\tsrz$ is bounded above by some finite constant which does not depend upon $N$.
    \end{assumption}

    By the discussion above, Assumption~\ref{supp asm: VC dimension} is a sufficient condition for Assumption~\ref{supp asm: superpop simple realizations} and consequently Vapnik-Chervonenkis conditions are sufficient to force the vanishing of the error process from Assumption~\ref{app asm: error process vanishes}.
\end{remark}

\subsection{Entropy analysis in fixed covariate models}

\begin{lemma}[Fixed Covariate Linear Expansions via Entropy Bounds]\label{lem: asymptotic linear expansion around SATE in fixed cov model via entropy}
    Under Assumptions~\ref{supp asm: non-degen sampling limit}, \ref{supp asm: simple realizations}, and \ref{supp asm: fixed cov conditional stability} the random variable $N^{-1}\sum_{i = 1}^{N}\left(\muhat_{z}(x_{i}) - y_{i}(z)\right)$ is conditionally almost surely asymptotically linear in the sense that, for $\dot{\epsilon}_{i}(z) = y_{i}(z) - \overdotmu_{z}(x_i)$
    \begin{equation*}
        \frac{1}{N} \sum_{i = 1}^{N}\left(\muhat_{z}(x_{i}) - y_{i}(z)\right) = \frac{1}{n_{z}}\sum_{i \st Z_{i} = z}\dot{\epsilon}_{i}(z) + o_{p}(N^{-1/2})
    \end{equation*}
    holds for almost all conditioning events of the form \eqref{eqn: fixed cov conditioning on outcomes}.
\end{lemma}
To prove Lemma~\ref{lem: asymptotic linear expansion around SATE in fixed cov model via entropy} condition on \eqref{eqn: fixed cov conditioning on outcomes} and then apply the proof of Theorem~3 from \citet{generalizedOB} to the conditional random functions $\muhat_{z}$ given $\{(y_{i}(0), y_{i}(1)) = (\mathbf{y}_{i}(0), \mathbf{y}_{i}(1))  \, i = 1, \ldots, N\}$.

\begin{theorem}\label{supp thm: fixed covariate clt for taucal using entropy}
    Under the fixed-covariate model, subject to Assumptions~\ref{supp asm: non-degen sampling limit}, \ref{supp asm: simple realizations}, \ref{supp asm: fixed cov means and covs stabilize}, \ref{supp asm: fixed cov bounded fourth moment}, and \ref{supp asm: fixed cov conditional stability}, $N^{1/2}\left(\tauhat_{cal} - \CATE\right)$ obeys a central limit theorem.
\end{theorem}
\begin{proof}
    We start out with the simple observation that
    \begin{equation*}
        N^{1/2}\left(\tauhat_{cal} - \CATE\right) = N^{1/2}\left(\tauhat_{cal} - \SATE\right) + N^{1/2}\left(\SATE - \CATE\right).
    \end{equation*}
    Thus, to show a central limit theorem for $N^{1/2}\left(\tauhat_{cal} - \CATE\right)$ it suffices to show that
    \begin{enumerate}
        \item \label{itm: two CLTs using entropy} Conditionally upon the potential outcomes, the term $N^{1/2}\left(\tauhat_{cal} - \SATE\right)$ converges weakly in probability to a fixed Gaussian distribution and term $N^{1/2}\left(\SATE - \CATE\right)$ obeys a central limit theorem.
        \item \label{itm: asymptotic independence using entropy} The terms $N^{1/2}\left(\tauhat_{cal} - \SATE\right)$ and $N^{1/2}\left(\SATE - \CATE\right)$ are asymptotically independent in the sense that their limiting joint distribution is the product of the two limiting marginal distributions.
    \end{enumerate}
    
    We tackle \ref{itm: two CLTs using entropy} first.  By Lemma~\ref{lem: Lindeberg condition} and the Lindeberg central limit theorem \citep[Theorem 11.2.5]{TestingStatHyp} it follows that $N^{1/2}\left(\SATE - \CATE\right)$ converges in distribution to a Gaussian distribution; denote this limiting distribution as $\Normal{0}{s_{m}}$.  
    
    Next, we show that $N^{1/2}\left(\tauhat_{cal} - \SATE\right)$ converges weakly in probability to a fixed Gaussian distribution.  We start with an algebraic rearrangement of $N^{1/2}\left(\tauhat_{cal} - \SATE\right)$:
    \begin{align*}
        N^{1/2}\left(\tauhat_{cal} - \SATE\right) &=  N^{1/2}\left(\frac{1}{N}\sum_{i = 1}^{N}\left(\muhat_{OLS,1}(x_{i}) - \muhat_{OLS,0}(x_{i}) \right) - \frac{1}{N}\sum_{i = 1}^{N}\left(y_{i}(1) - y_{i}(0) \right)\right) \\
        &= N^{1/2}\left(\frac{1}{N}\sum_{i = 1}^{N}\left(\muhat_{OLS,1}(x_{i}) - y_{i}(1) \right) - \frac{1}{N}\sum_{i = 1}^{N}\left(\muhat_{OLS,0}(x_{i}) - y_{i}(0) \right)\right).
    \end{align*}
    
    For each population of size $N$, condition upon some realization of the potential outcomes $\{(y_{i}(0), y_{i}(1)) = (\mathbf{y}_{i}(0), \mathbf{y}_{i}(1))  \, i = 1, \ldots, N\}$.  
    By Lemma~\ref{lem: asymptotic linear expansion around SATE in fixed cov model} for almost all such conditioning events, we have that 
    \begin{equation}\label{eqn: linear expansion around SATE in fixed covariate model}
        N^{1/2}\left(\frac{1}{N}\sum_{i = 1}^{N}Z_{i}Nn_{1}^{-1}\dot{\epsilon}_{i}(1) - \frac{1}{N}\sum_{i = 1}^{N}(1 - Z_{i})Nn_{0}^{-1}\dot{\epsilon}_{i}(0) \right) + o_{P}(1)
    \end{equation}
    where $\dot{\epsilon}_{i}(z) = y_{i}(z) - \overdotmu_{OLS,z}(x_{i})$ and the randomness in the $o_{P}(1)$ term is only with respect to randomness in $Z$ since we condition upon the covariates implicitly in the fixed-covariate model.  
    
    Under Assumption~\ref{supp asm: fixed cov conditional stability} and the assumption that $N^{-1}\sum_{i = 1}^{N}\left( \overdotmu_{z}(x_{i}) - y_{i}(z)\right)^{2} = o(N)$ as a numeric sequence for almost all conditioning events of the potential outcomes, by Lemma~3 in the appendix of \citet{generalizedOB} we can, without loss of generality, stipulate that $N^{-1}\sum_{i = 1}^{N}\dot{\epsilon}_{i}(z) = 0$ for $z \in \{0, 1\}$ almost surely with respect to the conditioning \eqref{eqn: fixed cov conditioning on outcomes}.

    Under Assumptions~\ref{supp asm: fixed cov means and covs stabilize} and \ref{supp asm: fixed cov bounded fourth moment} the finite population analysis provided in Theorem~\ref{supp thm: reg adj does no harm} shows that $N^{1/2}\left(\frac{1}{N}\sum_{i = 1}^{N}Z_{i}Nn_{1}^{-1}\dot{\epsilon}_{i}(1) - \frac{1}{N}\sum_{i = 1}^{N}(1 - Z_{i})Nn_{0}^{-1}\dot{\epsilon}_{i}(0) \right)$ converges weakly to a centered Gaussian distribution with variance given by the limit of $\sigma_{N}^{2}$ defined in \eqref{eqn: squared studenizing factor for cal}.  This limit exists by Assumption~\ref{supp asm: fixed cov means and covs stabilize} and is common to all conditioning events of the form \eqref{eqn: fixed cov conditioning on outcomes} up to a set of measure zero; we denote it by $s_{d}$.  Consequently, $N^{1/2}\left(\tauhat_{cal} - \SATE\right)$ converges weakly in probability to a $\Normal{0}{s_{d}}$.

    Finally, we turn to \ref{itm: asymptotic independence using entropy}.  By Theorem 5.1 (iii) of \citet{twoPhaseFramework} it follows that the random vector $\left(N^{1/2}\left(\tauhat_{cal} - \SATE\right) ,N^{1/2}\left(\SATE - \CATE\right) \right)$ converges in distribution to  $(\mathcal{C}, \mathcal{D}) \sim \Normal{0}{s_{d}} \tensor \Normal{0}{s_{m}}$.\footnote{The original work of \citet{twoPhaseFramework} focuses on survey-sampling; however, nothing of their result Theorem 5.1 (iii) relies upon the survey-sampling framework of having only a single potential outcome, so we apply their result to the causal inference context of multiple potential outcomes.}
    
    By the continuous mapping theorem \citep[Theorem 18.11]{asymptoticStats_vdv}, $N^{1/2}\left(\tauhat_{cal} - \SATE\right)  + N^{1/2}\left(\SATE - \CATE\right)$ converges in distribution to $\mathcal{C} + \mathcal{D}$.  Since the sum of independent Gaussian random variables is itself Gaussian we have that $N^{1/2}\left(\tauhat_{cal} - \SATE\right)  + N^{1/2}\left(\SATE - \CATE\right)$ converges in distribution to $\Normal{0}{s_{d} + s_{m}}$.  In turn, this implies that $N^{1/2}\left(\tauhat_{cal} - \CATE\right)$ converges in distribution to $\Normal{0}{s_{d} + s_{m}}$.
\end{proof}

\bibliographystyle{apalike}
\bibliography{bibliography}

\end{document}